%% file: abel-coquand-pagano-lmcs.tex
\documentclass{LMCS}

\makeatletter
\@ifundefined{lhs2tex.lhs2tex.sty.read}%
  {\@namedef{lhs2tex.lhs2tex.sty.read}{}%
   \newcommand\SkipToFmtEnd{}%
   \newcommand\EndFmtInput{}%
   \long\def\SkipToFmtEnd#1\EndFmtInput{}%
  }\SkipToFmtEnd

\newcommand\ReadOnlyOnce[1]{\@ifundefined{#1}{\@namedef{#1}{}}\SkipToFmtEnd}
\usepackage{amstext}
\usepackage{amssymb}
\usepackage{stmaryrd}
\DeclareFontFamily{OT1}{cmtex}{}
\DeclareFontShape{OT1}{cmtex}{m}{n}
  {<5><6><7><8>cmtex8
   <9>cmtex9
   <10><10.95><12><14.4><17.28><20.74><24.88>cmtex10}{}
\DeclareFontShape{OT1}{cmtex}{m}{it}
  {<-> ssub * cmtt/m/it}{}

\DeclareFontShape{OT1}{cmtt}{bx}{n}
  {<5><6><7><8>cmtt8
   <9>cmbtt9
   <10><10.95><12><14.4><17.28><20.74><24.88>cmbtt10}{}
\DeclareFontShape{OT1}{cmtex}{bx}{n}
  {<-> ssub * cmtt/bx/n}{}

\newcommand{\Conid}[1]{\mathit{#1}}
\newcommand{\Varid}[1]{\mathit{#1}}
\newcommand{\anonymous}{\kern0.06em \vbox{\hrule\@width.5em}}

\renewcommand{\leq}{\leqslant}
\renewcommand{\geq}{\geqslant}
\usepackage{polytable}

\@ifundefined{mathindent}%
  {\newdimen\mathindent\mathindent\leftmargini}%
  {}%

\def\resethooks{%
  \global\let\SaveRestoreHook\empty
  \global\let\ColumnHook\empty}
\newcommand*{\savecolumns}[1][default]%
  {\g@addto@macro\SaveRestoreHook{\savecolumns[#1]}}
\newcommand*{\restorecolumns}[1][default]%
  {\g@addto@macro\SaveRestoreHook{\restorecolumns[#1]}}
\newcommand*{\aligncolumn}[2]%
  {\g@addto@macro\ColumnHook{\column{#1}{#2}}}

\resethooks

\newcommand{\onelinecommentchars}{\quad-{}- }
\newcommand{\commentbeginchars}{\enskip\{-}
\newcommand{\commentendchars}{-\}\enskip}

\newcommand{\visiblecomments}{%
  \let\onelinecomment=\onelinecommentchars
  \let\commentbegin=\commentbeginchars
  \let\commentend=\commentendchars}

\newcommand{\invisiblecomments}{%
  \let\onelinecomment=\empty
  \let\commentbegin=\empty
  \let\commentend=\empty}

\visiblecomments

\newlength{\blanklineskip}
\newcommand{\hsindent}[1]{\quad}
\let\hspre\empty
\let\hspost\empty

\EndFmtInput
\makeatother
\ReadOnlyOnce{polycode.fmt}%
\makeatletter

\newcommand{\hsnewpar}[1]%
  {{\parskip=0pt\parindent=0pt\par\vskip #1\noindent}}

\newcommand{\hscodestyle}{}

\newcommand{\sethscode}[1]%
  {\expandafter\let\expandafter\hscode\csname #1\endcsname
   \expandafter\let\expandafter\endhscode\csname end#1\endcsname}

  {\par\noindent
   \advance\leftskip\mathindent
   \hscodestyle
   \let\\=\@normalcr
   \(\pboxed}%
  {\endpboxed\)%
   \par\noindent
   \ignorespacesafterend}

  {\hsnewpar\abovedisplayskip
   \advance\leftskip\mathindent
   \hscodestyle
   \let\hspre\(\let\hspost\)%
   \pboxed}%
  {\endpboxed%
   \hsnewpar\belowdisplayskip
   \ignorespacesafterend}

  {\hsnewpar\abovedisplayskip
   \advance\leftskip\mathindent
   \hscodestyle
   \let\\=\@normalcr
   \(\pboxed}%
  {\endpboxed\)%
   \hsnewpar\belowdisplayskip
   \ignorespacesafterend}

\newcommand{\plainhs}{\sethscode{plainhscode}}

\plainhs

  {\hsnewpar\abovedisplayskip
   \advance\leftskip\mathindent
   \hscodestyle
   \let\\=\@normalcr
   \(\parray}%
  {\endparray\)%
   \hsnewpar\belowdisplayskip
   \ignorespacesafterend}

  {\parray}{\endparray}

  {\(\parray}{\endparray\)}

\def\codeframewidth{\arrayrulewidth}
\RequirePackage{calc}

  {\parskip=\abovedisplayskip\par\noindent
   \hscodestyle
   \arrayrulewidth=\codeframewidth
   \tabular{@{}|p{\linewidth-2\arraycolsep-2\arrayrulewidth-2pt}|@{}}%
   \hline\framedhslinecorrect\\{-1.5ex}%
   \let\endoflinesave=\\
   \let\\=\@normalcr
   \(\pboxed}%
  {\endpboxed\)%
   \framedhslinecorrect\endoflinesave{.5ex}\hline
   \endtabular
   \parskip=\belowdisplayskip\par\noindent
   \ignorespacesafterend}

\newcommand{\framedhslinecorrect}[2]%
  {#1[#2]}

  {\(\def\column##1##2{}%
   \let\>\undefined\let\<\undefined\let\\\undefined
   \newcommand\>[1][]{}\newcommand\<[1][]{}\newcommand\\[1][]{}%
   \def\fromto##1##2##3{##3}%
   }{\) }%

  {\let\orighscode=\hscode
   \let\origendhscode=\endhscode
   \def\endhscode{\def\hscode{\endgroup\def\@currenvir{hscode}\\}\begingroup}
   \orighscode\def\hscode{\endgroup\def\@currenvir{hscode}}}%
  {\origendhscode
   \global\let\hscode=\orighscode
   \global\let\endhscode=\origendhscode}%

\makeatother
\EndFmtInput
\usepackage{enumerate}
\usepackage{hyperref}

\usepackage{amssymb} 
\usepackage{amsmath} 
\usepackage{mathpartir} 
\input{prooftree}

\theoremstyle{plain}
\def\eg{{\em e.g.}}
\def\cf{{\em cf.}}
\def\ie{{\em i.e.}}

\newcommand{\LONGVERSION}[1]{#1}
\newcommand{\SHORTVERSION}[1]{}
\input{commands}

\def\doi{7 (2:4) 2011}
\lmcsheading%
{\doi}
{1--57}
{}
{}
{Nov.~15, 2009}
{May~\phantom{.~0}4, 2011}
{}   

\begin{document}

\title[A Modular Type-Checking Algorithm for Type Theory]{A Modular
  Type-Checking Algorithm for \\ Type Theory with Singleton Types and
  \\ Proof Irrelevance}

\author[A.~Abel]{Andreas Abel\rsuper a}	
\address{{\lsuper a}Ludwig-Maximilians-Universit\"{a}t M\"{u}nchen}	
\email{\abelmail}  
\thanks{{\lsuper a}Supported by INRIA 
as guest researcher in the PI.R2 team, PPS,
Paris, France, from October 2009 to March 2010.}	

\author[T.~Coquand]{Thierry Coquand\rsuper b}	
\address{{\lsuper b}G\"{o}teborg University}	
\email{\coquandmail}  

\author[M.~Pagano]{Miguel Pagano\rsuper c} 
\address{{\lsuper c}Universidad Nacional de C\'{o}rdoba} 
\email{\paganomail} 
\thanks{{\lsuper c}Partially supported by CONICET, Argentina.} 







\begin{abstract}
  \noindent We define a logical framework with singleton types and one
  universe of small types. We give the semantics using a PER model; it
  is used for constructing a normalis\-ation-by-evaluation algorithm.
  We prove completeness and soundness of the algorithm; and get as a
  corollary the injectivity of type constructors. Then we give the
  definition of a correct and complete type-checking algorithm for
  terms in normal form. We extend the results to proof-irrelevant
  propositions.
\end{abstract}

\maketitle

\newtheorem{lemma}[thm]{Lemma}
\newtheorem{corollary}[thm]{Corollary}
\newtheorem{remark}[thm]{Remark}
\newtheorem{definition}[thm]{Definition}
\newtheorem{theorem}[thm]{Theorem}
\input{tlca09-intro}
\input{tlca09-calc}
\input{examples}
\input{nbeprimer}
\input{tlca09-model}
\input{tlca09-nbe}
\input{tlca09-tc}

\section{Conclusion}

\noindent The main contributions of the paper are the definition of a correct
and complete type-checking algorithm, and a simpler solution to the
problem of generating fresh identifiers in the NbE algorithm for a
calculus with singletons, one universe, and proof-irrelevant
types. The type-checker is based on the NbE algorithm which is used to
decide equality and to prove the injectivity of the type
constructors. We emphasise that the type-checking algorithm is modular
with respect to the normalisation algorithm. All the results can be
extended to a calculus with annotated lambda abstractions, yielding a
type-checking algorithm for terms not necessarily in normal forms.
The NbE algorithm can be implemented fairly easily in Haskell (\cf\
Appendix \ref{sec:nbealg}), but the correctness of the implementation
depends on proving the computational adequacy of the domain semantics
with respect to Haskell's operational semantics. We have not developed
this proof in this article and leave it for to future work.

\input{related}

\para{Acknowledgments} The authors have led many discussions with Peter
Dybjer on normalization by evaluation, generalized algebraic theories,
and presentation of type theory as categories with families.  The
authors thank the referees of a previous version of this article for
their helpful comments. The authors are also grateful for the
suggestions made by the anonymous referees on this long version.

\bibliographystyle{hplain}
\bibliography{biblio,auto-lmcs}

\appendix

\section{Normalisation by evaluation}
\label{sec:nbealg}

\begin{hscode}\SaveRestoreHook
\column{B}{@{}>{\hspre}l<{\hspost}@{}}%
\column{5}{@{}>{\hspre}l<{\hspost}@{}}%
\column{6}{@{}>{\hspre}l<{\hspost}@{}}%
\column{8}{@{}>{\hspre}l<{\hspost}@{}}%
\column{9}{@{}>{\hspre}l<{\hspost}@{}}%
\column{10}{@{}>{\hspre}l<{\hspost}@{}}%
\column{11}{@{}>{\hspre}l<{\hspost}@{}}%
\column{12}{@{}>{\hspre}l<{\hspost}@{}}%
\column{13}{@{}>{\hspre}l<{\hspost}@{}}%
\column{14}{@{}>{\hspre}l<{\hspost}@{}}%
\column{15}{@{}>{\hspre}l<{\hspost}@{}}%
\column{16}{@{}>{\hspre}l<{\hspost}@{}}%
\column{19}{@{}>{\hspre}l<{\hspost}@{}}%
\column{20}{@{}>{\hspre}l<{\hspost}@{}}%
\column{22}{@{}>{\hspre}l<{\hspost}@{}}%
\column{23}{@{}>{\hspre}l<{\hspost}@{}}%
\column{24}{@{}>{\hspre}l<{\hspost}@{}}%
\column{25}{@{}>{\hspre}l<{\hspost}@{}}%
\column{29}{@{}>{\hspre}c<{\hspost}@{}}%
\column{29E}{@{}l@{}}%
\column{30}{@{}>{\hspre}l<{\hspost}@{}}%
\column{31}{@{}>{\hspre}l<{\hspost}@{}}%
\column{32}{@{}>{\hspre}l<{\hspost}@{}}%
\column{35}{@{}>{\hspre}l<{\hspost}@{}}%
\column{36}{@{}>{\hspre}l<{\hspost}@{}}%
\column{37}{@{}>{\hspre}l<{\hspost}@{}}%
\column{42}{@{}>{\hspre}l<{\hspost}@{}}%
\column{43}{@{}>{\hspre}l<{\hspost}@{}}%
\column{44}{@{}>{\hspre}l<{\hspost}@{}}%
\column{45}{@{}>{\hspre}l<{\hspost}@{}}%
\column{47}{@{}>{\hspre}l<{\hspost}@{}}%
\column{53}{@{}>{\hspre}l<{\hspost}@{}}%
\column{62}{@{}>{\hspre}c<{\hspost}@{}}%
\column{62E}{@{}l@{}}%
\column{65}{@{}>{\hspre}l<{\hspost}@{}}%
\column{E}{@{}>{\hspre}l<{\hspost}@{}}%
\>[B]{}\mathbf{type}\;\Conid{Type}{}\<[13]%
\>[13]{}\mathrel{=}{}\<[16]%
\>[16]{}\Conid{Term}{}\<[E]%
\\
\>[B]{}\mathbf{data}\;\Conid{Term}{}\<[13]%
\>[13]{}\mathrel{=}{}\<[16]%
\>[16]{}\Conid{U}{}\<[44]%
\>[44]{}\mbox{\onelinecomment  universe}{}\<[E]%
\\
\>[13]{}\mid {}\<[16]%
\>[16]{}\Conid{Fun}\;\Conid{Type}\;\Conid{Type}{}\<[44]%
\>[44]{}\mbox{\onelinecomment  dependent function space}{}\<[E]%
\\
\>[13]{}\mid {}\<[16]%
\>[16]{}\Conid{Singl}\;\Conid{Term}\;\Conid{Type}{}\<[44]%
\>[44]{}\mbox{\onelinecomment  singleton type ($\{a\}_A$)}{}\<[E]%
\\
\>[13]{}\mid {}\<[16]%
\>[16]{}\Conid{App}\;\Conid{Term}\;\Conid{Term}{}\<[44]%
\>[44]{}\mbox{\onelinecomment  application}{}\<[E]%
\\
\>[13]{}\mid {}\<[16]%
\>[16]{}\Conid{Lam}\;\Conid{Term}{}\<[44]%
\>[44]{}\mbox{\onelinecomment  abstraction}{}\<[E]%
\\
\>[13]{}\mid {}\<[16]%
\>[16]{}\Conid{Q}{}\<[44]%
\>[44]{}\mbox{\onelinecomment  variable}{}\<[E]%
\\
\>[13]{}\mid {}\<[16]%
\>[16]{}\Conid{Sub}\;\Conid{Term}\;\Conid{Subst}{}\<[44]%
\>[44]{}\mbox{\onelinecomment  substitution}{}\<[E]%
\\
\>[13]{}\mid {}\<[16]%
\>[16]{}\Conid{Sigma}\;\Conid{Type}\;\Conid{Type}{}\<[44]%
\>[44]{}\mbox{\onelinecomment  dependent pair type}{}\<[E]%
\\
\>[13]{}\mid {}\<[16]%
\>[16]{}\Conid{Fst}\;\Conid{Term}{}\<[44]%
\>[44]{}\mbox{\onelinecomment  first projection}{}\<[E]%
\\
\>[13]{}\mid {}\<[16]%
\>[16]{}\Conid{Snd}\;\Conid{Term}{}\<[44]%
\>[44]{}\mbox{\onelinecomment  second projection}{}\<[E]%
\\
\>[13]{}\mid {}\<[16]%
\>[16]{}\Conid{Pair}\;\Conid{Term}\;\Conid{Term}{}\<[44]%
\>[44]{}\mbox{\onelinecomment  dependent pair }{}\<[E]%
\\
\>[13]{}\mid {}\<[16]%
\>[16]{}\Conid{Nat}{}\<[44]%
\>[44]{}\mbox{\onelinecomment  naturals}{}\<[E]%
\\
\>[13]{}\mid {}\<[16]%
\>[16]{}\Conid{Zero}{}\<[44]%
\>[44]{}\mbox{\onelinecomment  0}{}\<[E]%
\\
\>[13]{}\mid {}\<[16]%
\>[16]{}\Conid{Suc}\;\Conid{Term}{}\<[44]%
\>[44]{}\mbox{\onelinecomment  +1}{}\<[E]%
\\
\>[13]{}\mid {}\<[16]%
\>[16]{}\Conid{Natrec}\;\Conid{Type}\;\Conid{Term}\;\Conid{Term}\;\Conid{Term}{}\<[44]%
\>[44]{}\mbox{\onelinecomment  elimination for Nat}{}\<[E]%
\\
\>[13]{}\mid {}\<[16]%
\>[16]{}\Conid{Prf}\;\Conid{Type}{}\<[44]%
\>[44]{}\mbox{\onelinecomment  proof (with proof irrelevance)}{}\<[E]%
\\
\>[13]{}\mid {}\<[16]%
\>[16]{}\Conid{Box}\;\Conid{Term}{}\<[44]%
\>[44]{}\mbox{\onelinecomment  a term in \ensuremath{\Conid{Prf}\;\Conid{A}}}{}\<[E]%
\\
\>[13]{}\mid {}\<[16]%
\>[16]{}\Conid{Star}{}\<[44]%
\>[44]{}\mbox{\onelinecomment  canonical element of \ensuremath{\Conid{Prf}\;\Conid{A}}}{}\<[E]%
\\
\>[13]{}\mid {}\<[16]%
\>[16]{}\Conid{Where}\;\Conid{Type}\;\Conid{Term}\;\Conid{Term}{}\<[44]%
\>[44]{}\mbox{\onelinecomment  \ensuremath{\Conid{Box}} elimination}{}\<[E]%
\\
\>[13]{}\mid {}\<[16]%
\>[16]{}\Conid{Enum}\;\Conid{Int}{}\<[44]%
\>[44]{}\mbox{\onelinecomment  \ensuremath{\Conid{Enum}\;\Varid{n}} has n elements}{}\<[E]%
\\
\>[13]{}\mid {}\<[16]%
\>[16]{}\Conid{Const}\;\Conid{Int}\;\Conid{Int}{}\<[44]%
\>[44]{}\mbox{\onelinecomment  \ensuremath{\Conid{Const}\;\Varid{n}\;\Varid{i}} is the \ensuremath{\Varid{i}}th element}{}\<[E]%
\\
\>[13]{}\mid {}\<[16]%
\>[16]{}\Conid{Case}\;\Conid{Int}\;\Conid{Type}\;[\mskip1.5mu \Conid{Term}\mskip1.5mu]\;\Conid{Term}{}\<[44]%
\>[44]{}\mbox{\onelinecomment  elimination for \ensuremath{\Conid{Enum}\;\Varid{n}}}{}\<[E]%
\\
\>[16]{}\mathbf{deriving}\;(\Conid{Eq},\Conid{Show}){}\<[E]%
\\[\blanklineskip]%
\>[B]{}\mathbf{data}\;\Conid{Subst}{}\<[13]%
\>[13]{}\mathrel{=}{}\<[16]%
\>[16]{}\Conid{E}{}\<[44]%
\>[44]{}\mbox{\onelinecomment  empty substitution}{}\<[E]%
\\
\>[13]{}\mid {}\<[16]%
\>[16]{}\Conid{Is}{}\<[44]%
\>[44]{}\mbox{\onelinecomment  identity substitution}{}\<[E]%
\\
\>[13]{}\mid {}\<[16]%
\>[16]{}\Conid{Ext}\;\Conid{Subst}\;\Conid{Term}{}\<[44]%
\>[44]{}\mbox{\onelinecomment  extension}{}\<[E]%
\\
\>[13]{}\mid {}\<[16]%
\>[16]{}\Conid{P}{}\<[44]%
\>[44]{}\mbox{\onelinecomment  weakening}{}\<[E]%
\\
\>[13]{}\mid {}\<[16]%
\>[16]{}\Conid{Comp}\;\Conid{Subst}\;\Conid{Subst}{}\<[44]%
\>[44]{}\mbox{\onelinecomment  composition}{}\<[E]%
\\
\>[16]{}\mathbf{deriving}\;(\Conid{Eq},\Conid{Show}){}\<[E]%
\\[\blanklineskip]%
\>[B]{}\mathbf{type}\;\Conid{DT}{}\<[13]%
\>[13]{}\mathrel{=}{}\<[16]%
\>[16]{}\Conid{D}{}\<[44]%
\>[44]{}\mbox{\onelinecomment  semantic types}{}\<[E]%
\\
\>[B]{}\mathbf{data}\;\Conid{D}{}\<[13]%
\>[13]{}\mathrel{=}{}\<[16]%
\>[16]{}\Conid{T}{}\<[44]%
\>[44]{}\mbox{\onelinecomment  terminal object (empty context)}{}\<[E]%
\\
\>[13]{}\mid {}\<[16]%
\>[16]{}\Conid{Ld}\;(\Conid{D}\to \Conid{D}){}\<[44]%
\>[44]{}\mbox{\onelinecomment  function}{}\<[E]%
\\
\>[13]{}\mid {}\<[16]%
\>[16]{}\Conid{FunD}\;\Conid{DT}\;(\Conid{D}\to \Conid{DT}){}\<[44]%
\>[44]{}\mbox{\onelinecomment  dependent function type}{}\<[E]%
\\
\>[13]{}\mid {}\<[16]%
\>[16]{}\Conid{UD}{}\<[44]%
\>[44]{}\mbox{\onelinecomment  universe}{}\<[E]%
\\
\>[13]{}\mid {}\<[16]%
\>[16]{}\Conid{SingD}\;\Conid{D}\;\Conid{DT}{}\<[44]%
\>[44]{}\mbox{\onelinecomment  singleton type}{}\<[E]%
\\
\>[13]{}\mid {}\<[16]%
\>[16]{}\Conid{Vd}\;\Conid{Int}{}\<[44]%
\>[44]{}\mbox{\onelinecomment  free variable}{}\<[E]%
\\
\>[13]{}\mid {}\<[16]%
\>[16]{}\Conid{AppD}\;\Conid{D}\;\Conid{D}{}\<[44]%
\>[44]{}\mbox{\onelinecomment  neutral application}{}\<[E]%
\\
\>[13]{}\mid {}\<[16]%
\>[16]{}\Conid{SumD}\;\Conid{DT}\;(\Conid{D}\to \Conid{DT}){}\<[44]%
\>[44]{}\mbox{\onelinecomment  dependent pair type}{}\<[E]%
\\
\>[13]{}\mid {}\<[16]%
\>[16]{}\Conid{PairD}\;\Conid{D}\;\Conid{D}{}\<[44]%
\>[44]{}\mbox{\onelinecomment  context comprehension}{}\<[E]%
\\
\>[13]{}\mid {}\<[16]%
\>[16]{}\Conid{FstD}\;\Conid{D}{}\<[44]%
\>[44]{}\mbox{\onelinecomment  first projection of neutral}{}\<[E]%
\\
\>[13]{}\mid {}\<[16]%
\>[16]{}\Conid{SndD}\;\Conid{D}{}\<[44]%
\>[44]{}\mbox{\onelinecomment  second projection of neutral}{}\<[E]%
\\
\>[13]{}\mid {}\<[16]%
\>[16]{}\Conid{NatD}{}\<[44]%
\>[44]{}\mbox{\onelinecomment  natural number type}{}\<[E]%
\\
\>[13]{}\mid {}\<[16]%
\>[16]{}\Conid{ZeroD}{}\<[44]%
\>[44]{}\mbox{\onelinecomment  0}{}\<[E]%
\\
\>[13]{}\mid {}\<[16]%
\>[16]{}\Conid{SucD}\;\Conid{D}{}\<[44]%
\>[44]{}\mbox{\onelinecomment  +1}{}\<[E]%
\\
\>[13]{}\mid {}\<[16]%
\>[16]{}\Conid{NatrecD}\;(\Conid{D}\to \Conid{DT})\;\Conid{D}\;\Conid{D}\;\Conid{D}{}\<[44]%
\>[44]{}\mbox{\onelinecomment  recursion on neutrals}{}\<[E]%
\\
\>[13]{}\mid {}\<[16]%
\>[16]{}\Conid{PrfD}\;\Conid{DT}{}\<[44]%
\>[44]{}\mbox{\onelinecomment  proof type}{}\<[E]%
\\
\>[13]{}\mid {}\<[16]%
\>[16]{}\Conid{StarD}{}\<[44]%
\>[44]{}\mbox{\onelinecomment  don't care}{}\<[E]%
\\
\>[13]{}\mid {}\<[16]%
\>[16]{}\Conid{EnumD}\;\Conid{Int}{}\<[44]%
\>[44]{}\mbox{\onelinecomment  enumeration type}{}\<[E]%
\\
\>[13]{}\mid {}\<[16]%
\>[16]{}\Conid{ConstD}\;\Conid{Int}\;\Conid{Int}{}\<[44]%
\>[44]{}\mbox{\onelinecomment  constants in \ensuremath{\Conid{EnumD}}}{}\<[E]%
\\
\>[13]{}\mid {}\<[16]%
\>[16]{}\Conid{CaseD}\;\Conid{Int}\;(\Conid{D}\to \Conid{DT})\;[\mskip1.5mu \Conid{D}\mskip1.5mu]\;\Conid{D}{}\<[44]%
\>[44]{}\mbox{\onelinecomment  elimination on neutrals}{}\<[E]%
\\[\blanklineskip]%
\>[B]{}\mathbf{type}\;\Conid{Ctx}{}\<[13]%
\>[13]{}\mathrel{=}{}\<[16]%
\>[16]{}[\mskip1.5mu \Conid{Type}\mskip1.5mu]{}\<[E]%
\\[\blanklineskip]%
\>[B]{}\Varid{pi1},\Varid{pi2}\mathbin{::}{}\<[14]%
\>[14]{}\Conid{D}\to \Conid{D}{}\<[E]%
\\
\>[B]{}\Varid{pi1}\;{}\<[6]%
\>[6]{}(\Conid{PairD}\;\Varid{d}\;\Varid{d'}){}\<[20]%
\>[20]{}\mathrel{=}{}\<[23]%
\>[23]{}\Varid{d}{}\<[E]%
\\
\>[B]{}\Varid{pi1}\;{}\<[6]%
\>[6]{}\Conid{StarD}{}\<[20]%
\>[20]{}\mathrel{=}\Conid{StarD}{}\<[E]%
\\
\>[B]{}\Varid{pi1}\;{}\<[6]%
\>[6]{}\Varid{k}{}\<[20]%
\>[20]{}\mathrel{=}{}\<[23]%
\>[23]{}\Conid{FstD}\;\Varid{k}{}\<[E]%
\\
\>[B]{}\Varid{pi2}\;{}\<[6]%
\>[6]{}(\Conid{PairD}\;\Varid{d}\;\Varid{d'}){}\<[20]%
\>[20]{}\mathrel{=}{}\<[23]%
\>[23]{}\Varid{d'}{}\<[E]%
\\
\>[B]{}\Varid{pi2}\;{}\<[6]%
\>[6]{}\Conid{StarD}{}\<[20]%
\>[20]{}\mathrel{=}\Conid{StarD}{}\<[E]%
\\
\>[B]{}\Varid{pi2}\;{}\<[6]%
\>[6]{}\Varid{k}{}\<[20]%
\>[20]{}\mathrel{=}{}\<[23]%
\>[23]{}\Conid{SndD}\;\Varid{k}{}\<[E]%
\\[\blanklineskip]%
\>[B]{}\Varid{ap}\mathbin{::}\Conid{D}\to \Conid{D}\to \Conid{D}{}\<[E]%
\\
\>[B]{}\Varid{ap}\;(\Conid{Ld}\;\Varid{f})\;\Varid{d}\mathrel{=}\Varid{f}\;\Varid{d}{}\<[E]%
\\
\>[B]{}\Varid{ap}\;\Conid{StarD}\;{}\<[11]%
\>[11]{}\anonymous \mathrel{=}\Conid{StarD}{}\<[E]%
\\[\blanklineskip]%
\>[B]{}\Varid{neutralD}\mathbin{::}\Conid{D}\to \Conid{Bool}{}\<[E]%
\\
\>[B]{}\Varid{neutralD}\;(\Conid{Vd}\;\anonymous ){}\<[30]%
\>[30]{}\mathrel{=}\Conid{True}{}\<[E]%
\\
\>[B]{}\Varid{neutralD}\;(\Conid{AppD}\;\anonymous \;\anonymous ){}\<[30]%
\>[30]{}\mathrel{=}\Conid{True}{}\<[E]%
\\
\>[B]{}\Varid{neutralD}\;(\Conid{FstD}\;\anonymous ){}\<[30]%
\>[30]{}\mathrel{=}\Conid{True}{}\<[E]%
\\
\>[B]{}\Varid{neutralD}\;(\Conid{SndD}\;\anonymous ){}\<[30]%
\>[30]{}\mathrel{=}\Conid{True}{}\<[E]%
\\
\>[B]{}\Varid{neutralD}\;(\Conid{NatrecD}\;\anonymous \;\anonymous \;\anonymous \;\anonymous ){}\<[30]%
\>[30]{}\mathrel{=}\Conid{True}{}\<[E]%
\\
\>[B]{}\Varid{neutralD}\;(\Conid{CaseD}\;\anonymous \;\anonymous \;\anonymous \;\anonymous ){}\<[30]%
\>[30]{}\mathrel{=}\Conid{True}{}\<[E]%
\\
\>[B]{}\Varid{neutralD}\;\Conid{StarD}{}\<[30]%
\>[30]{}\mathrel{=}\Conid{True}{}\<[E]%
\\
\>[B]{}\Varid{neutralD}\;\anonymous {}\<[30]%
\>[30]{}\mathrel{=}\Conid{False}{}\<[E]%
\\[\blanklineskip]%
\>[B]{}\Varid{natrec}\mathbin{::}(\Conid{D}\to \Conid{DT})\to \Conid{D}\to \Conid{D}\to \Conid{D}\to \Conid{D}{}\<[E]%
\\
\>[B]{}\Varid{natrec}\;\Varid{b}\;\Varid{z}\;\Varid{s}\;\Conid{StarD}{}\<[31]%
\>[31]{}\mathrel{=}\Conid{StarD}{}\<[E]%
\\
\>[B]{}\Varid{natrec}\;\Varid{b}\;\Varid{z}\;\Varid{s}\;\Conid{ZeroD}{}\<[31]%
\>[31]{}\mathrel{=}\Varid{z}{}\<[E]%
\\
\>[B]{}\Varid{natrec}\;\Varid{b}\;\Varid{z}\;\Varid{s}\;(\Conid{SucD}\;\Varid{e}){}\<[31]%
\>[31]{}\mathrel{=}(\Varid{s}\mathbin{`\Varid{ap}`}\Varid{e})\mathbin{`\Varid{ap}`}(\Varid{natrec}\;\Varid{b}\;\Varid{z}\;\Varid{s}\;\Varid{e}){}\<[E]%
\\
\>[B]{}\Varid{natrec}\;\Varid{b}\;\Varid{z}\;\Varid{s}\;\Varid{d}\mid \Varid{neutralD}\;\Varid{d}{}\<[31]%
\>[31]{}\mathrel{=}\Varid{up}\;{}\<[37]%
\>[37]{}(\Varid{b}\;\Varid{d})\;{}\<[E]%
\\
\>[37]{}(\Conid{NatrecD}\;{}\<[47]%
\>[47]{}(\lambda \Varid{e}\to \Varid{downT}\;(\Varid{b}\;\Varid{e}))\;{}\<[E]%
\\
\>[47]{}(\Varid{down}\;(\Varid{b}\;\Conid{ZeroD})\;\Varid{z})\;{}\<[E]%
\\
\>[47]{}\Varid{downSuc}\;{}\<[E]%
\\
\>[47]{}\Varid{d}){}\<[E]%
\\
\>[B]{}\hsindent{13}{}\<[13]%
\>[13]{}\mathbf{where}\;\Varid{downSuc}\mathrel{=}\Varid{down}\;{}\<[35]%
\>[35]{}(\Conid{FunD}\;{}\<[42]%
\>[42]{}\Conid{NatD}\;{}\<[E]%
\\
\>[42]{}(\lambda \Varid{n}\to \Conid{FunD}\;(\Varid{b}\;\Varid{n}){}\<[E]%
\\
\>[42]{}(\lambda \Varid{e}\to \Varid{b}\;(\Conid{SucD}\;\Varid{n}))))\;{}\<[E]%
\\
\>[35]{}\Varid{s}{}\<[E]%
\\[\blanklineskip]%
\>[B]{}\Varid{downs}\mathbin{::}\Conid{Int}\to (\Conid{D}\to \Conid{DT})\to [\mskip1.5mu \Conid{D}\mskip1.5mu]\to \Conid{Int}\to [\mskip1.5mu \Conid{D}\mskip1.5mu]{}\<[E]%
\\
\>[B]{}\Varid{downs}\;{}\<[8]%
\>[8]{}\anonymous \;{}\<[11]%
\>[11]{}\anonymous \;{}\<[14]%
\>[14]{}[\mskip1.5mu \mskip1.5mu]\;{}\<[22]%
\>[22]{}\anonymous {}\<[25]%
\>[25]{}\mathrel{=}[\mskip1.5mu \mskip1.5mu]{}\<[E]%
\\
\>[B]{}\Varid{downs}\;{}\<[8]%
\>[8]{}\Varid{n}\;{}\<[11]%
\>[11]{}\Varid{f}\;{}\<[14]%
\>[14]{}(\Varid{d}\mathbin{:}\Varid{ds})\;{}\<[22]%
\>[22]{}\Varid{i}{}\<[25]%
\>[25]{}\mathrel{=}\Varid{down}\;(\Varid{f}\;(\Conid{ConstD}\;\Varid{n}\;\Varid{i}))\;\Varid{d}\mathbin{:}\Varid{downs}\;\Varid{n}\;\Varid{f}\;\Varid{ds}\;(\Varid{i}\mathbin{+}\mathrm{1}){}\<[E]%
\\[\blanklineskip]%
\>[B]{}\Varid{constD}{}\<[9]%
\>[9]{}\mathbin{::}\Conid{Int}\to \Conid{Int}\to \Conid{D}\to \Conid{Bool}{}\<[E]%
\\
\>[B]{}\Varid{constD}\;{}\<[9]%
\>[9]{}\Varid{n}\;{}\<[12]%
\>[12]{}\Varid{i}\;{}\<[15]%
\>[15]{}(\Conid{ConstD}\;\Varid{m}\;\Varid{j}){}\<[29]%
\>[29]{}\mathrel{=}{}\<[29E]%
\>[32]{}\Varid{m}\equiv \Varid{n}\mathrel{\wedge}\Varid{i}\equiv \Varid{j}{}\<[E]%
\\
\>[B]{}\Varid{constD}\;{}\<[9]%
\>[9]{}\anonymous \;{}\<[12]%
\>[12]{}\anonymous \;{}\<[15]%
\>[15]{}\anonymous {}\<[29]%
\>[29]{}\mathrel{=}{}\<[29E]%
\>[32]{}\Conid{False}{}\<[E]%
\\[\blanklineskip]%
\>[B]{}\Varid{caseD}\mathbin{::}\Conid{Int}\to (\Conid{D}\to \Conid{DT})\to [\mskip1.5mu \Conid{D}\mskip1.5mu]\to \Conid{D}\to \Conid{D}{}\<[E]%
\\
\>[B]{}\Varid{caseD}\;\Varid{n}\;\Varid{b}\;\Varid{ds}\;\Conid{StarD}{}\<[62]%
\>[62]{}\mathrel{=}{}\<[62E]%
\>[65]{}\Conid{StarD}{}\<[E]%
\\
\>[B]{}\Varid{caseD}\;\Varid{n}\;\Varid{b}\;\Varid{ds}\;(\Conid{ConstD}\;\Varid{m}\;\Varid{i})\mid \Varid{n}\equiv \Varid{m}\mathrel{\wedge}\Varid{i}\mathbin{<}\Varid{n}{}\<[62]%
\>[62]{}\mathrel{=}{}\<[62E]%
\>[65]{}\Varid{ds}\mathbin{!!}\Varid{i}{}\<[E]%
\\
\>[B]{}\Varid{caseD}\;\Varid{n}\;\Varid{b}\;\Varid{ds}\;\Varid{d}\mid {}\<[19]%
\>[19]{}\Varid{neutralD}\;\Varid{d}\mathrel{\wedge}{}\<[E]%
\\
\>[19]{}\Varid{and}\;[\mskip1.5mu \Varid{constD}\;\Varid{n}\;\Varid{i}\;(\Varid{ds}\mathbin{!!}\Varid{i})\mid \Varid{i}\leftarrow [\mskip1.5mu \mathrm{0}\mathinner{\ldotp\ldotp}\Varid{n}\mathbin{-}\mathrm{1}\mskip1.5mu]\mskip1.5mu]{}\<[62]%
\>[62]{}\mathrel{=}{}\<[62E]%
\>[65]{}\Varid{up}\;(\Varid{b}\;\Varid{d})\;\Varid{d}{}\<[E]%
\\
\>[B]{}\Varid{caseD}\;\Varid{n}\;\Varid{b}\;\Varid{ds}\;\Varid{d}\mid {}\<[19]%
\>[19]{}\Varid{neutralD}\;\Varid{d}{}\<[62]%
\>[62]{}\mathrel{=}{}\<[62E]%
\>[65]{}\Varid{up}\;(\Varid{b}\;\Varid{d}){}\<[E]%
\\
\>[19]{}\hsindent{16}{}\<[35]%
\>[35]{}(\Conid{CaseD}\;\Varid{n}\;{}\<[45]%
\>[45]{}(\lambda \Varid{e}\to \Varid{downT}\;(\Varid{b}\;\Varid{e}))\;{}\<[E]%
\\
\>[45]{}(\Varid{downs}\;\Varid{n}\;\Varid{b}\;\Varid{ds}\;\mathrm{0})\;{}\<[E]%
\\
\>[45]{}\Varid{d}){}\<[E]%
\\
\>[B]{}\Varid{up}\mathbin{::}\Conid{DT}\to \Conid{D}\to \Conid{D}{}\<[E]%
\\
\>[B]{}\Varid{up}\;(\Conid{SingD}\;\Varid{a}\;\Varid{x})\;{}\<[19]%
\>[19]{}\Varid{k}{}\<[22]%
\>[22]{}\mathrel{=}\Varid{a}{}\<[E]%
\\
\>[B]{}\Varid{up}\;(\Conid{FunD}\;\Varid{a}\;\Varid{f})\;{}\<[19]%
\>[19]{}\Varid{k}{}\<[22]%
\>[22]{}\mathrel{=}\Conid{Ld}\;(\lambda \Varid{d}\to \Varid{up}\;(\Varid{f}\;\Varid{d})\;(\Conid{AppD}\;\Varid{k}\;(\Varid{down}\;\Varid{a}\;\Varid{d}))){}\<[E]%
\\
\>[B]{}\Varid{up}\;(\Conid{SumD}\;\Varid{a}\;\Varid{f})\;{}\<[19]%
\>[19]{}\Varid{k}{}\<[22]%
\>[22]{}\mathrel{=}\Conid{PairD}\;{}\<[31]%
\>[31]{}(\Varid{up}\;\Varid{a}\;(\Conid{FstD}\;\Varid{k}))\;{}\<[E]%
\\
\>[31]{}(\Varid{up}\;(\Varid{f}\;(\Varid{up}\;\Varid{a}\;(\Conid{FstD}\;\Varid{k})))\;(\Conid{SndD}\;\Varid{k})){}\<[E]%
\\
\>[B]{}\Varid{up}\;(\Conid{PrfD}\;\Varid{a})\;{}\<[19]%
\>[19]{}\Varid{k}{}\<[22]%
\>[22]{}\mathrel{=}\Conid{StarD}{}\<[E]%
\\
\>[B]{}\Varid{up}\;(\Conid{EnumD}\;\mathrm{0})\;{}\<[19]%
\>[19]{}\Varid{k}{}\<[22]%
\>[22]{}\mathrel{=}\Conid{StarD}{}\<[E]%
\\
\>[B]{}\Varid{up}\;(\Conid{EnumD}\;\mathrm{1})\;{}\<[19]%
\>[19]{}\Varid{k}{}\<[22]%
\>[22]{}\mathrel{=}\Conid{ConstD}\;\mathrm{1}\;\mathrm{0}{}\<[E]%
\\
\>[B]{}\Varid{up}\;\Varid{d}\;{}\<[19]%
\>[19]{}\Varid{k}{}\<[22]%
\>[22]{}\mathrel{=}\Varid{k}{}\<[E]%
\\[\blanklineskip]%
\>[B]{}\Varid{down}\mathbin{::}\Conid{DT}\to \Conid{D}\to \Conid{D}{}\<[E]%
\\
\>[B]{}\Varid{down}\;\Conid{UD}\;{}\<[19]%
\>[19]{}\Varid{d}{}\<[22]%
\>[22]{}\mathrel{=}\Varid{downT}\;\Varid{d}{}\<[E]%
\\
\>[B]{}\Varid{down}\;(\Conid{SingD}\;\Varid{a}\;\Varid{x})\;{}\<[19]%
\>[19]{}\Varid{d}{}\<[22]%
\>[22]{}\mathrel{=}\Varid{down}\;\Varid{x}\;\Varid{a}{}\<[E]%
\\
\>[B]{}\Varid{down}\;(\Conid{FunD}\;\Varid{a}\;\Varid{f})\;{}\<[19]%
\>[19]{}\Varid{d}{}\<[22]%
\>[22]{}\mathrel{=}\Conid{Ld}\;(\lambda \Varid{e}\to \Varid{down}\;(\Varid{f}\;(\Varid{up}\;\Varid{a}\;\Varid{e}))\;(\Varid{d}\mathbin{`\Varid{ap}`}(\Varid{up}\;\Varid{a}\;\Varid{e}))){}\<[E]%
\\
\>[B]{}\Varid{down}\;(\Conid{SumD}\;\Varid{a}\;\Varid{b})\;{}\<[19]%
\>[19]{}\Varid{d}{}\<[22]%
\>[22]{}\mathrel{=}\Conid{PairD}\;(\Varid{down}\;\Varid{a}\;(\Varid{pi1}\;\Varid{a}))\;(\Varid{down}\;(\Varid{b}\;(\Varid{pi1}\;\Varid{d}))\;(\Varid{pi2}\;\Varid{d})){}\<[E]%
\\
\>[B]{}\Varid{down}\;(\Conid{PrfD}\;\Varid{a})\;{}\<[19]%
\>[19]{}\Varid{d}{}\<[22]%
\>[22]{}\mathrel{=}\Conid{StarD}{}\<[E]%
\\
\>[B]{}\Varid{down}\;(\Conid{EnumD}\;\mathrm{1})\;{}\<[19]%
\>[19]{}\Varid{d}{}\<[22]%
\>[22]{}\mathrel{=}\Conid{ConstD}\;\mathrm{1}\;\mathrm{0}{}\<[E]%
\\
\>[B]{}\Varid{down}\;\Varid{d}\;{}\<[19]%
\>[19]{}\Varid{e}{}\<[22]%
\>[22]{}\mathrel{=}\Varid{e}{}\<[E]%
\\[\blanklineskip]%
\>[B]{}\Varid{downT}\mathbin{::}\Conid{DT}\to \Conid{DT}{}\<[E]%
\\
\>[B]{}\Varid{downT}\;(\Conid{SingD}\;\Varid{a}\;\Varid{x}){}\<[22]%
\>[22]{}\mathrel{=}\Conid{SingD}\;{}\<[31]%
\>[31]{}(\Varid{down}\;\Varid{x}\;\Varid{a})\;{}\<[43]%
\>[43]{}(\Varid{downT}\;\Varid{x}){}\<[E]%
\\
\>[B]{}\Varid{downT}\;(\Conid{FunD}\;\Varid{a}\;\Varid{f}){}\<[22]%
\>[22]{}\mathrel{=}\Conid{FunD}\;{}\<[31]%
\>[31]{}(\Varid{downT}\;\Varid{a})\;{}\<[43]%
\>[43]{}(\lambda \Varid{d}\to \Varid{downT}\;(\Varid{f}\;(\Varid{up}\;\Varid{a}\;\Varid{d}))){}\<[E]%
\\
\>[B]{}\Varid{downT}\;(\Conid{SumD}\;\Varid{a}\;\Varid{b}){}\<[22]%
\>[22]{}\mathrel{=}\Conid{SumD}\;{}\<[31]%
\>[31]{}(\Varid{downT}\;\Varid{a})\;{}\<[43]%
\>[43]{}(\lambda \Varid{d}\to \Varid{downT}\;(\Varid{b}\;(\Varid{up}\;\Varid{a}\;\Varid{d}))){}\<[E]%
\\
\>[B]{}\Varid{downT}\;(\Conid{PrfD}\;\Varid{a}){}\<[22]%
\>[22]{}\mathrel{=}\Conid{PrfD}\;{}\<[31]%
\>[31]{}(\Varid{downT}\;\Varid{a}){}\<[E]%
\\
\>[B]{}\Varid{downT}\;\Varid{d}{}\<[22]%
\>[22]{}\mathrel{=}\Varid{d}{}\<[E]%
\\[\blanklineskip]%
\>[B]{}\Varid{readback}\mathbin{::}\Conid{Int}\to \Conid{D}\to \Conid{Term}{}\<[E]%
\\
\>[B]{}\Varid{readback}\;\Varid{i}\;\Conid{UD}{}\<[31]%
\>[31]{}\mathrel{=}\Conid{U}{}\<[E]%
\\
\>[B]{}\Varid{readback}\;\Varid{i}\;(\Conid{FunD}\;\Varid{a}\;\Varid{f}){}\<[31]%
\>[31]{}\mathrel{=}\Conid{Fun}\;{}\<[43]%
\>[43]{}(\Varid{readback}\;\Varid{i}\;\Varid{a})\;(\Varid{readback}\;(\Varid{i}\mathbin{+}\mathrm{1})\;(\Varid{f}\;(\Conid{Vd}\;\Varid{i}))){}\<[E]%
\\
\>[B]{}\Varid{readback}\;\Varid{i}\;(\Conid{SingD}\;\Varid{a}\;\Varid{x}){}\<[31]%
\>[31]{}\mathrel{=}\Conid{Singl}\;{}\<[43]%
\>[43]{}(\Varid{readback}\;\Varid{i}\;\Varid{a})\;(\Varid{readback}\;\Varid{i}\;\Varid{x}){}\<[E]%
\\
\>[B]{}\Varid{readback}\;\Varid{i}\;(\Conid{Ld}\;\Varid{f}){}\<[31]%
\>[31]{}\mathrel{=}\Conid{Lam}\;{}\<[43]%
\>[43]{}(\Varid{readback}\;(\Varid{i}\mathbin{+}\mathrm{1})\;(\Varid{f}\;(\Conid{Vd}\;\Varid{i}))){}\<[E]%
\\
\>[B]{}\Varid{readback}\;\Varid{i}\;(\Conid{Vd}\;\Varid{n}){}\<[31]%
\>[31]{}\mathrel{=}\Varid{mkvar}\;{}\<[43]%
\>[43]{}(\Varid{i}\mathbin{-}\Varid{n}\mathbin{-}\mathrm{1}){}\<[E]%
\\
\>[B]{}\Varid{readback}\;\Varid{i}\;(\Conid{AppD}\;\Varid{k}\;\Varid{d}){}\<[31]%
\>[31]{}\mathrel{=}\Conid{App}\;{}\<[43]%
\>[43]{}(\Varid{readback}\;\Varid{i}\;\Varid{k})\;(\Varid{readback}\;\Varid{i}\;\Varid{d}){}\<[E]%
\\
\>[B]{}\Varid{readback}\;\Varid{i}\;(\Conid{FstD}\;\Varid{d}){}\<[31]%
\>[31]{}\mathrel{=}\Conid{Fst}\;{}\<[43]%
\>[43]{}(\Varid{readback}\;\Varid{i}\;\Varid{d}){}\<[E]%
\\
\>[B]{}\Varid{readback}\;\Varid{i}\;(\Conid{SndD}\;\Varid{d}){}\<[31]%
\>[31]{}\mathrel{=}\Conid{Snd}\;{}\<[43]%
\>[43]{}(\Varid{readback}\;\Varid{i}\;\Varid{d}){}\<[E]%
\\
\>[B]{}\Varid{readback}\;\Varid{i}\;(\Conid{PairD}\;\Varid{d}\;\Varid{e}){}\<[31]%
\>[31]{}\mathrel{=}\Conid{Pair}\;{}\<[43]%
\>[43]{}(\Varid{readback}\;\Varid{i}\;\Varid{d})\;(\Varid{readback}\;\Varid{i}\;\Varid{e}){}\<[E]%
\\
\>[B]{}\Varid{readback}\;\Varid{i}\;(\Conid{SumD}\;\Varid{a}\;\Varid{b}){}\<[31]%
\>[31]{}\mathrel{=}\Conid{Sigma}\;{}\<[43]%
\>[43]{}(\Varid{readback}\;\Varid{i}\;\Varid{a})\;(\Varid{readback}\;(\Varid{i}\mathbin{+}\mathrm{1})\;(\Varid{b}\;(\Conid{Vd}\;\Varid{i}))){}\<[E]%
\\
\>[B]{}\Varid{readback}\;\Varid{i}\;\Conid{NatD}{}\<[31]%
\>[31]{}\mathrel{=}\Conid{Nat}{}\<[E]%
\\
\>[B]{}\Varid{readback}\;\Varid{i}\;\Conid{ZeroD}{}\<[31]%
\>[31]{}\mathrel{=}\Conid{Zero}{}\<[E]%
\\
\>[B]{}\Varid{readback}\;\Varid{i}\;(\Conid{SucD}\;\Varid{e}){}\<[31]%
\>[31]{}\mathrel{=}\Conid{Suc}\;{}\<[43]%
\>[43]{}(\Varid{readback}\;\Varid{i}\;\Varid{e}){}\<[E]%
\\
\>[B]{}\Varid{readback}\;\Varid{i}\;(\Conid{NatrecD}\;\Varid{b}\;\Varid{z}\;\Varid{s}\;\Varid{e}){}\<[31]%
\>[31]{}\mathrel{=}\Conid{Natrec}\;{}\<[43]%
\>[43]{}(\Conid{Fun}\;\Conid{Nat}\;{}\<[53]%
\>[53]{}(\Varid{readback}\;(\Varid{i}\mathbin{+}\mathrm{1})\;(\Varid{b}\;(\Conid{Vd}\;\Varid{i}))))\;{}\<[E]%
\\
\>[43]{}(\Varid{readback}\;\Varid{i}\;\Varid{z})\;{}\<[E]%
\\
\>[43]{}(\Varid{readback}\;\Varid{i}\;\Varid{s})\;{}\<[E]%
\\
\>[43]{}(\Varid{readback}\;\Varid{i}\;\Varid{e}){}\<[E]%
\\
\>[B]{}\Varid{readback}\;\Varid{i}\;(\Conid{PrfD}\;\Varid{d}){}\<[32]%
\>[32]{}\mathrel{=}\Conid{Prf}\;{}\<[43]%
\>[43]{}(\Varid{readback}\;\Varid{i}\;\Varid{d}){}\<[E]%
\\
\>[B]{}\Varid{readback}\;\Varid{i}\;\Conid{StarD}{}\<[32]%
\>[32]{}\mathrel{=}\Conid{Star}{}\<[E]%
\\
\>[B]{}\Varid{readback}\;\Varid{i}\;(\Conid{EnumD}\;\Varid{n}){}\<[32]%
\>[32]{}\mathrel{=}\Conid{Enum}\;\Varid{n}{}\<[E]%
\\
\>[B]{}\Varid{readback}\;\Varid{i}\;(\Conid{ConstD}\;\Varid{n}\;\Varid{j}){}\<[32]%
\>[32]{}\mathrel{=}\Conid{Const}\;\Varid{n}\;{}\<[43]%
\>[43]{}\Varid{j}{}\<[E]%
\\
\>[B]{}\Varid{readback}\;\Varid{i}\;(\Conid{CaseD}\;\Varid{n}\;\Varid{b}\;\Varid{ds}\;\Varid{d}){}\<[32]%
\>[32]{}\mathrel{=}\Conid{Case}\;\Varid{n}\;{}\<[43]%
\>[43]{}(\Varid{readback}\;(\Varid{i}\mathbin{+}\mathrm{1})\;(\Varid{b}\;(\Conid{Vd}\;\Varid{i})))\;{}\<[E]%
\\
\>[43]{}(\Varid{map}\;(\Varid{readback}\;\Varid{i})\;\Varid{ds})\;{}\<[E]%
\\
\>[43]{}(\Varid{readback}\;\Varid{i}\;\Varid{d}){}\<[E]%
\\[\blanklineskip]%
\>[B]{}\mbox{\onelinecomment  Evaluation}{}\<[E]%
\\[\blanklineskip]%
\>[B]{}\mathbf{type}\;\Conid{Env}\mathrel{=}\Conid{D}{}\<[E]%
\\[\blanklineskip]%
\>[B]{}\Varid{eval}\mathbin{::}\Conid{Term}\to \Conid{Env}\to \Conid{D}{}\<[E]%
\\
\>[B]{}\Varid{eval}\;\Conid{U}\;{}\<[24]%
\>[24]{}\Varid{d}\mathrel{=}\Conid{UD}{}\<[E]%
\\
\>[B]{}\Varid{eval}\;(\Conid{Fun}\;\Varid{t}\;\Varid{f})\;{}\<[24]%
\>[24]{}\Varid{d}\mathrel{=}\Conid{FunD}\;(\Varid{eval}\;\Varid{t}\;\Varid{d})\;(\lambda \Varid{d'}\to \Varid{eval}\;\Varid{f}\;(\Conid{PairD}\;\Varid{d}\;\Varid{d'})){}\<[E]%
\\
\>[B]{}\Varid{eval}\;(\Conid{Singl}\;\Varid{t}\;\Varid{a})\;{}\<[24]%
\>[24]{}\Varid{d}\mathrel{=}\Conid{SingD}\;(\Varid{eval}\;\Varid{t}\;\Varid{d})\;(\Varid{eval}\;\Varid{a}\;\Varid{d}){}\<[E]%
\\
\>[B]{}\Varid{eval}\;(\Conid{Lam}\;\Varid{t})\;{}\<[24]%
\>[24]{}\Varid{d}\mathrel{=}\Conid{Ld}\;(\lambda \Varid{d'}\to \Varid{eval}\;\Varid{t}\;(\Conid{PairD}\;\Varid{d}\;\Varid{d'})){}\<[E]%
\\
\>[B]{}\Varid{eval}\;(\Conid{App}\;\Varid{t}\;\Varid{r})\;{}\<[24]%
\>[24]{}\Varid{d}\mathrel{=}(\Varid{eval}\;\Varid{t}\;\Varid{d})\mathbin{`\Varid{ap}`}(\Varid{eval}\;\Varid{r}\;\Varid{d}){}\<[E]%
\\
\>[B]{}\Varid{eval}\;\Conid{Q}\;{}\<[24]%
\>[24]{}\Varid{d}\mathrel{=}\Varid{pi2}\;\Varid{d}{}\<[E]%
\\
\>[B]{}\Varid{eval}\;(\Conid{Sub}\;\Varid{t}\;\Varid{s})\;{}\<[24]%
\>[24]{}\Varid{d}\mathrel{=}\Varid{eval}\;\Varid{t}\;(\Varid{evalS}\;\Varid{s}\;\Varid{d}){}\<[E]%
\\[\blanklineskip]%
\>[B]{}\Varid{eval}\;(\Conid{Sigma}\;\Varid{t}\;\Varid{r})\;{}\<[24]%
\>[24]{}\Varid{d}\mathrel{=}\Conid{SumD}\;(\Varid{eval}\;\Varid{t}\;\Varid{d})\;(\lambda \Varid{e}\to \Varid{eval}\;\Varid{r}\;(\Conid{PairD}\;\Varid{d}\;\Varid{e})){}\<[E]%
\\
\>[B]{}\Varid{eval}\;(\Conid{Fst}\;\Varid{t})\;{}\<[24]%
\>[24]{}\Varid{d}\mathrel{=}\Varid{pi1}\;(\Varid{eval}\;\Varid{t}\;\Varid{d}){}\<[E]%
\\
\>[B]{}\Varid{eval}\;(\Conid{Snd}\;\Varid{t})\;{}\<[24]%
\>[24]{}\Varid{d}\mathrel{=}\Varid{pi2}\;(\Varid{eval}\;\Varid{t}\;\Varid{d}){}\<[E]%
\\
\>[B]{}\Varid{eval}\;(\Conid{Pair}\;\Varid{t}\;\Varid{r})\;{}\<[24]%
\>[24]{}\Varid{d}\mathrel{=}\Conid{PairD}\;(\Varid{eval}\;\Varid{t}\;\Varid{d})\;(\Varid{eval}\;\Varid{r}\;\Varid{d}){}\<[E]%
\\[\blanklineskip]%
\>[B]{}\Varid{eval}\;\Conid{Nat}\;{}\<[24]%
\>[24]{}\Varid{d}\mathrel{=}\Conid{NatD}{}\<[E]%
\\
\>[B]{}\Varid{eval}\;\Conid{Zero}\;{}\<[24]%
\>[24]{}\Varid{d}\mathrel{=}\Conid{ZeroD}{}\<[E]%
\\
\>[B]{}\Varid{eval}\;(\Conid{Suc}\;\Varid{t})\;{}\<[24]%
\>[24]{}\Varid{d}\mathrel{=}\Conid{SucD}\;(\Varid{eval}\;\Varid{t}\;\Varid{d}){}\<[E]%
\\
\>[B]{}\Varid{eval}\;(\Conid{Natrec}\;\Varid{b}\;\Varid{z}\;\Varid{s}\;\Varid{t})\;{}\<[24]%
\>[24]{}\Varid{d}\mathrel{=}\Varid{natrec}\;{}\<[36]%
\>[36]{}(\lambda \Varid{e}\to \Varid{eval}\;\Varid{b}\;(\Conid{PairD}\;\Varid{d}\;\Varid{e}))\;{}\<[E]%
\\
\>[36]{}(\Varid{eval}\;\Varid{z}\;\Varid{d})\;{}\<[E]%
\\
\>[36]{}(\Varid{eval}\;\Varid{s}\;\Varid{d})\;{}\<[E]%
\\
\>[36]{}(\Varid{eval}\;\Varid{t}\;\Varid{d}){}\<[E]%
\\[\blanklineskip]%
\>[B]{}\Varid{eval}\;(\Conid{Prf}\;\Varid{t})\;{}\<[24]%
\>[24]{}\Varid{d}\mathrel{=}\Conid{PrfD}\;(\Varid{eval}\;\Varid{t}\;\Varid{d}){}\<[E]%
\\
\>[B]{}\Varid{eval}\;(\Conid{Box}\;\Varid{t})\;{}\<[24]%
\>[24]{}\Varid{d}\mathrel{=}\Conid{StarD}{}\<[E]%
\\
\>[B]{}\Varid{eval}\;\Conid{Star}\;{}\<[24]%
\>[24]{}\Varid{d}\mathrel{=}\Conid{StarD}{}\<[E]%
\\
\>[B]{}\Varid{eval}\;(\Conid{Where}\;\Varid{t}\;\Varid{b}\;\Varid{p})\;{}\<[24]%
\>[24]{}\Varid{d}\mathrel{=}\Varid{eval}\;\Varid{b}\;(\Conid{PairD}\;\Varid{d}\;\Conid{StarD}){}\<[E]%
\\
\>[B]{}\Varid{eval}\;(\Conid{Enum}\;\Varid{n})\;{}\<[24]%
\>[24]{}\Varid{d}\mathrel{=}\Conid{EnumD}\;\Varid{n}{}\<[E]%
\\
\>[B]{}\Varid{eval}\;(\Conid{Const}\;\Varid{n}\;\Varid{i})\;{}\<[24]%
\>[24]{}\Varid{d}\mathrel{=}\Conid{ConstD}\;\Varid{n}\;\Varid{i}{}\<[E]%
\\
\>[B]{}\Varid{eval}\;(\Conid{Case}\;\Varid{n}\;\Varid{b}\;\Varid{ts}\;\Varid{t})\;{}\<[24]%
\>[24]{}\Varid{d}\mathrel{=}\Varid{caseD}\;\Varid{n}\;{}\<[37]%
\>[37]{}(\lambda \Varid{e}\to \Varid{eval}\;\Varid{b}\;(\Conid{PairD}\;\Varid{d}\;\Varid{e}))\;{}\<[E]%
\\
\>[37]{}(\Varid{map}\;((\Varid{flip}\;\Varid{eval})\;\Varid{d})\;\Varid{ts})\;{}\<[E]%
\\
\>[37]{}(\Varid{eval}\;\Varid{t}\;\Varid{d}){}\<[E]%
\\[\blanklineskip]%
\>[B]{}\Varid{evalS}\mathbin{::}\Conid{Subst}\to \Conid{Env}\to \Conid{Env}{}\<[E]%
\\
\>[B]{}\Varid{evalS}\;\Conid{E}\;{}\<[20]%
\>[20]{}\Varid{d}\mathrel{=}\Conid{T}{}\<[E]%
\\
\>[B]{}\Varid{evalS}\;\Conid{Is}\;{}\<[20]%
\>[20]{}\Varid{d}\mathrel{=}\Varid{d}{}\<[E]%
\\
\>[B]{}\Varid{evalS}\;(\Conid{Ext}\;\Varid{s}\;\Varid{t})\;{}\<[20]%
\>[20]{}\Varid{d}\mathrel{=}\Conid{PairD}\;(\Varid{evalS}\;\Varid{s}\;\Varid{d})\;(\Varid{eval}\;\Varid{t}\;\Varid{d}){}\<[E]%
\\
\>[B]{}\Varid{evalS}\;\Conid{P}\;{}\<[20]%
\>[20]{}\Varid{d}\mathrel{=}\Varid{pi1}\;\Varid{d}{}\<[E]%
\\
\>[B]{}\Varid{evalS}\;(\Conid{Comp}\;\Varid{s}\;\Varid{s'})\;{}\<[20]%
\>[20]{}\Varid{d}\mathrel{=}(\Varid{evalS}\;\Varid{s}\mathbin{\circ}\Varid{evalS}\;\Varid{s'})\;\Varid{d}{}\<[E]%
\\[\blanklineskip]%
\>[B]{}\Varid{nbe}\mathbin{::}\Conid{Type}\to \Conid{Term}\to \Conid{Term}{}\<[E]%
\\
\>[B]{}\Varid{nbe}\;\Varid{ty}\;\Varid{t}\mathrel{=}\Varid{readback}\;\mathrm{0}\;(\Varid{down}\;(\Varid{eval}\;\Varid{ty}\;\Conid{T})\;(\Varid{eval}\;\Varid{t}\;\Conid{T})){}\<[E]%
\\[\blanklineskip]%
\>[B]{}\Varid{nbeTy}\mathbin{::}\Conid{Type}\to \Conid{Type}{}\<[E]%
\\
\>[B]{}\Varid{nbeTy}\;\Varid{ty}\mathrel{=}\Varid{readback}\;\mathrm{0}\;(\Varid{downT}\;(\Varid{eval}\;\Varid{ty}\;\Conid{T})){}\<[E]%
\\[\blanklineskip]%
\>[B]{}\Varid{nbeOpen}\mathbin{::}\Conid{Ctx}\to \Conid{Type}\to \Conid{Term}\to \Conid{Term}{}\<[E]%
\\
\>[B]{}\Varid{nbeOpen}\;\Varid{ctx}\;\Varid{ty}\;\Varid{t}{}\<[19]%
\>[19]{}\mathrel{=}\Varid{readback}\;\Varid{n}\;(\Varid{down}\;(\Varid{eval}\;\Varid{ty}\;\Varid{env})\;(\Varid{eval}\;\Varid{t}\;\Varid{env})){}\<[E]%
\\
\>[B]{}\hsindent{5}{}\<[5]%
\>[5]{}\mathbf{where}\;{}\<[12]%
\>[12]{}\Varid{n}{}\<[19]%
\>[19]{}\mathrel{=}\Varid{length}\;\Varid{ctx}{}\<[E]%
\\
\>[12]{}\Varid{env}{}\<[19]%
\>[19]{}\mathrel{=}\Varid{mkenv}\;\Varid{n}\;\Varid{ctx}{}\<[E]%
\\[\blanklineskip]%
\>[B]{}\Varid{nbeOpenTy}\mathbin{::}\Conid{Ctx}\to \Conid{Type}\to \Conid{Type}{}\<[E]%
\\
\>[B]{}\Varid{nbeOpenTy}\;\Varid{ctx}\;\Varid{ty}{}\<[19]%
\>[19]{}\mathrel{=}\Varid{readback}\;\Varid{n}\;(\Varid{downT}\;(\Varid{eval}\;\Varid{ty}\;\Varid{env})){}\<[E]%
\\
\>[B]{}\hsindent{5}{}\<[5]%
\>[5]{}\mathbf{where}\;{}\<[12]%
\>[12]{}\Varid{n}{}\<[19]%
\>[19]{}\mathrel{=}\Varid{length}\;\Varid{ctx}{}\<[E]%
\\
\>[12]{}\Varid{env}{}\<[19]%
\>[19]{}\mathrel{=}\Varid{mkenv}\;\Varid{n}\;\Varid{ctx}{}\<[E]%
\\[\blanklineskip]%
\>[B]{}\Varid{mkenv}\mathbin{::}\Conid{Int}\to \Conid{Ctx}\to \Conid{Env}{}\<[E]%
\\
\>[B]{}\Varid{mkenv}\;\mathrm{0}\;[\mskip1.5mu \mskip1.5mu]{}\<[19]%
\>[19]{}\mathrel{=}\Conid{T}{}\<[E]%
\\
\>[B]{}\Varid{mkenv}\;\Varid{n}\;(\Varid{t}\mathbin{:}\Varid{ts}){}\<[19]%
\>[19]{}\mathrel{=}\Conid{PairD}\;\Varid{d'}\;(\Varid{up}\;\Varid{td}\;(\Conid{Vd}\;(\Varid{n}\mathbin{-}\mathrm{1}))){}\<[E]%
\\
\>[B]{}\hsindent{5}{}\<[5]%
\>[5]{}\mathbf{where}\;{}\<[12]%
\>[12]{}\Varid{d'}{}\<[19]%
\>[19]{}\mathrel{=}\Varid{mkenv}\;(\Varid{n}\mathbin{-}\mathrm{1})\;\Varid{ts}{}\<[E]%
\\
\>[12]{}\Varid{td}{}\<[19]%
\>[19]{}\mathrel{=}\Varid{eval}\;\Varid{t}\;\Varid{d'}{}\<[E]%
\\[\blanklineskip]%
\>[B]{}\Varid{mkvar}\mathbin{::}\Conid{Int}\to \Conid{Term}{}\<[E]%
\\
\>[B]{}\Varid{mkvar}\;\Varid{n}{}\<[10]%
\>[10]{}\mid \Varid{n}\equiv \mathrm{0}{}\<[23]%
\>[23]{}\mathrel{=}\Conid{Q}{}\<[E]%
\\
\>[10]{}\mid \Varid{otherwise}{}\<[23]%
\>[23]{}\mathrel{=}\Conid{Sub}\;\Conid{Q}\;(\Varid{subs}\;(\Varid{n}\mathbin{-}\mathrm{1})){}\<[E]%
\\[\blanklineskip]%
\>[B]{}\Varid{subs}\;\Varid{n}{}\<[9]%
\>[9]{}\mid \Varid{n}\equiv \mathrm{0}{}\<[23]%
\>[23]{}\mathrel{=}\Conid{P}{}\<[E]%
\\
\>[B]{}\Varid{subs}\;\Varid{n}{}\<[9]%
\>[9]{}\mid \Varid{otherwise}{}\<[23]%
\>[23]{}\mathrel{=}\Conid{Comp}\;\Conid{P}\;(\Varid{subs}\;(\Varid{n}\mathbin{-}\mathrm{1})){}\<[E]%
\ColumnHook
\end{hscode}\resethooks

\section{Type-checking algorithm}
\label{sec:alg}

Type checking algorithm for normal forms, and type inference algorithm
for neutral terms.

\subsection*{Checking well-formedness of types}\bla\par

\begin{hscode}\SaveRestoreHook
\column{B}{@{}>{\hspre}l<{\hspost}@{}}%
\column{13}{@{}>{\hspre}l<{\hspost}@{}}%
\column{35}{@{}>{\hspre}l<{\hspost}@{}}%
\column{E}{@{}>{\hspre}l<{\hspost}@{}}%
\>[B]{}\Varid{chkType}\mathbin{::}\Conid{Ctx}\to \Conid{Type}\to \Conid{Bool}{}\<[E]%
\\
\>[B]{}\Varid{chkType}\;\Varid{ts}\;{}\<[13]%
\>[13]{}\Conid{U}{}\<[35]%
\>[35]{}\mathrel{=}\Conid{True}{}\<[E]%
\\
\>[B]{}\Varid{chkType}\;\Varid{ts}\;{}\<[13]%
\>[13]{}(\Conid{Fun}\;\Varid{t}\;\Varid{r}){}\<[35]%
\>[35]{}\mathrel{=}\Varid{chkType}\;\Varid{ts}\;\Varid{t}\mathrel{\wedge}\Varid{chkType}\;(\Varid{t}\mathbin{:}\Varid{ts})\;\Varid{r}{}\<[E]%
\\
\>[B]{}\Varid{chkType}\;\Varid{ts}\;{}\<[13]%
\>[13]{}(\Conid{Singl}\;\Varid{a}\;\Varid{t}){}\<[35]%
\>[35]{}\mathrel{=}\Varid{chkType}\;\Varid{ts}\;\Varid{t}\mathrel{\wedge}\Varid{chkTerm}\;\Varid{ts}\;\Varid{t}\;\Varid{a}{}\<[E]%
\\
\>[B]{}\Varid{chkType}\;\Varid{ts}\;{}\<[13]%
\>[13]{}(\Conid{Sigma}\;\Varid{t}\;\Varid{r}){}\<[35]%
\>[35]{}\mathrel{=}\Varid{chkType}\;\Varid{ts}\;\Varid{t}\mathrel{\wedge}\Varid{chkType}\;(\Varid{t}\mathbin{:}\Varid{ts})\;\Varid{r}{}\<[E]%
\\
\>[B]{}\Varid{chkType}\;\Varid{ts}\;{}\<[13]%
\>[13]{}\Conid{Nat}{}\<[35]%
\>[35]{}\mathrel{=}\Conid{True}{}\<[E]%
\\
\>[B]{}\Varid{chkType}\;\Varid{ts}\;{}\<[13]%
\>[13]{}(\Conid{Prf}\;\Varid{t}){}\<[35]%
\>[35]{}\mathrel{=}\Varid{chkType}\;\Varid{ts}\;\Varid{t}{}\<[E]%
\\
\>[B]{}\Varid{chkType}\;\Varid{ts}\;{}\<[13]%
\>[13]{}(\Conid{Enum}\;\Varid{n}){}\<[35]%
\>[35]{}\mathrel{=}\Conid{True}{}\<[E]%
\\
\>[B]{}\Varid{chkType}\;\Varid{ts}\;{}\<[13]%
\>[13]{}\Conid{Q}{}\<[35]%
\>[35]{}\mathrel{=}\Varid{chkNeTerm}\;\Varid{ts}\;\Conid{U}\;\Conid{Q}{}\<[E]%
\\
\>[B]{}\Varid{chkType}\;\Varid{ts}\;{}\<[13]%
\>[13]{}\Varid{w}\mathord{@}(\Conid{Sub}\;\Conid{Q}\;\Varid{s}){}\<[35]%
\>[35]{}\mathrel{=}\Varid{chkNeTerm}\;\Varid{ts}\;\Conid{U}\;\Varid{w}{}\<[E]%
\\
\>[B]{}\Varid{chkType}\;\Varid{ts}\;{}\<[13]%
\>[13]{}\Varid{w}\mathord{@}(\Conid{App}\;\Varid{k}\;\Varid{v}){}\<[35]%
\>[35]{}\mathrel{=}\Varid{chkNeTerm}\;\Varid{ts}\;\Conid{U}\;\Varid{w}{}\<[E]%
\\
\>[B]{}\Varid{chkType}\;\Varid{ts}\;{}\<[13]%
\>[13]{}\Varid{w}\mathord{@}(\Conid{Fst}\;\Varid{k}){}\<[35]%
\>[35]{}\mathrel{=}\Varid{chkNeTerm}\;\Varid{ts}\;\Conid{U}\;\Varid{w}{}\<[E]%
\\
\>[B]{}\Varid{chkType}\;\Varid{ts}\;{}\<[13]%
\>[13]{}\Varid{w}\mathord{@}(\Conid{Snd}\;\Varid{k}){}\<[35]%
\>[35]{}\mathrel{=}\Varid{chkNeTerm}\;\Varid{ts}\;\Conid{U}\;\Varid{w}{}\<[E]%
\\
\>[B]{}\Varid{chkType}\;\Varid{ts}\;{}\<[13]%
\>[13]{}\Varid{w}\mathord{@}(\Conid{Natrec}\;\Varid{t'}\;\Varid{v}\;\Varid{v'}\;\Varid{k}){}\<[35]%
\>[35]{}\mathrel{=}\Varid{chkNeTerm}\;\Varid{ts}\;\Conid{U}\;\Varid{w}{}\<[E]%
\\
\>[B]{}\Varid{chkType}\;\anonymous \;{}\<[13]%
\>[13]{}\anonymous {}\<[35]%
\>[35]{}\mathrel{=}\Conid{False}{}\<[E]%
\ColumnHook
\end{hscode}\resethooks

\subsection*{Checking the types of terms}\bla\par

\begin{hscode}\SaveRestoreHook
\column{B}{@{}>{\hspre}l<{\hspost}@{}}%
\column{13}{@{}>{\hspre}l<{\hspost}@{}}%
\column{20}{@{}>{\hspre}l<{\hspost}@{}}%
\column{21}{@{}>{\hspre}l<{\hspost}@{}}%
\column{26}{@{}>{\hspre}l<{\hspost}@{}}%
\column{27}{@{}>{\hspre}l<{\hspost}@{}}%
\column{29}{@{}>{\hspre}l<{\hspost}@{}}%
\column{30}{@{}>{\hspre}l<{\hspost}@{}}%
\column{41}{@{}>{\hspre}c<{\hspost}@{}}%
\column{41E}{@{}l@{}}%
\column{44}{@{}>{\hspre}l<{\hspost}@{}}%
\column{E}{@{}>{\hspre}l<{\hspost}@{}}%
\>[B]{}\Varid{sgSub}\mathbin{::}\Conid{Term}\to \Conid{Term}\to \Conid{Term}{}\<[E]%
\\
\>[B]{}\Varid{sgSub}\;\Varid{t}\;\Varid{t'}\mathrel{=}\Conid{Sub}\;\Varid{t}\;(\Conid{Ext}\;\Conid{Is}\;\Varid{t'}){}\<[E]%
\\[\blanklineskip]%
\>[B]{}\Varid{chkTerm}\mathbin{::}\Conid{Ctx}\to \Conid{Type}\to \Conid{Term}\to \Conid{Bool}{}\<[E]%
\\
\>[B]{}\Varid{chkTerm}\;\Varid{ts}\;{}\<[13]%
\>[13]{}\Conid{U}\;{}\<[26]%
\>[26]{}(\Conid{Fun}\;\Varid{t}\;\Varid{t'}){}\<[41]%
\>[41]{}\mathrel{=}{}\<[41E]%
\>[44]{}\Varid{chkTerm}\;\Varid{ts}\;\Conid{U}\;\Varid{t}\mathrel{\wedge}{}\<[E]%
\\
\>[44]{}\Varid{chkTerm}\;(\Varid{t}\mathbin{:}\Varid{ts})\;\Conid{U}\;\Varid{t'}{}\<[E]%
\\
\>[B]{}\Varid{chkTerm}\;\Varid{ts}\;{}\<[13]%
\>[13]{}\Conid{U}\;{}\<[26]%
\>[26]{}(\Conid{Singl}\;\Varid{e}\;\Varid{t}){}\<[41]%
\>[41]{}\mathrel{=}{}\<[41E]%
\>[44]{}\Varid{chkTerm}\;\Varid{ts}\;\Conid{U}\;\Varid{t}\mathrel{\wedge}{}\<[E]%
\\
\>[44]{}\Varid{chkTerm}\;\Varid{ts}\;\Varid{t}\;\Varid{e}{}\<[E]%
\\
\>[B]{}\Varid{chkTerm}\;\Varid{ts}\;{}\<[13]%
\>[13]{}\Conid{U}\;{}\<[26]%
\>[26]{}(\Conid{Sigma}\;\Varid{t}\;\Varid{t'}){}\<[41]%
\>[41]{}\mathrel{=}{}\<[41E]%
\>[44]{}\Varid{chkTerm}\;\Varid{ts}\;\Conid{U}\;\Varid{t}\mathrel{\wedge}{}\<[E]%
\\
\>[44]{}\Varid{chkTerm}\;(\Varid{t}\mathbin{:}\Varid{ts})\;\Conid{U}\;\Varid{t'}{}\<[E]%
\\
\>[B]{}\Varid{chkTerm}\;\Varid{ts}\;{}\<[13]%
\>[13]{}\Conid{U}\;{}\<[27]%
\>[27]{}\Conid{Nat}{}\<[41]%
\>[41]{}\mathrel{=}{}\<[41E]%
\>[44]{}\Conid{True}{}\<[E]%
\\
\>[B]{}\Varid{chkTerm}\;\Varid{ts}\;{}\<[13]%
\>[13]{}(\Conid{Fun}\;\Varid{t}\;\Varid{t'})\;{}\<[26]%
\>[26]{}(\Conid{Lam}\;\Varid{e}){}\<[41]%
\>[41]{}\mathrel{=}{}\<[41E]%
\>[44]{}\Varid{chkTerm}\;(\Varid{t}\mathbin{:}\Varid{ts})\;\Varid{t'}\;\Varid{e}{}\<[E]%
\\
\>[B]{}\Varid{chkTerm}\;\Varid{ts}\;{}\<[13]%
\>[13]{}(\Conid{Singl}\;\Varid{e}\;\Varid{t})\;{}\<[26]%
\>[26]{}\Varid{e'}{}\<[41]%
\>[41]{}\mathrel{=}{}\<[41E]%
\>[44]{}\Varid{chkTerm}\;\Varid{ts}\;(\Varid{nbeOpenTy}\;\Varid{ts}\;\Varid{t})\;\Varid{e'}\mathrel{\wedge}{}\<[E]%
\\
\>[44]{}(\Varid{nbeOpen}\;\Varid{ts}\;\Varid{e}\;\Varid{t})\equiv (\Varid{nbeOpen}\;\Varid{ts}\;\Varid{e'}\;\Varid{t}){}\<[E]%
\\
\>[B]{}\Varid{chkTerm}\;\Varid{ts}\;{}\<[13]%
\>[13]{}(\Conid{Sigma}\;\Varid{t}\;\Varid{r})\;{}\<[26]%
\>[26]{}(\Conid{Pair}\;\Varid{e}\;\Varid{e'}){}\<[41]%
\>[41]{}\mathrel{=}{}\<[41E]%
\>[44]{}\Varid{chkTerm}\;\Varid{ts}\;\Varid{t}\;\Varid{e}\mathrel{\wedge}{}\<[E]%
\\
\>[44]{}\Varid{chkTerm}\;\Varid{ts}\;(\Varid{nbeOpenTy}\;\Varid{ts}\;(\Varid{sgSub}\;\Varid{r}\;\Varid{e}))\;\Varid{e'}{}\<[E]%
\\
\>[B]{}\Varid{chkTerm}\;\Varid{ts}\;{}\<[13]%
\>[13]{}\Conid{Nat}\;{}\<[26]%
\>[26]{}\Conid{Zero}{}\<[41]%
\>[41]{}\mathrel{=}{}\<[41E]%
\>[44]{}\Conid{True}{}\<[E]%
\\
\>[B]{}\Varid{chkTerm}\;\Varid{ts}\;{}\<[13]%
\>[13]{}\Conid{Nat}\;{}\<[26]%
\>[26]{}(\Conid{Suc}\;\Varid{t}){}\<[41]%
\>[41]{}\mathrel{=}{}\<[41E]%
\>[44]{}\Varid{chkTerm}\;\Varid{ts}\;\Conid{Nat}\;\Varid{t}{}\<[E]%
\\
\>[B]{}\Varid{chkTerm}\;\Varid{ts}\;{}\<[13]%
\>[13]{}(\Conid{Prf}\;\Varid{t})\;{}\<[26]%
\>[26]{}(\Conid{Box}\;\Varid{e}){}\<[41]%
\>[41]{}\mathrel{=}{}\<[41E]%
\>[44]{}\Varid{chkTerm}\;\Varid{ts}\;\Varid{t}\;\Varid{e}{}\<[E]%
\\
\>[B]{}\Varid{chkTerm}\;\Varid{ts}\;{}\<[13]%
\>[13]{}(\Conid{Enum}\;\Varid{n})\;{}\<[26]%
\>[26]{}(\Conid{Const}\;\Varid{m}\;\Varid{i}){}\<[41]%
\>[41]{}\mathrel{=}{}\<[41E]%
\>[44]{}\Varid{m}\equiv \Varid{n}\mathrel{\wedge}\Varid{i}\mathbin{<}\Varid{n}{}\<[E]%
\\
\>[B]{}\Varid{chkTerm}\;\Varid{ts}\;{}\<[13]%
\>[13]{}\Varid{t}\;{}\<[26]%
\>[26]{}\Varid{e}\mid \Varid{neutral}\;\Varid{e}{}\<[41]%
\>[41]{}\mathrel{=}{}\<[41E]%
\>[44]{}\Varid{chkNeTerm}\;\Varid{ts}\;\Varid{t}\;\Varid{e}{}\<[E]%
\\
\>[B]{}\Varid{chkTerm}\;\anonymous \;{}\<[13]%
\>[13]{}\anonymous \;{}\<[26]%
\>[26]{}\anonymous {}\<[41]%
\>[41]{}\mathrel{=}{}\<[41E]%
\>[44]{}\Conid{False}{}\<[E]%
\\[\blanklineskip]%
\>[B]{}\Varid{neutral}\mathbin{::}\Conid{Term}\to \Conid{Bool}{}\<[E]%
\\
\>[B]{}\Varid{neutral}\;\Conid{Q}{}\<[29]%
\>[29]{}\mathrel{=}\Conid{True}{}\<[E]%
\\
\>[B]{}\Varid{neutral}\;(\Conid{Sub}\;\Conid{Q}\;\Varid{s}){}\<[29]%
\>[29]{}\mathrel{=}\Conid{True}{}\<[E]%
\\
\>[B]{}\Varid{neutral}\;(\Conid{App}\;\Varid{k}\;\Varid{v}){}\<[29]%
\>[29]{}\mathrel{=}\Conid{True}{}\<[E]%
\\
\>[B]{}\Varid{neutral}\;(\Conid{Fst}\;\Varid{k}){}\<[29]%
\>[29]{}\mathrel{=}\Conid{True}{}\<[E]%
\\
\>[B]{}\Varid{neutral}\;(\Conid{Snd}\;\Varid{k}){}\<[29]%
\>[29]{}\mathrel{=}\Conid{True}{}\<[E]%
\\
\>[B]{}\Varid{neutral}\;(\Conid{Natrec}\;\Varid{t'}\;\Varid{v}\;\Varid{v'}\;\Varid{k}){}\<[29]%
\>[29]{}\mathrel{=}\Conid{True}{}\<[E]%
\\
\>[B]{}\Varid{neutral}\;(\Conid{Case}\;\Varid{n}\;\Varid{b}\;\Varid{ts}\;\Varid{t}){}\<[29]%
\>[29]{}\mathrel{=}\Conid{True}{}\<[E]%
\\
\>[B]{}\Varid{neutral}\;(\Conid{Where}\;\Varid{t}\;\Varid{b}\;\Varid{p}){}\<[29]%
\>[29]{}\mathrel{=}\Conid{True}{}\<[E]%
\\
\>[B]{}\Varid{neutral}\;\anonymous {}\<[29]%
\>[29]{}\mathrel{=}\Conid{False}{}\<[E]%
\\[\blanklineskip]%
\>[B]{}\Varid{erase}\mathbin{::}\Conid{Type}\to \Conid{Type}{}\<[E]%
\\
\>[B]{}\Varid{erase}\;(\Conid{Singl}\;\Varid{e}\;\Varid{t}){}\<[20]%
\>[20]{}\mathrel{=}\Varid{erase}\;\Varid{t}{}\<[E]%
\\
\>[B]{}\Varid{erase}\;\Varid{t}{}\<[20]%
\>[20]{}\mathrel{=}\Varid{t}{}\<[E]%
\\[\blanklineskip]%
\>[B]{}\Varid{maybeEr}\mathbin{::}\Conid{Maybe}\;\Conid{Type}\to \Conid{Maybe}\;\Conid{Type}{}\<[E]%
\\
\>[B]{}\Varid{maybeEr}\mathrel{=}\Varid{maybe}\;\Conid{Nothing}\;(\Conid{Just}\mathbin{\circ}\Varid{erase}){}\<[E]%
\\[\blanklineskip]%
\>[B]{}\Varid{chkNeTerm}\mathbin{::}\Conid{Ctx}\to \Conid{Type}\to \Conid{Term}\to \Conid{Bool}{}\<[E]%
\\
\>[B]{}\Varid{chkNeTerm}\;\Varid{ts}\;\Varid{t}\;\Varid{e}\mathrel{=}{}\<[21]%
\>[21]{}\mathbf{case}\;{}\<[27]%
\>[27]{}\Varid{maybeEr}\;(\Varid{infType}\;\Varid{ts}\;\Varid{e})\;\mathbf{of}{}\<[E]%
\\
\>[21]{}\Conid{Just}\;\Varid{t'}{}\<[30]%
\>[30]{}\to \Varid{t}\equiv \Varid{t'}{}\<[E]%
\\
\>[21]{}\Conid{Nothing}{}\<[30]%
\>[30]{}\to \Conid{False}{}\<[E]%
\ColumnHook
\end{hscode}\resethooks

\subsection*{Inferring the types of neutral terms}\bla\par

\begin{hscode}\SaveRestoreHook
\column{B}{@{}>{\hspre}l<{\hspost}@{}}%
\column{9}{@{}>{\hspre}l<{\hspost}@{}}%
\column{13}{@{}>{\hspre}l<{\hspost}@{}}%
\column{16}{@{}>{\hspre}l<{\hspost}@{}}%
\column{17}{@{}>{\hspre}l<{\hspost}@{}}%
\column{19}{@{}>{\hspre}l<{\hspost}@{}}%
\column{22}{@{}>{\hspre}l<{\hspost}@{}}%
\column{29}{@{}>{\hspre}l<{\hspost}@{}}%
\column{30}{@{}>{\hspre}c<{\hspost}@{}}%
\column{30E}{@{}l@{}}%
\column{33}{@{}>{\hspre}l<{\hspost}@{}}%
\column{35}{@{}>{\hspre}l<{\hspost}@{}}%
\column{37}{@{}>{\hspre}l<{\hspost}@{}}%
\column{39}{@{}>{\hspre}l<{\hspost}@{}}%
\column{43}{@{}>{\hspre}l<{\hspost}@{}}%
\column{47}{@{}>{\hspre}l<{\hspost}@{}}%
\column{50}{@{}>{\hspre}l<{\hspost}@{}}%
\column{52}{@{}>{\hspre}l<{\hspost}@{}}%
\column{53}{@{}>{\hspre}l<{\hspost}@{}}%
\column{55}{@{}>{\hspre}l<{\hspost}@{}}%
\column{57}{@{}>{\hspre}l<{\hspost}@{}}%
\column{E}{@{}>{\hspre}l<{\hspost}@{}}%
\>[B]{}\Varid{nbeType}\mathbin{::}\Conid{Ctx}\to \Conid{Type}\to \Conid{Maybe}\;\Conid{Type}{}\<[E]%
\\
\>[B]{}\Varid{nbeType}\;\Varid{ctx}\;\Varid{t}\mathrel{=}\Conid{Just}\;(\Varid{nbeOpenTy}\;\Varid{ctx}\;\Varid{t}){}\<[E]%
\\[\blanklineskip]%
\>[B]{}\Varid{infType}\mathbin{::}\Conid{Ctx}\to \Conid{Term}\to \Conid{Maybe}\;\Conid{Type}{}\<[E]%
\\
\>[B]{}\Varid{infType}\;(\Varid{t}\mathbin{:}\Varid{ts})\;\Conid{Q}{}\<[30]%
\>[30]{}\mathrel{=}{}\<[30E]%
\>[33]{}\Varid{nbeType}\;(\Varid{t}\mathbin{:}\Varid{ts})\;(\Conid{Sub}\;\Varid{t}\;\Conid{P}){}\<[E]%
\\
\>[B]{}\Varid{infType}\;\Varid{ts}\;{}\<[16]%
\>[16]{}(\Conid{Sub}\;\Conid{Q}\;\Varid{s}){}\<[30]%
\>[30]{}\mathrel{=}{}\<[30E]%
\>[33]{}\mathbf{case}\;\Varid{infType}\;(\Varid{infCtx}\;\Varid{ts}\;\Varid{s})\;\Conid{Q}\;\mathbf{of}{}\<[E]%
\\
\>[33]{}\Conid{Just}\;\Varid{t}\to \Varid{nbeType}\;\Varid{ts}\;(\Conid{Sub}\;\Varid{t}\;\Varid{s}){}\<[E]%
\\
\>[33]{}\anonymous \to \Conid{Nothing}{}\<[E]%
\\[\blanklineskip]%
\>[B]{}\Varid{infType}\;\Varid{ts}\;(\Conid{App}\;\Varid{e}\;\Varid{e'}){}\<[30]%
\>[30]{}\mathrel{=}{}\<[30E]%
\>[33]{}\mathbf{case}\;\Varid{maybeEr}\;(\Varid{infType}\;\Varid{ts}\;\Varid{e})\;\mathbf{of}{}\<[E]%
\\
\>[33]{}\hsindent{2}{}\<[35]%
\>[35]{}\Conid{Just}\;(\Conid{Fun}\;\Varid{t}\;\Varid{t'})\to {}\<[E]%
\\
\>[35]{}\hsindent{4}{}\<[39]%
\>[39]{}\mathbf{if}\;\Varid{chkTerm}\;\Varid{ts}\;\Varid{t}\;\Varid{e'}{}\<[E]%
\\
\>[35]{}\hsindent{4}{}\<[39]%
\>[39]{}\mathbf{then}\;\Varid{nbeType}\;\Varid{ts}\;(\Varid{sgSub}\;\Varid{t'}\;\Varid{e'}){}\<[E]%
\\
\>[35]{}\hsindent{4}{}\<[39]%
\>[39]{}\mathbf{else}\;\Conid{Nothing}{}\<[E]%
\\
\>[33]{}\hsindent{2}{}\<[35]%
\>[35]{}\anonymous \to \Conid{Nothing}{}\<[E]%
\\[\blanklineskip]%
\>[B]{}\Varid{infType}\;\Varid{ts}\;(\Conid{Fst}\;\Varid{e}){}\<[30]%
\>[30]{}\mathrel{=}{}\<[30E]%
\>[33]{}\mathbf{case}\;\Varid{maybeEr}\;(\Varid{infType}\;\Varid{ts}\;\Varid{e})\;\mathbf{of}{}\<[E]%
\\
\>[33]{}\hsindent{2}{}\<[35]%
\>[35]{}\Conid{Just}\;(\Conid{Sigma}\;\Varid{t}\;\Varid{t'})\to {}\<[57]%
\>[57]{}\Conid{Just}\;\Varid{t}{}\<[E]%
\\
\>[33]{}\hsindent{2}{}\<[35]%
\>[35]{}\anonymous \to \Conid{Nothing}{}\<[E]%
\\
\>[B]{}\Varid{infType}\;\Varid{ts}\;(\Conid{Snd}\;\Varid{e}){}\<[30]%
\>[30]{}\mathrel{=}{}\<[30E]%
\>[33]{}\mathbf{case}\;\Varid{maybeEr}\;(\Varid{infType}\;\Varid{ts}\;\Varid{e})\;\mathbf{of}{}\<[E]%
\\
\>[33]{}\hsindent{2}{}\<[35]%
\>[35]{}\Conid{Just}\;(\Conid{Sigma}\;\Varid{t}\;\Varid{t'})\to \Varid{nbeType}\;\Varid{ts}\;(\Varid{sgSub}\;\Varid{t'}\;(\Conid{Fst}\;\Varid{e})){}\<[E]%
\\
\>[33]{}\hsindent{2}{}\<[35]%
\>[35]{}\anonymous \to \Conid{Nothing}{}\<[E]%
\\
\>[B]{}\Varid{infType}\;\Varid{ts}\;(\Conid{Natrec}\;\Varid{t}\;\Varid{v}\;\Varid{w}\;\Varid{k}){}\<[30]%
\>[30]{}\mathrel{=}{}\<[30E]%
\>[33]{}\mathbf{case}\;\Varid{maybeEr}\;(\Varid{infType}\;\Varid{ts}\;\Varid{k})\;\mathbf{of}{}\<[E]%
\\
\>[33]{}\Conid{Just}\;\Conid{Nat}\to \mathbf{if}{}\<[E]%
\\
\>[33]{}\hsindent{4}{}\<[37]%
\>[37]{}\Varid{chkType}\;(\Conid{Nat}\mathbin{:}\Varid{ts})\;\Varid{t}\mathrel{\wedge}{}\<[E]%
\\
\>[33]{}\hsindent{4}{}\<[37]%
\>[37]{}\Varid{chkTerm}\;\Varid{ts}\;(\Varid{nbeOpenTy}\;\Varid{ts}\;(\Varid{sgSub}\;\Varid{t}\;\Conid{Zero}))\;\Varid{v}\mathrel{\wedge}{}\<[E]%
\\
\>[33]{}\hsindent{4}{}\<[37]%
\>[37]{}\Varid{chkTerm}\;(\Conid{Nat}\mathbin{:}\Varid{ts})\;{}\<[E]%
\\
\>[37]{}\hsindent{10}{}\<[47]%
\>[47]{}(\Conid{Fun}\;(\Varid{sgSub}\;\Varid{t}\;\Conid{Q})\;{}\<[E]%
\\
\>[47]{}\hsindent{5}{}\<[52]%
\>[52]{}(\Varid{sgSub}\;\Varid{t}\;(\Conid{Suc}\;(\Conid{Sub}\;\Conid{Q}\;\Conid{P}))))\;\Varid{w}{}\<[E]%
\\
\>[33]{}\hsindent{4}{}\<[37]%
\>[37]{}\mathbf{then}\;\Varid{nbeType}\;\Varid{ts}\;(\Varid{sgSub}\;\Varid{t}\;\Varid{k}){}\<[E]%
\\
\>[33]{}\hsindent{4}{}\<[37]%
\>[37]{}\mathbf{else}\;\Conid{Nothing}{}\<[E]%
\\
\>[33]{}\anonymous \to \Conid{Nothing}{}\<[E]%
\\
\>[B]{}\Varid{infType}\;\Varid{ts}\;(\Conid{Where}\;\Varid{t}\;\Varid{b}\;\Varid{k}){}\<[30]%
\>[30]{}\mathrel{=}{}\<[30E]%
\>[33]{}\mathbf{case}\;\Varid{maybeEr}\;(\Varid{infType}\;\Varid{ts}\;\Varid{k})\;\mathbf{of}{}\<[E]%
\\
\>[33]{}\Conid{Just}\;(\Conid{Prf}\;\Varid{t'})\to \mathbf{if}\;\Varid{chkType}\;\Varid{ts}\;\Varid{t}\mathrel{\wedge}{}\<[E]%
\\
\>[33]{}\hsindent{20}{}\<[53]%
\>[53]{}\Varid{chkTerm}\;(\Varid{t'}\mathbin{:}\Varid{ts})\;\Varid{t}\;\Varid{b}\mathrel{\wedge}{}\<[E]%
\\
\>[53]{}\hsindent{2}{}\<[55]%
\>[55]{}\Varid{nbeOpen}\;\Varid{ts'}\;\Varid{w}\;(\Conid{Sub}\;\Varid{b}\;(\Conid{Ext}\;(\Varid{subs}\;\mathrm{1})\;\Conid{Q}))\equiv {}\<[E]%
\\
\>[53]{}\hsindent{2}{}\<[55]%
\>[55]{}\Varid{nbeOpen}\;\Varid{ts'}\;\Varid{w}\;(\Conid{Sub}\;\Varid{b}\;\Conid{P}){}\<[E]%
\\
\>[33]{}\hsindent{17}{}\<[50]%
\>[50]{}\mathbf{then}\;\Conid{Just}\;\Varid{t}{}\<[E]%
\\
\>[33]{}\hsindent{17}{}\<[50]%
\>[50]{}\mathbf{else}\;\Conid{Nothing}{}\<[E]%
\\
\>[33]{}\hsindent{4}{}\<[37]%
\>[37]{}\mathbf{where}\;\Varid{ts'}\mathrel{=}\Conid{Sub}\;\Varid{t'}\;\Conid{P}\mathbin{:}\Varid{t'}\mathbin{:}\Varid{ts}{}\<[E]%
\\
\>[37]{}\hsindent{6}{}\<[43]%
\>[43]{}\Varid{w}{}\<[47]%
\>[47]{}\mathrel{=}\Conid{Sub}\;\Varid{t}\;(\Varid{subs}\;\mathrm{1}){}\<[E]%
\\
\>[33]{}\anonymous \to \Conid{Nothing}{}\<[E]%
\\
\>[B]{}\Varid{infType}\;\Varid{ts}\;(\Conid{Case}\;\Varid{n}\;\Varid{b}\;\Varid{cs}\;\Varid{k}){}\<[30]%
\>[30]{}\mathrel{=}{}\<[30E]%
\>[33]{}\mathbf{case}\;\Varid{maybeEr}\;(\Varid{infType}\;\Varid{ts}\;\Varid{k})\;\mathbf{of}{}\<[E]%
\\
\>[33]{}\Conid{Just}\;(\Conid{Enum}\;\Varid{m})\to \mathbf{if}\;\Varid{m}\equiv \Varid{n}\mathrel{\wedge}{}\<[E]%
\\
\>[33]{}\hsindent{20}{}\<[53]%
\>[53]{}\Varid{chkType}\;(\Conid{Enum}\;\Varid{n}\mathbin{:}\Varid{ts})\;\Varid{b}\mathrel{\wedge}{}\<[E]%
\\
\>[33]{}\hsindent{20}{}\<[53]%
\>[53]{}\Varid{chkList}\;\Varid{ts}\;\Varid{n}\;\Varid{b}\;\mathrm{0}\;\Varid{cs}{}\<[E]%
\\
\>[33]{}\hsindent{17}{}\<[50]%
\>[50]{}\mathbf{then}\;\Varid{nbeType}\;\Varid{ts}\;(\Varid{sgSub}\;\Varid{b}\;\Varid{k}){}\<[E]%
\\
\>[33]{}\hsindent{17}{}\<[50]%
\>[50]{}\mathbf{else}\;\Conid{Nothing}{}\<[E]%
\\
\>[33]{}\anonymous \to \Conid{Nothing}{}\<[E]%
\\[\blanklineskip]%
\>[B]{}\Varid{infType}\;\anonymous \;\anonymous {}\<[30]%
\>[30]{}\mathrel{=}{}\<[30E]%
\>[33]{}\Conid{Nothing}{}\<[E]%
\\[\blanklineskip]%
\>[B]{}\Varid{chkList}\mathbin{::}\Conid{Ctx}\to \Conid{Int}\to \Conid{Type}\to \Conid{Int}\to [\mskip1.5mu \Conid{Term}\mskip1.5mu]\to \Conid{Bool}{}\<[E]%
\\
\>[B]{}\Varid{chkList}\;\Varid{ts}\;{}\<[13]%
\>[13]{}\anonymous \;{}\<[16]%
\>[16]{}\anonymous \;{}\<[19]%
\>[19]{}\anonymous \;{}\<[22]%
\>[22]{}[\mskip1.5mu \mskip1.5mu]{}\<[30]%
\>[30]{}\mathrel{=}{}\<[30E]%
\>[33]{}\Conid{True}{}\<[E]%
\\
\>[B]{}\Varid{chkList}\;\Varid{ts}\;{}\<[13]%
\>[13]{}\Varid{n}\;{}\<[16]%
\>[16]{}\Varid{b}\;{}\<[19]%
\>[19]{}\Varid{i}\;{}\<[22]%
\>[22]{}(\Varid{e}\mathbin{:}\Varid{es}){}\<[30]%
\>[30]{}\mathrel{=}{}\<[30E]%
\>[33]{}\Varid{chkTerm}\;\Varid{ts}\;(\Varid{nbeOpenTy}\;\Varid{ts}\;(\Varid{sgSub}\;\Varid{b}\;(\Conid{Const}\;\Varid{n}\;\Varid{i})))\;\Varid{e}\mathrel{\wedge}{}\<[E]%
\\
\>[33]{}\Varid{chkList}\;\Varid{ts}\;\Varid{n}\;\Varid{b}\;(\Varid{i}\mathbin{+}\mathrm{1})\;\Varid{es}{}\<[E]%
\\[\blanklineskip]%
\>[B]{}\Varid{infCtx}\mathbin{::}\Conid{Ctx}\to \Conid{Subst}\to \Conid{Ctx}{}\<[E]%
\\
\>[B]{}\Varid{infCtx}\;{}\<[9]%
\>[9]{}(\Varid{t}\mathbin{:}\Varid{ts})\;{}\<[17]%
\>[17]{}\Conid{P}{}\<[29]%
\>[29]{}\mathrel{=}\Varid{ts}{}\<[E]%
\\
\>[B]{}\Varid{infCtx}\;{}\<[9]%
\>[9]{}(\Varid{t}\mathbin{:}\Varid{ts})\;{}\<[17]%
\>[17]{}(\Conid{Comp}\;\Conid{P}\;\Varid{s}){}\<[29]%
\>[29]{}\mathrel{=}\Varid{infCtx}\;\Varid{ts}\;\Varid{s}{}\<[E]%
\ColumnHook
\end{hscode}\resethooks

\input{lmcs-proofs}
\end{document}

%% file: prooftree.tex
\newdimen\proofrulebreadth \proofrulebreadth=.05em
\newdimen\proofdotseparation \proofdotseparation=1.25ex
\newdimen\proofrulebaseline \proofrulebaseline=2ex
\newcount\proofdotnumber \proofdotnumber=3
\let\then\relax
\def\hfi{\hskip0pt plus.0001fil}
\mathchardef\squigto="3A3B
\newif\ifinsideprooftree\insideprooftreefalse
\newif\ifonleftofproofrule\onleftofproofrulefalse
\newif\ifproofdots\proofdotsfalse
\newif\ifdoubleproof\doubleprooffalse
\let\wereinproofbit\relax
\newdimen\shortenproofleft
\newdimen\shortenproofright
\newdimen\proofbelowshift
\newbox\proofabove
\newbox\proofbelow
\newbox\proofrulename
\def\shiftproofbelow{\let\next\relax\afterassignment\setshiftproofbelow\dimen0 }
\def\shiftproofbelowneg{\def\next{\multiply\dimen0 by-1 }%
\afterassignment\setshiftproofbelow\dimen0 }
\def\setshiftproofbelow{\next\proofbelowshift=\dimen0 }
\def\setproofrulebreadth{\proofrulebreadth}

\def\prooftree{
\ifnum  \lastpenalty=1
\then   \unpenalty
\else   \onleftofproofrulefalse
\fi
\ifonleftofproofrule
\else   \ifinsideprooftree
        \then   \hskip.5em plus1fil
        \fi
\fi
\bgroup
\setbox\proofbelow=\hbox{}\setbox\proofrulename=\hbox{}%
\let\justifies\proofover\let\leadsto\proofoverdots\let\Justifies\proofoverdbl
\let\using\proofusing\let\[\prooftree
\ifinsideprooftree\let\]\endprooftree\fi
\proofdotsfalse\doubleprooffalse
\let\thickness\setproofrulebreadth
\let\shiftright\shiftproofbelow \let\shift\shiftproofbelow
\let\shiftleft\shiftproofbelowneg
\let\ifwasinsideprooftree\ifinsideprooftree
\insideprooftreetrue
\setbox\proofabove=\hbox\bgroup$\displaystyle 
\let\wereinproofbit\prooftree
\shortenproofleft=0pt \shortenproofright=0pt \proofbelowshift=0pt
\onleftofproofruletrue\penalty1
}

\def\eproofbit{
\ifx    \wereinproofbit\prooftree
\then   \ifcase \lastpenalty
        \then   \shortenproofright=0pt  
        \or     \unpenalty\hfil         
        \or     \unpenalty\unskip       
        \else   \shortenproofright=0pt  
        \fi
\fi
\global\dimen0=\shortenproofleft
\global\dimen1=\shortenproofright
\global\dimen2=\proofrulebreadth
\global\dimen3=\proofbelowshift
\global\dimen4=\proofdotseparation
\global\count255=\proofdotnumber
$\egroup  
\shortenproofleft=\dimen0
\shortenproofright=\dimen1
\proofrulebreadth=\dimen2
\proofbelowshift=\dimen3
\proofdotseparation=\dimen4
\proofdotnumber=\count255
}

\def\proofover{
\eproofbit 
\setbox\proofbelow=\hbox\bgroup 
\let\wereinproofbit\proofover
$\displaystyle
}%
\def\proofoverdbl{
\eproofbit 
\doubleprooftrue
\setbox\proofbelow=\hbox\bgroup 
\let\wereinproofbit\proofoverdbl
$\displaystyle
}%
\def\proofoverdots{
\eproofbit 
\proofdotstrue
\setbox\proofbelow=\hbox\bgroup 
\let\wereinproofbit\proofoverdots
$\displaystyle
}%
\def\proofusing{
\eproofbit 
\setbox\proofrulename=\hbox\bgroup 
\let\wereinproofbit\proofusing
\kern0.3em$
}

\def\endprooftree{
\eproofbit 
  \dimen5 =0pt
\dimen0=\wd\proofabove \advance\dimen0-\shortenproofleft
\advance\dimen0-\shortenproofright
\dimen1=.5\dimen0 \advance\dimen1-.5\wd\proofbelow
\dimen4=\dimen1
\advance\dimen1\proofbelowshift \advance\dimen4-\proofbelowshift
\ifdim  \dimen1<0pt
\then   \advance\shortenproofleft\dimen1
        \advance\dimen0-\dimen1
        \dimen1=0pt
        \ifdim  \shortenproofleft<0pt
        \then   \setbox\proofabove=\hbox{%
                        \kern-\shortenproofleft\unhbox\proofabove}%
                \shortenproofleft=0pt
        \fi
\fi
\ifdim  \dimen4<0pt
\then   \advance\shortenproofright\dimen4
        \advance\dimen0-\dimen4
        \dimen4=0pt
\fi
\ifdim  \shortenproofright<\wd\proofrulename
\then   \shortenproofright=\wd\proofrulename
\fi
\dimen2=\shortenproofleft \advance\dimen2 by\dimen1
\dimen3=\shortenproofright\advance\dimen3 by\dimen4
\ifproofdots
\then
        \dimen6=\shortenproofleft \advance\dimen6 .5\dimen0
        \setbox1=\vbox to\proofdotseparation{\vss\hbox{$\cdot$}\vss}%
        \setbox0=\hbox{%
                \advance\dimen6-.5\wd1
                \kern\dimen6
                $\vcenter to\proofdotnumber\proofdotseparation
                        {\leaders\box1\vfill}$%
                \unhbox\proofrulename}%
\else   \dimen6=\fontdimen22\the\textfont2 
        \dimen7=\dimen6
        \advance\dimen6by.5\proofrulebreadth
        \advance\dimen7by-.5\proofrulebreadth
        \setbox0=\hbox{%
                \kern\shortenproofleft
                \ifdoubleproof
                \then   \hbox to\dimen0{%
                        $\mathsurround0pt\mathord=\mkern-6mu%
                        \cleaders\hbox{$\mkern-2mu=\mkern-2mu$}\hfill
                        \mkern-6mu\mathord=$}%
                \else   \vrule height\dimen6 depth-\dimen7 width\dimen0
                \fi
                \unhbox\proofrulename}%
        \ht0=\dimen6 \dp0=-\dimen7
\fi
\let\doll\relax
\ifwasinsideprooftree
\then   \let\VBOX\vbox
\else   \ifmmode\else$\let\doll=$\fi
        \let\VBOX\vcenter
\fi
\VBOX   {\baselineskip\proofrulebaseline \lineskip.2ex
        \expandafter\lineskiplimit\ifproofdots0ex\else-0.6ex\fi
        \hbox   spread\dimen5   {\hfi\unhbox\proofabove\hfi}%
        \hbox{\box0}%
        \hbox   {\kern\dimen2 \box\proofbelow}}\doll%
\global\dimen2=\dimen2
\global\dimen3=\dimen3
\egroup 
\ifonleftofproofrule
\then   \shortenproofleft=\dimen2
\fi
\shortenproofright=\dimen3
\onleftofproofrulefalse
\ifinsideprooftree
\then   \hskip.5em plus 1fil \penalty2
\fi
}


%% file: commands.tex
\newcommand{\proofLine}[2]{&#1&&\text{#2}}
\newcommand{\textdbldagger}{\textdagger\!\!\textdagger}
\newcommand{\textdbldaggerdbl}{\textdaggerdbl\!\!\textdaggerdbl}

\newcommand{\abelmail}{andreas.abel@ifi.lmu.de}
\newcommand{\coquandmail}{coquand@chalmers.se}
\newcommand{\paganomail}{pagano@famaf.unc.edu.ar}
\newcommand{\bla}{\ensuremath{\mbox{$$}}} 

\newcommand{\of}{\!:\!}
\newcommand{\subtype}{\mathrel{\mathord{<}\mathord{:}}}
\newcommand{\derN}{\vdash}
\newcommand{\FV}{\mathsf{FV}}
\newcommand{\ru}{\dfrac}
\newcommand{\NN}{\mathbb{N}}
\def\bprg#1\eprg{\begin{multline*}\qquad#1\end{multline*}}

\newcommand{\ABwhere}[3]{#1 \; \mathsf{where}\; [#2] \leftarrow #3}

\newcommand{\lambdaSing}{\texorpdfstring{\ensuremath{\lambda^{\mathsf{Sing}}}}{Singletons}}
\newcommand{\lambdaPI}{\texorpdfstring{\ensuremath{\lambda^{\mathsf{Irr}}}}{Proof-irrelevance}}
\newcommand{\lambdaIrr}{\lambdaPI}

\newcommand{\Vecraw}{\mathsf{Vec}}
\newcommand{\tVec}[1]{\Vecraw\;#1\;}
\newcommand{\tleqraw}{\mathsf{leq}}
\newcommand{\tleq}[1]{\tleqraw\;#1\;}
\newcommand{\Trueraw}{\mathsf{True}}
\newcommand{\TrueC}{\Trueraw\;}
\newcommand{\tlookupraw}{\mathsf{lookup}}
\newcommand{\tlookup}[4]{\tlookupraw\;#1\;#2\;#3\;#4\;}
\newcommand{\funT}[2]{(#1 \of #2) \to}
\newcommand{\funS}[2]{(#1 : #2) \to}
\newcommand{\sigT}[2]{(#1 \of #2) \times}
\newcommand{\Bool}{\mathsf{Bool}}
\newcommand{\ttrue}{\mathsf{true}}
\newcommand{\tfalse}{\mathsf{false}}
\newcommand{\Ltraw}{\mathsf{Lt}}
\newcommand{\Lt}[1]{\Ltraw\,#1\,}
\newcommand{\tmagicraw}{\mathsf{magic}}
\newcommand{\tmagic}[1]{\tmagicraw\,#1\,}

\newcommand{\den}{\semc}
\newcommand{\Den}{\semc}
\newcommand{\Lam}{\mathsf{Lam}}
\newcommand{\App}{\mathsf{App}}
\newcommand{\D}{\mathsf{D}}

\newcommand{\tq}{\mathsf{q}}
\newcommand{\tp}{\mathsf{p}}
\newcommand{\up}{\upa}
\newcommand{\down}{\da}
\newcommand{\dapp}{\cdot}

\newcommand{\A}{\mathcal{A}}
\newcommand{\seq}[1]{#1^*}
\newcommand{\Aseq}{\seq\A}
\newcommand{\tfst}{\mathsf{fst}}
\newcommand{\tsnd}{\mathsf{snd}}
\newcommand{\tPair}{\mathsf{Pair}}
\newcommand{\update}[3]{#1[#3/#2]}

\newcommand{\tyrule}[3]{\inferrule* [right=(#1)] {#2} {#3}}

\newcommand{\algrule}[2]{\inferrule* {#1} {#2}} 

\newcommand{\rulename}[1]{\ensuremath{\mbox{\textsc{(#1)}}}}

\newcommand{\ctx}{\mathsf{Ctx}}
\newcommand{\type}[1]{\mathsf{Type}(#1)}
\newcommand{\term}[2]{\mathsf{Term}(#1,#2)}

\newcommand{\into}{\rightarrow}

\newcommand{\ectx}{\diamond} 
\newcommand{\ctxe}[2]{#1.#2} 

\newcommand{\idsubs}[1]{\mathsf{id}_{#1}}
\newcommand{\esubs}{\langle\,\!\rangle}
\newcommand{\exsubs}[2]{( #1, #2)}
\newcommand{\subsc}[2]{#1\,#2}
\newcommand{\subsTm}[2]{#1\,#2}
\newcommand{\subsTy}[2]{#1\,#2}

\newcommand{\TmU}{\mathsf{U}}
\newcommand{\F}[2]{\mathsf{Fun}\,#1\,#2}
\newcommand{\DSum}[2]{\Sigma\,#1\,#2}
\newcommand{\p}{\mathsf{p}}
\newcommand{\q}{\mathsf{q}}
\newcommand{\appTm}[2]{\mathsf{app}\;#1\;#2} 
\newcommand{\singTm}[2]{\{#1\}_{#2}}
\newcommand{\depair}[2]{(#1,#2)}
\newcommand{\dfst}[1]{\mathsf{fst}\;#1}
\newcommand{\dsnd}[1]{\mathsf{snd}\;#1}
\newcommand{\natty}{\mathsf{Nat}}
\newcommand{\ztm}{\mathsf{zero}}
\newcommand{\suctm}[1]{\mathsf{suc}\;#1}
\newcommand{\natrecraw}{\mathsf{natrec}}
\newcommand{\natrec}[4]{\natrecraw\;#1\;#2\;#3\;#4}

\newcommand{\boxty}[1]{[#1]}
\newcommand{\boxtm}[1]{[#1]}
\newcommand{\subid}[2]{\exsubs{\idsubs{#1}}{#2}}
\newcommand{\ind}[1]{\mathsf{v}_{#1}} 

\newcommand{\dctx}[1]{#1\vdash} 
\newcommand{\dtype}[2]{#1\vdash#2} 
\newcommand{\dterm}[3]{#1\vdash#3:#2} 
\newcommand{\dsubs}[3]{#1\vdash#3:#2} 
\newcommand{\deqctx}[2]{\vdash#1=#2} 
\newcommand{\deqtype}[3]{#1\vdash#2=#3} 
\newcommand{\deqterm}[4]{#1\vdash#3=#4:#2} 
\newcommand{\deqsubs}[4]{#1\vdash #3=#4: #2}

\newcommand{\vdashp}{\vdash^{\oprf}}
\newcommand{\sdctx}[1]{#1\vdashp} 
\newcommand{\sdtype}[2]{#1\vdashp#2} 
\newcommand{\sdterm}[3]{#1\vdashp#3:#2} 
\newcommand{\sdsubs}[3]{#1\vdashp#3:#2} 
 
\newcommand{\sdeqtype}[3]{#1\vdashp#2=#3} 
\newcommand{\sdeqterm}[4]{#1\vdashp#3=#4:#2} 
\newcommand{\sdeqsubs}[4]{#1\vdashp #3=#4: #2}

\newcommand{\lift}[2]{\subsTm{#2}{\p^{#1}}}
\newcommand{\Da}[1]{\mathop{\Downarrow} #1} 
\newcommand{\upa}[2]{\mathop{\uparrow}\nolimits_{#1}{#2}} 
\newcommand{\da}[2]{\mathop{\downarrow}\nolimits_{#1}{#2}}
\newcommand{\reify}[2]{\mathsf{R}_{#1}\,#2}
\newcommand{\reifyC}[2]{\mathsf{R}_{|#1|}\,#2}

\newcommand{\calR}{\mathcal{R}}
\newcommand{\calX}{\mathcal{X}}
\newcommand{\calF}{\mathcal{F}}
\newcommand{\X}{\mathcal{X}}
\newcommand{\smashed}[1]{|#1|}
\newcommand{\singD}[2]{\{\!\!\{#1\}\!\!\}_{#2}}

\newcommand{\sigD}[2]{\coprod\,#1\,#2}

\newcommand{\dom}{\mathit{dom}}
\newcommand{\vdom}{\dom}
\newcommand{\perD}{\mbox{PER}(D)}
\newcommand{\perne}{\mathord{\mathcal{N}\!\mathit{e}}}
\newcommand{\pernf}{\mathord{\mathcal{N}\!\mathit{f}}}
\newcommand{\perU}{\mathcal{U}}
\newcommand{\perT}{\mathcal{T}}
\newcommand{\Per}{\mathsf{Per}}

\newcommand{\Var}{\mathit{Var}}

\newcommand{\semc}[1]{[\![#1]\!]} 
\newcommand{\rel}{\sim}

\newcommand{\one}{\mathbf{1}}

\newcommand{\iLam}[1]{\mathsf{Lam}\,#1}
\newcommand{\iNe}[2]{\mathsf{App}\,#1\,#2}
\newcommand{\iPair}[2]{(#1 , #2)}
\newcommand{\iVar}[1]{\mathsf{Var}\,x_{#1}}
\newcommand{\iO}{\top}
\newcommand{\iU}{\mathsf{U}}
\newcommand{\iPi}[2]{\mathsf{Fun}\,#1\,#2}
\newcommand{\iSing}[2]{\mathsf{Sing}\,#1\,{#2}}
\newcommand{\iDs}[2]{\mathsf{Sum}\,#1\,#2}

\newcommand{\Fst}{\mathsf{Fst}}
\newcommand{\Snd}{\mathsf{Snd}}
\newcommand{\iBoolF}[1]{\Fst\;#1}
\newcommand{\iBoolT}[1]{\Snd\;#1}

\newcommand{\fst}{\tfst}
\newcommand{\snd}{\tsnd}

\newcommand{\fstnew}{\tfst\,}
\newcommand{\sndnew}{\tsnd\,}

\newcommand{\ertype}[1]{\overline{#1}}

\newcommand{\terms}{\mathit{Terms}}

\newcommand{\neterms}{\mathit{Ne}}
\newcommand{\nfterms}{\mathit{Nf}}

\newcommand{\chktype}[2]{#1 \vdash #2\, \mathord{\Leftarrow}}
\newcommand{\chkterm}[3]{#1\vdash#3\Leftarrow #2}
\newcommand{\inftype}[3]{#1\vdash #2\Rightarrow #3}

\newcommand{\nbe}[1]{\mathbf{nbe}(#1)}
\newcommand{\nbety}[2]{\mathbf{nbe}_{#1}(#2)}
\newcommand{\nbetm}[3]{\mathbf{nbe}_{#1}^{#2}(#3)}
\newcommand{\norm}[1]{\mathbf{nbe}(#1)}
\newcommand{\normty}[2]{\mathbf{nbe}_{#1}(#2)}
\newcommand{\normtm}[3]{\mathbf{nbe}_{#1}^{#2}(#3)}

\newcommand{\ruleref}[1]{(\RefTirName{#1})}
\newcommand{\oprf}{\dprf}
\newcommand{\dprf}{\star}
\renewcommand{\boxty}[1]{\mathsf{Prf}\,#1}
\newcommand{\boxpar}[1]{\mathsf{Prf}\,(#1)}
\newcommand{\prfraw}{\mathsf{Prf}}
\newcommand{\prf}[1]{\mathsf{Prf}\,#1}
\newcommand{\recty}[1]{\mathsf{Rec}(#1)}
\newcommand{\whereraw}{\mathsf{where}}
\newcommand{\wheretm}[3]{#1\,\whereraw\!^{#3}\,#2}

\newcommand{\elimraw}{\mathsf{case}}
\newcommand{\Elimraw}{\mathsf{Case}}
\newcommand{\enum}[1]{\mathsf{N}_{#1}}
\newcommand{\elim}[4]{\elimraw\!^{#1}\;#2\;#3\;#4}
\newcommand{\elimn}[3]{\elimraw\;#1\;#2\;#3}
\newcommand{\const}[2]{\mathsf{c}^{#1}_{#2}}
\newcommand{\constn}[1]{\mathsf{c}_{#1}}
\newcommand{\constD}[2]{\mathsf{c}^{#1}_{#2}}
\newcommand{\elimD}[4]{\Elimraw\!^{#1}(#2, #3, #4)}
\newcommand{\enumD}[1]{\mathsf{N}_{#1}}
\newcommand{\tuple}[1]{\langle #1\rangle}
\newcommand{\erecraw}{\elimraw}
\newcommand{\erec}[4]{\erecraw^{#1}(#2,#3,#4)}

\newcommand{\enumrule}[2][n]{n\ensuremath{_{#1}}-{#2}}
\newcommand{\C}[1]{\ensuremath{\mathcal{C}_{#1}}}
\newcommand{\tnrecd}{[D\into D]\times D\times[D\into [D\into D]]\times D}
\newcommand{\iZero}{\ztm}
\newcommand{\iNat}{\natty}
\newcommand{\perNat}{\mathcal{N}}
\newcommand{\nrecraw}{\mathsf{natrec}}
\newcommand{\Nrecraw}{\mathsf{Natrec}}
\newcommand{\iNrec}[4]{\Nrecraw(#1,#2,#3,#4)}
\newcommand{\drec}[4]{\nrecraw(#1,#2,#3,#4)}
\newcommand{\iSuc}[1]{\mathsf{suc}\,#1}
\ifdefined\LONGVERSION
  \relax
\else
\newcommand{\LONGVERSION}[1]{}
\newcommand{\SHORTVERSION}[1]{#1}
\fi
\newcommand{\LONGSHORT}[2]{\LONGVERSION{#1}\SHORTVERSION{#2}}

\newcommand{\REDUNDANT}[1]{}
\newcommand{\EXPLAINREDUNDANT}[1]{#1}

\newcommand{\PrfIrrTitle}{\subsection{Calculus \lambdaPI with proof irrelevance}}

\newcommand{\para}[1]{
\LONGSHORT{\paragraph{\it #1.}}
          {\vspace{1ex}\noindent{\it #1.}}
}
\newcommand{\localpara}[1]{\noindent #1.}

%% file: tlca09-intro.tex
\section{Introduction and Related Work}
\label{sec:intro}
\setcounter{footnote}{0} 

\noindent One of the raisons d'\^{e}tre of
proof-checkers like Agda \cite{norell:PhD}, Coq \cite{coq81}, and
Epigram \cite{mcBride:afp04} is to decide if a given term has some
type (either checking for a given type or inferring one); i.e., if
a term corresponds to a proof of a proposition
\cite{harperHonsellPlotkin:LF}.  Hence, the convenience of such a
system is, in part, determined by the types for which the system can
check membership. We extend the decidability of type-checking done in
previous works \cite{abelAehligDybjer:mfps07,abelCoquandDybjer:lics07}
for Martin-L\"{o}f type theories \cite{mlitt,nordstroem:mltt} by
considering singleton types and proof-irrelevant propositions.

\LONGVERSION{We consider a type theory with a \emph{universe}, which
  allows large eliminations, i.e., types defined by recursion on
  natural numbers. 
The universe of small types was introduced by
  Martin-L\"{o}f \cite{martinlof72} for formalising category theory.
  Martin-L\"of presents universes in two different styles \cite{mlitt}: 
  {\em \`a la Russell} (the one considered here), and {\em \`a la Tarski}.}

\emph{Singleton types} were introduced by Aspinall \cite{aspinall:csl94} in
the context of specification languages. An important use of singletons
is as definitions by abbreviations (see
\cite{aspinall:csl94,coquandPollackTakeyama:fundinf05}); they were
also used to model translucent sums in the formalisation of SML
\cite{harper:popl07}.  It is interesting to consider singleton types
because beta-eta phase separation fails: one cannot do eta-expansion before
beta-normalisation of types because the shape of the types at which to
eta-expand is still unknown at this point; and one cannot postpone
eta-expansion after beta-normalisation, because eta-expansion at
singleton type can
trigger new beta-reductions. 
%
%
Stone and Harper \cite{stoneHarper:tocl06} decide type checking in a
logical framework (LF) with singleton types and subtyping.  Yet it is not clear
whether their method extends to computation on the type level.  As far
as we know, our work is the first where singleton types are considered
together with a universe.

De Bruijn proposed the concept of \emph{irrelevance of proofs} \cite{aut-4},
for reducing the burden in the formalisation of mathematics.  As shown
by Werner \cite{werner:strengthProofIrrelevance}, the use of
proof-irrelevance types together with sigma types is one way to get
subset types \`a la PVS \cite{pvs-sub} in type-theories having the
eta rule.  This style of subset types was also explored by Sozeau
\cite[Sec.~3.3]{sozeau:types06}\LONGVERSION{; for another presentation of subset types in
  Martin-L\"{o}f type-theory see \cite{sambin}}.
\LONGVERSION{Berardi conjectured that
  (impredicative) type-theory with proof-irrelevance is equivalent to
  constructive mathematics \cite{berardi}.}

Checking dependent types relies on checking types for equality.  To
this end, we compute $\eta$-long normal forms using
\emph{normalisation by evaluation} (NbE) \cite{martinlof:jaist04}.  
Syntactic expressions are evaluated into a semantic domain and then
\emph{reified} back to expressions in normal form.
To handle functional and
open expressions, the semantic domain has to be equipped
with variables; a major challenge in rigorous treatments of NbE has
been the problem to generate fresh identifiers.
Solutions include term families \cite{bergerSchwichtenberg:lics91},
liftable de Bruijn terms \cite{aehlig:nbe}, or Kripke semantics
\cite{abelCoquandDybjer:mpc08}.  In this work we present a novel
formulation of NbE which avoids the problem completely: reification is
split into an $\eta$-expansion phase $(\downarrow)$ in the semantics,
followed by a read back function ($\mathsf{R}$)
into the syntax which is indexed by
the number of already used variables.  This way, a standard PER
model is sufficient, and technical difficulties are avoided.

\paragraph{Outline.} 
In Section~\ref{sec:cwf-gat}, we first present \lambdaSing,
Martin-L\"of's logical framework with one universe and singleton
types, as a generalized algebraic theory \cite{cartmell}.  Secondly,
we introduce \lambdaPI, Martin-L\"of type theory with natural numbers,
sigma types, and proof-irrelevant propositions.  In
Section~\ref{sec:examples}, we show some examples using singleton types and
proof-irrelevant types.  In Section~\ref{sec:nbe-primer}, we present
briefly NbE for untyped and simply typed lambda calculi; in particular
we illustrate our novel approach to generate fresh identifiers.  In
Section~\ref{sec:semantics}, we define the semantics of the type
theories by a PER model and prove the soundness of the inference rules.
We use this model to introduce a normalization algorithm
$\textbf{nbe}$, 
for which we prove completeness (if $t = t'$ is derivable,
then $\nbe{t}$ and $\nbe{t'}$ are identical). 
The soundness of the algorithm (i.e., $t = \nbe{t}$ is derivable)
is proven by logical relations in Section~\ref{sec:logrel}.
In Section~\ref{sec:proof-alg}, we define a bi-directional algorithm for
checking the type of normal forms and inferring the type of neutral
terms.  More related work is discussed in Section~\ref{sec:related}.
\LONGVERSION{ The Haskell programs corresponding to the NbE, and
  type-checking algorithms are shown in the
  appendices~\ref{sec:nbealg} and~\ref{sec:alg}, respectively.  }


%% file: tlca09-calc.tex
\section{The Calculus as a Generalised Algebraic Theory}
\label{sec:cwf-gat}
 
\noindent In this section, we introduce the type theory.  In order to show the
modularity of our approach, we present it as two calculi $\lambdaSing$
and $\lambdaPI$: the first one has dependent function spaces,
singleton types, and a universe closed under function spaces and
singletons. In the second calculus we leave out singleton types and we
add natural numbers, sigma types, and proof-irrelevant propositions.
It is not clear if singleton types can be combined with
proof-irrelevant propositions without turning the system inconsistent.

We present the calculi using the formalism proposed by Cartmell for
generalised algebraic theories (GAT) \cite{cartmell}.  A GAT consists
of \emph{sort symbols} and \emph{operator symbols}, each with
a dependent typing, and \emph{equations} between \emph{sort
  expressions} and \emph{terms} (``operator expressions'').
Following Dybjer \cite{dybjer:internalTypeTheory}, we are using
``informal syntax'' where redundant arguments to operators are left
implicit. 

\LONGSHORT{
\subsection{Calculus \lambdaSing with singleton types}
\label{sec:calc-sing}
}{
\subsubsection*{Calculus with singleton types.}
\label{sec:calc-sing}
}

\LONGVERSION{
  We use capital Greek letters ($\Gamma,\Delta$) for variables ranging
  over contexts; capital letters from the beginning of the Latin
  alphabet ($A,B$) for variables ranging over types; small Greek
  letters ($\delta,\rho,\sigma$) are used for variables denoting
  substitutions; and minuscule Latin characters ($r,s,t,u,a,b$) for
  variables on terms. Words in sans face denote constants (e.g.,
  $\mathsf{Type},\mathsf{q}$).  
}

\subsubsection{Sorts}
\label{sec:sorts}
The set of sort symbols is $\{\ctx, \into, \mathsf{Type},
\mathsf{Term} \}$ and their formation rules, in the sense of Cartmell's
GATs, are:
\begin{gather*}
  \tyrule{ctx-sort}{ }{\ctx \mbox{ is a type}} 
\qquad
  \tyrule{subs-sort}{\Gamma,\Delta \in \ctx}{\Gamma\into\Delta \mbox{
      is a type}} 
\\[2ex]
  \tyrule{type-sort}{\Gamma \in \ctx }{\type{\Gamma} \mbox{ is a
      type}}
\qquad
  \tyrule{term-sort}{\Gamma \in \ctx\\ A\in \type{\Gamma}}
  {\term{\Gamma}{A} \mbox{ is a type}}
\end{gather*}
%
\EXPLAINREDUNDANT{ In the following, whenever a rule has a hypothesis
  $A \in \type\Gamma$, then $\Gamma \in \ctx$ shall be a further,
  implicit hypothesis.  Similarly, $\sigma \in \Gamma \into \Delta$
  presupposes $\Gamma \in \ctx$ and $\Delta \in \ctx$, and $t \in
  \term\Gamma A$ presupposes $A \in \type\Gamma$, which in turn
  presupposes $\Gamma \in \ctx$.  } Note that judgements of the form
$\Gamma~\in~\ctx$, $A\in \type{\Gamma}$, $t \in \term{\Gamma}{A}$, and
$\sigma \in \Gamma\into\Delta$ correspond to the more conventional
forms $\dctx{\Gamma}$, $\dtype{\Gamma}{A}$, $\dterm{\Gamma}{A}{t}$, and
$\dsubs{\Gamma}{\Delta}{\sigma}$, resp.  After we have defined the
judgements, we will use the latter, more readable versions.

\subsubsection{Operators}
\label{sec:op}
The set of operators is quite large and instead of giving it at once,
we define it as the union of the disjoint sets of operators for
contexts, substitutions, types, and terms.

\para{Contexts}
\label{sec:ctx}
There are the usual two operators for constructing contexts: $S_C = \{\ectx,
\ctxe{\_}{\_}\}$.

\begin{gather*}
  \tyrule{empty-ctx}{ }{\ectx \in \ctx} \qquad
  \tyrule{ext-ctx}{\Gamma \in \ctx\\ A\in \type{\Gamma}}{\ctxe{\Gamma}{A}\in \ctx}
\end{gather*}

\para{Substitutions}
\label{sec:subs}
We have five operators for substitutions, which are the usual
operators for explicit substitutions \cite{p31-abadi}: $S_S = \{
\esubs,\exsubs{\_}{\_}, \idsubs{\_},\subsc{\_}{\_}, \p\}$.
Semantically, substitutions $\sigma \in \Gamma \into \Delta$ are
sequences of values, one for every variable declaration in $\Delta$.
The sequences are constructed from the empty sequence $\esubs$ by
sequence extension $\exsubs \sigma t$.  Substitutions form a category
with identity $\idsubs \Gamma$ and composition $\subsc \sigma \delta$.
Finally, we have the first projection $\p$ on sequences.

\begin{gather*}
  \tyrule{empty-subs%
        }{\Gamma \in \ctx
        }{\esubs\in\Gamma\into\ectx}
\qquad
  \tyrule{ext-subs%
        }{\REDUNDANT{\Gamma,\Delta \in \ctx\\ } 
          \sigma\in\Gamma\into\Delta\\
          \REDUNDANT{A\in \type{\Delta}\\ } 
          t\in \term{\Gamma}{\subsTy{A}{\sigma}}
        }{\exsubs{\sigma}{t} \in\Gamma\into\ctxe{\Delta}{A}}
\end{gather*}
\begin{gather*}
  \tyrule{id-subs%
        }{\Gamma \in \ctx
        }{\idsubs{\Gamma}\in\Gamma\into\Gamma} 
\qquad
  \tyrule{comp-subs%
        }{\REDUNDANT{\Gamma,\Delta,\Theta \in \ctx\\ }
          \delta\in\Gamma\into\Theta\\ 
          \sigma\in\Theta\into\Delta
        }{\subsc{\sigma}{\delta}\in\Gamma\into\Delta}
\\
  \tyrule{fst-subs%
        }{\REDUNDANT{\Gamma \in \ctx\\ } 
          A\in \type{\Gamma}
        }{\p \in \ctxe{\Gamma}{A} \into\Gamma}
\end{gather*}

\para{Types}
\label{sec:types}
The set of operators for types is $S_T =
\{\TmU,\F{\_}{\_},\subsTy{\_}{\_},\singTm{\_}{\_}\}$. 
$\TmU$ is a universe of small types \`a la Russell, which means its
elements are directly usable as types \rulename{u-el} without
coercion. Besides dependent function types $\F A B$ we have the
singleton type $\singTm t A$---a subtype of $A$
containing $t$ as single inhabitant.  Types $A$ are closed under
substitution $\subsTy A \sigma$.

\begin{gather*}
  \tyrule{u-f}{\Gamma \in \ctx }{\TmU \in \type{\Gamma}}  
\quad
  \tyrule{u-el%
        }{\REDUNDANT{\Gamma \in \ctx\\ } 
          A\in \term{\Gamma}{\TmU}
        }{A \in \type{\Gamma}}
\quad
  \tyrule{fun-f%
        }{\REDUNDANT{\Gamma \in \ctx\\ } 
          A\in \type{\Gamma}\\ 
          B\in \type{\ctxe{\Gamma}{A}}
        }{\F{A}{B} \in \type{\Gamma}} 
\\[2ex]
  \tyrule{sing-f%
        }{\REDUNDANT{\Gamma \in \ctx\\ } 
          A\in \type{\Gamma}\\
          t\in \term{\Gamma}{A}
        }{\singTm{t}{A} \in \type{\Gamma}} 
\quad
  \tyrule{subs-type%
        }{\REDUNDANT{\Gamma,\Delta \in \ctx\\ } 
          A\in \type{\Delta}\\ 
          \sigma\in\Gamma\into\Delta
        }{\subsTy{A}{\sigma} \in \type{\Gamma}}
\end{gather*}

\para{Terms}
\label{sec:terms}
The set of operators for terms is $S_E =
\{\F{\_}{\_},\lambda\_,
\appTm{\_}{\_},\subsTm{\_}{\_},\q,\singTm{\_}{\_}\}$. 
It includes function space $\F A B$ and singleton $\singTm t A$ as
small-type constructors in $\TmU$.  Lambda terms with explicit
substitutions are obtained via the constructions $\lambda t$, $\appTm t
u$, $\q$, and $\subsTm t \sigma$.
Since we have used juxtaposition for composition and application of
substitutions, we have the explicit $\mathsf{app}$ for term application.  Note
that $\q$ stands for the top ($0$th) variable, the $n$th variable is
expressed as $\q\,\p^n$. 

\begin{gather*}
  \tyrule{fun-u-i%
        }{\REDUNDANT{\Gamma \in \ctx\\ } 
          A\in \term{\Gamma}{\TmU}\\
          B\in \term{\ctxe{\Gamma}{A}}{\TmU}
        }{\F{A}{B} \in \term{\Gamma}{\TmU}}
\qquad
  \tyrule{fun-i%
        }{\REDUNDANT{
            \Gamma \in \ctx\\ 
            A\in \type{\Gamma}\\ 
            B\in \type{\ctxe{\Gamma}{A}}\\
          } 
          t\in \term{\ctxe{\Gamma}{A}}{B}
        }{\lambda t \in \term{\Gamma}{\F{A}{B}}} 
\\[2ex]
  \tyrule{fun-el%
        }{\REDUNDANT{\Gamma \in \ctx\\ A\in \type{\Gamma} \\ }
          B \in \type{\ctxe{\Gamma}{A}}\\
          t\in \term{\Gamma}{\F{A}{B}}\\ 
          u\in \term{\Gamma}{A} 
        }{\appTm{t}{u}\in \term{\Gamma}{\subsTy{B}{\subid{\Gamma}{u}}}}
\\[2ex] 
    \tyrule{hyp%
          }{\REDUNDANT{\Gamma \in \ctx\\ } 
            A\in \type{\Gamma}
          }{\q \in \term{\ctxe{\Gamma}{A}}{\subsTy{A}{\p}}}
\qquad
  \tyrule{subs-term%
        }{\REDUNDANT{\Gamma,\Delta \in \ctx\\ A\in \type{\Delta}\\ }
          \sigma\in\Gamma\into\Delta\\ t\in \term{\Delta}{A}
        }{\subsTm{t}{\sigma} \in \term{\Gamma}{\subsTy{A}{\sigma}}}
\end{gather*}
\begin{gather*}
    \tyrule{sing-u-i%
          }{\REDUNDANT{\Gamma \in \ctx\\ } 
            A\in \term{\Gamma} {\TmU}\\ 
            t\in \term{\Gamma}{A}
          }{\singTm{t}{A} \in \term{\Gamma}{\TmU}}
\qquad
    \tyrule{sing-i%
          }{\REDUNDANT{\Gamma \in \ctx\\ A\in \type{\Gamma}\\ }
            t \in \term{\Gamma}{ A}\\
          }{t \in \term{\Gamma}{ \singTm{t}{A}}} 
\\[2ex]
    \tyrule{sing-el%
          }{\REDUNDANT{\\Gamma \in \ctx\\ A\in \type{\Gamma}\\ }
            a \in \term{\Gamma}{A}\\
            t\in \term{\Gamma}{\singTm{a}{A}}
          }{t \in \term{\Gamma}{A}}
\end{gather*}

\subsubsection{Axioms for Equational Theory}
\label{sec:ax}

In the following, we present the axioms of the equational theory of
$\lambdaSing$.  Equality is considered as the congruence closure of
these axioms.  Congruence rules, also called derived rules, are
generated mechanically for each symbol from its typing.  For instance,
rule (\textsc{subs-type}) induces the derived rule
\begin{equation*}
   \inferrule{A = B \in \type{\Gamma}\\ \gamma = \delta \in
  \Delta\into\Gamma}{\subsTy{A}{\gamma} = \subsTy{B}{\delta}\in
  \type{\Delta}}\label{eq:example} \enspace .
\end{equation*}
Another instance of a derived rule is conversion, it holds because
equality between sorts, such as $\term \Gamma A = \term \Gamma {A'}$:
\[
  \ru{t \in \term \Gamma A \qquad
      A = A' \in \type \Gamma
    }{t \in \term \Gamma {A'}}
\]\medskip

\noindent  In the following, we present equality axioms 
without the premises concerning typing, 
except in the cases where they cannot be inferred.

\para{Substitutions}  The first two equations witness extensionality
for the identity substitution, the next three the composition laws for
the category of substitutions.  Then there is a law for the first
projection $\p$, and the last two laws show how to propagate a substitution
$\delta$ into a tuple.
\label{sec:ax-subs}
\begin{align*}
  \idsubs{\ectx} &= \esubs& 
  \idsubs{\ctxe{\Gamma}{A}} & =  \exsubs{\p}{\q}& 
\\
  \subsc{\idsubs{}}{\sigma} &=  \sigma & 
  \subsc{\sigma}{\idsubs{}} &=  \sigma & 
\\
  \subsc{(\subsc{\sigma}{\delta})}{\gamma} &=  \subsc{\sigma}{(\subsc{\delta}{\gamma})} & 
  \subsc{\p}{\exsubs{\sigma}{t}} & =  \sigma & 
\\
  \subsc{\esubs}{\delta} &=  \esubs& 
  \subsc{\exsubs{\sigma}{t}}{\delta} &=
    \exsubs{\subsc{\sigma}{\delta}}{\subsTm{t}{\delta}}&
\end{align*}

\para{Axioms for $\beta$ and $\eta$, propagation and resolution of
  substitutions}
An explicit substitution $\subid{\Gamma}{r}$ is created by contracting
a $\beta$-redex (first law).  It is then propagated into the various term
constructions until it can be resolved (last two laws).
\begin{align*}
  \appTm{(\lambda t)}{r} &=  \subsTm{t}{\subid{\Gamma}{r}} &
  \lambda (\appTm{(\subsTm{t}{\p})}{\q})& =  t &
\\
  \subsTy{\TmU}{\sigma} &=  \TmU& 
  \subsTy{(\singTm{t}{A})}{\sigma} &=
    \singTm{\subsTm{t}{\sigma}}{\subsTy{A}{\sigma}}&
\\
  \subsTy{(\F{A}{B})}{\sigma} &=
    \F{(\subsTy{A}{\sigma})}{(\subsTy{B}{\exsubs{\subsc{\sigma}{\p}}}{\q})} & 
  \subsTm{(\lambda t)}{\sigma} &=  \lambda (\subsTm{t}
    {\exsubs{\subsc{\sigma}{\p}}{\q}})& 
\\
  \subsTm{(\appTm{r}{s})}{\sigma} &=
    \appTm{(\subsTm{r}{\sigma})}{(\subsTm{s}{\sigma})}&
  \subsTm{(\subsTm{t}{\delta})}{\sigma}& =
    \subsTm{t}{(\subsc{\delta}{\sigma})}  &
\\
  \subsTm{\q}{\exsubs{\sigma}{t}}&=  t& 
  \subsTm{t}{\idsubs{}} &=  t & 
  \end{align*}

\para{Singleton types} All inhabitants of a singleton type are equal
\rulename{sing-eq-i}.  We mention the important derived rule
$\rulename{sing-eq-el}$ here explicitly. 
\begin{gather*}
  \tyrule{sing-eq-i}{t,t'\in \term{\Gamma}{\singTm{a}{A}}}
  {t = t'  \in \term{\Gamma}{\singTm{a}{A}}}
\qquad
  \tyrule{sing-eq-el}{t = t' \in \term{\Gamma}{\singTm{a}{A}}}
  {t = t'  \in \term{\Gamma}{A}}
\end{gather*}
\LONGVERSION{There is a choice how to express the last two rules;
  they could be replaced with
\begin{center}
  \begin{tabular}{lr}
    $\tyrule{sing-eq-i'}{t \in
      \term{\Gamma}{\singTm{a}{A}}}{t=a\in\term{\Gamma}{\singTm{a}{A}}} $ &
    $\tyrule{sing-eq-el'}{t \in
      \term{\Gamma}{\singTm{a}{A}}}{t=a\in\term{\Gamma}{A}} $ \\
  \end{tabular}
\end{center}
The rule \ruleref{sing-eq-el} is essential; in fact, since we have
eta-expansion for singletons, we would like to derive
\begin{equation*}
  \label{eq:eq-fst-q}
  \deqterm{\ctxe{\Gamma}{\singTm{\lambda t}{\F{A}{B}}}}{\subsTy{B}{\subid{}{a}}}
  {\appTm{\q}{a}}{\subsTm{t}{\subid{}{a}}}
\end{equation*}
from $ \deqterm{\ctxe{\Gamma}{\singTm{\lambda t}{\F{A}{B}}}}
{\singTm{\lambda t}{\F{A}{B}}}{\q}{\lambda t}$, and
$\dterm{\Gamma}{A}{a}$. Which would be impossible if
\ruleref{sing-eq-el} were not a rule.


}

\para{Conventions}
We denote with $|\Gamma|$ the length of the context $\Gamma$; and
$\Gamma!i$ is the projection of the $i$-th component of $\Gamma$, for
$0\leqslant i < |\Gamma|$; i.e. if $\Gamma=A_{n-1}\ldots A_{0}$ and
$0 \leqslant i < n$, then $\Gamma!i = A_{i}$. We say
$\Delta\leqslant^{i}\Gamma$ if $\dsubs{\Delta}{\Gamma}{\p^{i}}$; where
$\p^i$ is the $i$-fold composition of $\p$ with itself.  
\footnote{The direction
$\Delta\leqslant^{i}\Gamma$ (as opposed to
$\Delta\geqslant^{i}\Gamma$) has been chosen to be compatible with
subtyping $A \leqslant B$.  Weakening (Remark~\ref{rem:weak}) is a
special case of subsumption which states that $\Delta \leqslant \Gamma
\derN t : A \leqslant B$ implies $\Delta \derN t : B$.}

We denote with $\terms$ the set of words freely generated using symbols in
$S_S\cup S_T\cup S_E$.  We write $t \equiv_T t'$ for denoting
syntactically equality of $t$ and $t'$ in $T \subseteq \terms$. We
call $A$ the {\em tag} of $\singTm{a}{A}$.

\LONGVERSION{
\begin{rem}
  \label{rem:lift}
  Note that if $\dsubs{\Delta}{\Gamma}{\p^i}$, and
  $\dsubs{\Gamma}{\Theta}{\p^j}$, then
  $\dsubs{\Delta}{\Theta}{\p^{i+j}}$.
\end{rem}
}
\begin{defi}[de Bruijn index] The $i$th de~Bruijn index $\ind i$
  is defined as
\[
  \ind i = \left\{
    \begin{array}{ll}
      \q     & \mbox{ if } i \leq  0 \\
      \q\p^i & \mbox{ if } i > 0 . \\
    \end{array}
\right.
\]
For convenience, we identify negative indices with the 0th index.
\end{defi}
The following grammar describes the set $\nfterms$ of $\beta$-normal
forms.  As auxiliary notion, 
it uses the set $\neterms$ of neutral normal forms, \ie,
normal forms with a variable in head position, which blocks reduction.
A bit sloppily, we refer to elements of $\neterms$ as ``neutral
terms''; in general, the attribute \emph{neutral} shall mean
\emph{variable in head position} (this is stricter than Girard's
concept of neutral \cite{girardLafontTaylor:proofsAndTypes}).
\begin{defi}[Neutral terms, and normal forms]
  \label{def:nfterms}
  \begin{align*}
    \neterms \ni k &::= \ind i 
      \mid \appTm{k}{v}&\\
    \nfterms \ni v, V,W &::= \TmU \mid \F{V}{W} \mid \singTm{v}{V}
    \mid \lambda v\mid k&
  \end{align*}
\end{defi}\medskip

\LONGVERSION{\noindent An advantage of introducing the calculus as a GAT is
  that we can derive several syntactical results from the meta-theory
  of GATs; for instance, some of the following inversion results,
  which are needed in the proof of completeness of the type-checking
  algorithm.  }
\begin{rem}[Weakening of judgements] 
  \label{rem:weak}
  Let $\Delta\leqslant^i\Gamma$, $\deqtype{\Gamma}{A}{A'}$, and
  $\deqterm{\Gamma}{A}{t}{t'}$; then
  $\deqtype{\Delta}{\lift{i}{A}}{\lift{i}{A'}}$, and
  $\deqterm{\Delta}{\lift{i}{A}}{\lift{i}{t}}{\lift{i}{t'}}$.
\end{rem}

\begin{rem}[Syntactic validity]
  \label{rem:invctx}\hfill
  \begin{enumerate}[\em(1)]
  \item If $\dterm{\Gamma}{A}{t}$, then $\dtype{\Gamma}{A}$.
  \item If $\deqterm{\Gamma}{A}{t}{t'}$, then both
    $\dterm{\Gamma}{A}{t}$, and $\dterm{\Gamma}{A}{t'}$.
  \item If $\deqtype{\Gamma}{A}{A'}$, then both $\dtype{\Gamma}{A}$,
    and $\dtype{\Gamma}{A'}$.
  \end{enumerate}
\end{rem}

\begin{lem}[Inversion of types]
  \label{lem:inv-ty}\hfill
  \begin{enumerate}[\em(1)]
  \item If $\dtype{\Gamma}{\F{A}{B}}$, then $\dtype{\Gamma}{A}$, and
    $\dtype{\ctxe{\Gamma}{A}}{B}$.
  \item If $\dtype{\Gamma}{\singTm{a}{A}}$, then $\dtype{\Gamma}{A}$,
    and $\dterm{\Gamma}{A}{a}$.
  \item If $\dtype{\Gamma}{k}$, then $\dterm{\Gamma}{\TmU}{k}$.
  \end{enumerate}
\end{lem}

\SHORTVERSION{
\begin{lem}\label{lemma:inv-tm}\hfill
  \begin{enumerate}
  \item If $\dterm{\Gamma}{A}{\F{A'}{B'}}$, then
    $\dterm{\Gamma}{\TmU}{A'}$, and also
    $\dterm{\ctxe{\Gamma}{A'}}{\TmU}{B'}$;
  \item If $\dterm{\Gamma}{A}{\singTm{b}{B}}$, then
    $\dterm{\Gamma}{\TmU}{B}$, and also $\dterm{\Gamma}{B}{b}$;
  \item If $\dterm{\Gamma}{A}{\lambda t}$, then 
    $\dterm{\ctxe{\Gamma}{A'}}{B'}{t}$.
  \item If $\dterm{\Gamma}{\singTm{a}{A}}{t}$, then 
    $\dterm{\Gamma}{A}{t}$, and $\deqterm{\Gamma}{A}{t}{a}$.
  \item If $\dterm{\Gamma}{A}{\lift{i}{\q}}$, then either
    $\deqtype{\Gamma}{A}{\lift{i+1}{(\Gamma!i)}}$; or
    $\deqtype{\Gamma}{A}{\singTm{a}{A'}}$, and 
      $\deqterm{\Gamma}{A'}{a}{\lift{i}{\q}}$.
  \end{enumerate}
\end{lem}
}

\LONGVERSION{

  The following lemma can be proved directly by induction on
  derivations by checking the possibles rules used in the last step.

\begin{lem}[Inversion of typing]
  \label{lem:inv-tm}\label{lem:invneut}\label{L:invtm}\hfill
  \begin{enumerate}[\em(1)]
  \item If $\dterm{\Gamma}{A}{\F{A'}{B'}}$, then
    $\deqtype{\Gamma}{A}{\TmU}$, 
    $\dterm{\Gamma}{\TmU}{A'}$, and also
    $\dterm{\ctxe{\Gamma}{A'}}{\TmU}{B'}$;
  \item if $\dterm{\Gamma}{A}{\singTm{b}{B}}$, then 
    $\deqtype{\Gamma}{A}{\TmU}$, 
    $\dterm{\Gamma}{\TmU}{B}$, and also
    $\dterm{\Gamma}{B}{b}$;
  \item if $\dterm{\Gamma}{A}{\lambda t}$, then either
    \begin{enumerate}[\em(a)]
    \item $\deqtype{\Gamma}{A}{\F{A'}{B}}$ with $\dterm{\ctxe{\Gamma}{A'}}{B}{t}$; or
    \item $\deqtype{\Gamma}{A}{\singTm{a}{A'}}$ with 
      $\deqterm{\Gamma}{A'}{\lambda t}{a}$.
    \end{enumerate}
  \item if $\dterm{\Gamma}{A}{\appTm{t}{r}}$, then 
    $\dterm{\Gamma}{\F{A'}{B'}}{t}$, $\dterm{\Gamma}{A'}{r}$,
    and $\deqtype{\Gamma}{A}{\subsTy{B'}{\exsubs{\idsubs{}}{r'}}}$.
  \end{enumerate}
\end{lem}

\proof\hfill
(1) The last rule used is one of \ruleref{fun-u-i}, \ruleref{conv},
  \ruleref{sing-i}, or \ruleref{sing-e}. In the first case the
  premises of the rule are what is to be proved; in all other cases
  we have a premise with the form $\dterm{\Gamma}{B}{\F{A'}{B'}}$,
  hence we can apply the i.h.
(2-4) Analogously.
\qed  
}




\begin{rem}[Inversion of substitution]
  \label{invsubs}
  Any substitution $\dsubs{\Delta}{\ctxe{\Gamma}{A}}{\sigma}$ is equal
  to some substitution
  $\dsubs{\Delta}{\ctxe{\Gamma}{A}}{\exsubs{\sigma'}{t}}.$ It
  is enough to note $\idsubs{\ctxe{\Gamma}{A}} = \exsubs{\p}{\q}$, hence we have the
  equalities $\sigma = \subsc{\idsubs{}}{\sigma}
  =
  \subsc{\exsubs{\p}{\q}}{\sigma}
  =
  \exsubs{\subsc{\p}{\sigma}}{\subsTm{\q}{\sigma}}.$
\end{rem}

\LONGSHORT{\subsection{\lambdaPI: 
A type theory with  proof-irrelevance.}}
  {\paragraph*{{\bf Calculus with Proof-Irrelevance.}}}
\label{sec:pi-calc}

\LONGVERSION{ In this section we keep the basic rules of the previous
  calculus (those that do not refer to singleton types), and introduce
  types for natural numbers, enumeration sets, sigma types, and proof-irrelevant
  types.  The main difference with other presentations 
  \cite{nordstroem:mltt,mlitt}, on the syntactic level, is that the
  eliminator operator (for each type) has as an argument the type of
  the result. The presence of the resulting type in the eliminator is
  needed in order to define the normalisation function; it is also
  necessary for the type-inference algorithm.  }




\input{extrules}

\noindent After exposition of the formation, introduction, elimination, and
equality rules for the types of \lambdaPI, we continue with basic
properties of derivations.  From now, we use the more conventional
notation for judgements.

\LONGSHORT{
  \begin{defi}[Neutral terms and normal forms]\hfill
    \begin{align*}
      \neterms \ni k ::= &\ldots \mid \dfst{k} \mid \dsnd{k} \mid
        \natrec{V}{v}{v'}{k} 
        \mid \elim{n}{V}{v_0 \cdots v_{n-1}}{k} \mid \wheretm{v}{k}{V} \mid \oprf &\\
      \nfterms \ni v,V ::= &\ldots \mid \DSum{V}{W}\mid \natty\mid
      \enum{n} \mid \boxty{V}\mid \depair{v}{v'} 
       \mid \ztm \mid \suctm{v} \mid 
      \const{n}{i} \mid \boxtm{v} &
    \end{align*}
  \end{defi}
}{
As is expected we have now more normal forms, and more neutral terms:
\begin{align*}
  \neterms \ni k &::= \ldots \mid \wheretm{v}{k}{V} &\\
  \nfterms \ni v,V &::= \ldots \mid \boxty{V}\mid \boxtm{v} \mid \oprf &
\end{align*}
}

\LONGSHORT{
\begin{lem}[Inversion of types]\hfill
  \label{lem:inv-pi} 
  \begin{enumerate}[\em(1)]
  \item If $\dtype{\Gamma}{\DSum{A}{B}} $, then $\dtype{\Gamma}{A}$,
    and $\dtype{\ctxe{\Gamma}{A}}{B}$.
  \item If $\dtype{\Gamma}{\boxty{A}} $, then $\dtype{\Gamma}{A}$.
  \end{enumerate}
\end{lem}
\begin{lem}[Inversion of typing]\hfill
  \label{lem:inv-pi-tm} 
  \begin{enumerate}[\em\phantom0(1)]
  \item If $\dterm{\Gamma}{A}{\DSum{A'}{B}} $, then
    $\deqtype{\Gamma}{A}{\TmU}$, and
    $\dterm{\Gamma}{\TmU}{A'}$, and
    $\dterm{\ctxe{\Gamma}{A'}}{\TmU}{B}$.
  \item If $\dterm{\Gamma}{A}{\natty} $, then
    $\deqtype{\Gamma}{A}{\TmU}$.
  \item If $\dterm{\Gamma}{A}{\enum{n}} $, then
    $\deqtype{\Gamma}{A}{\TmU}$.
  \item If $\dterm{\Gamma}{A}{\depair{t}{b}}$ , then
    $\deqtype{\Gamma}{A}{\DSum{A'}{B}}$, and 
    $\dterm{\Gamma}{A'}{t}$, and
    $\dterm{\Gamma}{\subsTy{B}{\subid{}{t}}}{b}$.
  \item If $\dterm{\Gamma}{A}{\dfst{t}}$ , then
    $\deqtype{\Gamma}{A}{A'}$, and $\dterm{\Gamma}{\DSum{A'}{B}}{t}$,
    for some $A'$, and $B$.
  \item If $\dterm{\Gamma}{B}{\dsnd{t}}$, then
    $\deqtype{\Gamma}{B}{\subsTy{B'}{\subid{}{\dfst{t}}}}$, and 
    $\dterm{\Gamma}{\DSum{A}{B'}}{t}$, for some $A$,
    and $B'$.
  \item If $\dterm{\Gamma}{A}{\ztm}$, then
    $\deqtype{\Gamma}{A}{\natty}$.
    \item If $\dterm{\Gamma}{A}{\suctm{t}}$, then 
      $\dterm{\Gamma}{\natty}{t}$, and
      $\deqtype{\Gamma}{A}{\natty}$.
    \item If $\dterm{\Gamma}{A}{\natrec{B}{z}{s}{t}}$ , then
      $\dtype{\ctxe{\Gamma}{\natty}}{B} $,
      $\dterm{\Gamma}{\subsTy{B}{\subid{}{\ztm}}}{z}$,
      $\dterm{\Gamma}{\recty{B}}{s}$, $\dterm{\Gamma}{\natty}{t}$,
      and $\deqtype{\Gamma}{A}{\subsTy{B}{\subid{}{t}}}$.
   \item if $\dterm{\Gamma}{A}{\const{n}{i}}$, then    
      $\deqtype{\Gamma}{A}{\enum{n}}$;
    \item If $\dterm{\Gamma}{A}{\elimn{B}{\vec{t}}{t'}}$, then
      $\dtype{\ctxe{\Gamma}{\enum{n}}}{B} $,
      $\dterm{\Gamma}{\subsTy{B}{\subid{}{\constn{i}}}}{t_i}$,
       $\dterm{\Gamma}{\enum{n}}{t'}$, and
      $\deqtype{\Gamma}{A}{\subsTy{B}{\subid{}{t}}}$.  
  \item If $\dterm{\Gamma}{A}{\boxtm{t}}$, then
    $\deqtype{\Gamma}{A}{\boxty{A'}}$ and
    $\dterm{\Gamma}{A'}{t'}$.
  \item If $\dterm{\Gamma}{A}{\wheretm{b}{t}{B}}$, then
    $\deqtype{\Gamma}{A}{B}$, 
    $\dterm{\Gamma}{\boxty{A'}}{t}$ for some $A'$, 
    $\dterm{\ctxe{\Gamma}{A'}}{\subsTy{B}{\p}}{b}$, and
    $\deqterm{\ctxe{\ctxe{\Gamma}{A'}}{\subsTy{A'}{\p}}}
    {\subsTy{B}{\p}}{\subsTm{b}{\p}}{\subsTm{b}{\exsubs{\p\p}{\q}}}$.
    \end{enumerate}
  \end{lem}

} 
{
\begin{lem}[Inversion]\hfill
  \label{lem:inv-pi} 
  \begin{enumerate}[\em(1)]

  \item If $\dterm{\Gamma}{A}{\boxtm{t}}$, then
    $\deqtype{\Gamma}{A}{\boxty{A'}}$ and
    $\dterm{\Gamma}{A'}{t'}$.

  \item If $\dterm{\Gamma}{A}{\wheretm{b}{t}{B}}$, then
    $\deqtype{\Gamma}{A}{\boxty{B}}$, and
    $\dterm{\Gamma}{\boxty{A'}}{t}$, and
    $\dterm{\ctxe{\Gamma}{A'}}{\subsTy{B}{\p}}{b}$.
    \end{enumerate}
  \end{lem}
}

\LONGSHORT{

\subsection{Conservativity of \texorpdfstring{$\oprf$}{*}}
\label{sec:cons-vdashp}

In this section we prove that $(\vdashp)$ is a \emph{conservative}
extension of $(\vdash)$; i.e.,  any derivation in $(\vdashp)$ has a
counterpart derivation in $(\vdash)$ and the components of the
conclusions of those derivations are judgmentally equal in
$(\vdashp)$.


  


  \begin{defi}\label{D:lifting} 
    A term is called {\em pure} if it does not contain any occurrence
    of $\oprf$.  Let $\nu$ be a syntactical entity, if $\mu$ is
    obtained from $\nu$ by replacing all occurrences of $\oprf$ by
    pure terms, then $\mu$ is called a {\em lifting} of $\nu$.
  \end{defi}
  
  We will distinguish those liftings that are judgmentally equal to
  the lifted entity, these liftings are called {\em good liftings}.

  \begin{defi}[Good lifting]\label{D:good-lift}\bla
    \begin{enumerate}[(1)]
    \item A context $\dctx{\Gamma'}$ is a good lifting of $\sdctx{\Gamma}$
      if $\Gamma'$ is a lifting of $\Gamma$, such that $\vdashp \Gamma = \Gamma'$.
    \item A substitution $\dsubs{\Gamma'}{\Delta'}{\sigma'}$ is a good
      lifting of $\sdsubs{\Gamma}{\Delta}{\sigma}$ if $\dctx{\Gamma'}$
      and $\dctx{\Delta'}$ are good liftings of $\sdctx{\Gamma}$ and
      $\sdctx{\Delta}$, resp., and $\sigma'$ is a lifting of $\sigma$,
      such that $\sdeqsubs{\Gamma}{\Delta}{\sigma}{\sigma'}$.
    \item A type $\dtype{\Gamma}{A'}$ is a good lifting of
      $\sdtype{\Gamma}{A}$ if $\dctx{\Gamma'}$ is a good lifting of
      $\sdctx{\Gamma}$ and $A'$ is a lifting of $A$, such that
      $\sdeqtype{\Gamma}{A}{A'}$.
    \item A term $\dterm{\Gamma}{A'}{t'}$ is a good lifting of
      $\sdterm{\Gamma}{A}{t}$ if $\dtype{\Gamma'}{A'}$ is a good
      lifting of $\sdtype{\Gamma}{A}$ and $t'$ is a lifting of $t$,
      such that $\sdeqterm{\Gamma}{A}{t}{t'}$.
    \end{enumerate}
  \end{defi}\medskip

  \noindent Now we can prove that there is a good lifting for each syntactic 
  entity; for proving this, we need the stronger condition that 
  any pair of good liftings for some entity are judgmentally equal.

  \begin{thm}\label{T:good-lifting}\hfill
    \begin{enumerate}[\em(1)]
    \item Let $\sdctx{\Gamma}$; then there is a good lifting 
      $\dctx{\Gamma'}$ of $\sdctx{\Gamma}$; moreover if
      $\dctx{\Gamma''}$ is also a good lifting of $\sdctx{\Gamma}$
      then $\deqctx{\Gamma'}{\Gamma''}$.
    \item Let $\sdsubs{\Gamma}{\Delta}{\sigma}$; then there is a good
      lifting $\dsubs{\Gamma'}{\Delta'}{\sigma'}$ of
      $\sdsubs{\Gamma}{\Delta}{\sigma}$; moreover if
      $\dsubs{\Gamma''}{\Delta''}{\sigma''}$ is also a good lifting of
      $\sdsubs{\Gamma}{\Delta}{\sigma}$ then $\deqctx{\Gamma'}{\Gamma''}, \deqctx{\Delta'}{\Delta''}$,
      and $\deqsubs{\Gamma'}{\Delta'}{\sigma'}{\sigma''}$.
    \item Let $\sdtype{\Gamma}{A}$; then there is a good lifting
      $\dtype{\Gamma'}{A'}$ of $\sdtype{\Gamma}{A}$ ; moreover if
      $\dtype{\Gamma''}{A''}$ is also a good lifting of
      $\sdtype{\Gamma}{A}$ then $\deqctx{\Gamma'}{\Gamma''}$ and
      $\deqtype{\Gamma'}{A'}{A''}$.
    \item Let $\sdterm{\Gamma}{A}{t}$; then there is a good lifting
      $\dterm{\Gamma'}{A'}{t'}$ of $\sdterm{\Gamma}{A}{t}$ ; moreover
      if $\dterm{\Gamma''}{A''}{t''}$ is also a good lifting of
      $\sdterm{\Gamma}{A}{t}$ then $\deqctx{\Gamma'}{\Gamma''}$,
      $\deqtype{\Gamma'}{A'}{A''}$, and $\deqterm{\Gamma'}{A'}{t'}{t''}$.
    \end{enumerate}
  \end{thm}

  \begin{proof}
    By induction on derivations, in each rule we use i.h., and build
    up the corresponding entity to the good lifting for each part of
    the judgement; then, given any other good lifting of the whole
    judgement, we do inversion on the definition of good lifting, and
    get the equalities for each part; we finish using congruence for
    showing that both good lifting are judgmental equal.

  We show the case for \ruleref{prf-tm}. First we prove the existence of
  a good lifting.
  \begin{align*}
    \proofLine{\sdterm{\Gamma}{\boxty{A}}{\oprf}}{hypothesis}\tag{*}\label{eq:lift-ex-1}\\
    \proofLine{\sdterm{\Gamma}{A}{t}}{by inversion on \eqref{eq:lift-ex-1}}
      \tag{\textdagger}\label{eq:lift-ex-2} \\
    \proofLine{\dterm{\Gamma'}{A'}{t'}}{by ind. hyp. is a good lifting of \eqref{eq:lift-ex-2}} \\
    \proofLine{\dterm{\Gamma'}{\boxty{A'}}{\boxtm{t'}}}{by \ruleref{prf-i}, is a good lifting of\eqref{eq:lift-ex-1}}      
  \end{align*}

  Now we prove the second half of the theorem.
  \begin{align*}
    \proofLine{\dterm{\Gamma''}{B}{s}}{hypothesis, be other good lifting of \eqref{eq:lift-ex-2}}\tag{**}\label{eq-lift-ex-3}\\
    \proofLine{\dtype{\Gamma''}{B}}{by inversion, good lifting of $\sdtype{\Gamma}{\boxty{A}}$}\\
    \proofLine{\deqtype{\Gamma'}{B''}{\boxty{A'}}}{by ind. hyp.}\\
    \proofLine{\deqterm{\Gamma'}{\boxty{A'}}{\boxtm{t'}}{s}}{by \ruleref{prf-eq} and \ruleref{conv}}.
  \end{align*}
\end{proof}
}{
Now we prove that the calculus with \ruleref{prf-tm} is a
conservative extension of the one without it.  We decorate the
turnstile, and the equality symbol with $^\ast$ for referring to
judgements in the extended calculus.

\begin{defi}
  A term $t'$ is called a {\em lifting} of a term $t$, if all the
  occurrences of $\oprf$ in $t$ have been replaced by terms
  $s_0,\ldots,s_{n-1}$, and $\oprf$ does not occur in any $s_i$. We
  extend this definition to substitutions, contexts, and equality
  judgements.

  If $\Gamma'$ is a lifting of $\Gamma$, and $\Gamma \eqp \Gamma'$,
  and also $\Gamma'\vdash$ then we say that $\Gamma'$ is a {\em
    good lifting} of $\Gamma$.  We extend the definition of
  good lifting to the others kinds of judgement.
\end{defi}

\begin{lem}
  Let $\Gamma \vdashp J$, then there exists a good lifting $\Gamma'
  \vdash J'$; moreover for any other good lifting $\Gamma''\vdash J''$
  of $\Gamma\vdashp J$, we have $\Gamma' = \Gamma''$, and 
  $\Gamma'\vdash J' = J''$.
\end{lem}
}

\begin{cor}
  \label{cor:cons-prf-tm}
  The calculus $(\vdashp)$ is a conservative extension of $(\vdash)$. \qed
\end{cor}

\para{Combining singleton types and proof-irrelevant propositions}
For illustrating the difficulties one can find when extending
\lambdaIrr\ with singleton types, consider a slightly different
calculus where we drop the type annotation of the eliminator for
proof-irrelevance terms; i.e. we would have $\wheretm{b}{t}{\ }$
instead of $\wheretm{b}{t}{B}$. In the resulting system one can
derive: \newcommand{\czero}{\const{2}{0}}
\newcommand{\cone}{\const{2}{1}} \newcommand{\TyN}{\enum{2}}
\newcommand{\TyZero}{\singTm{\czero}{\TyN}}
\newcommand{\TyOne}{\singTm{\cone}{\TyN}}

\[
\begin{prooftree}
  \[
  \[
  \[
  \[
  \[ \justifies
  \derN\czero\of\TyZero
  \]
  \justifies
  \derN\oprf\of\boxty{\TyZero}
  \]
  \justifies
  \derN \ABwhere x x \oprf \of \TyZero
  \]
  \justifies
  \derN \ABwhere x x \oprf = \czero \of \TyZero
  \]
  \justifies
  \derN \ABwhere x x \oprf = \czero \of \TyN
  \]\quad
  \[
  \[
  \[
  \[
  \[ \justifies
  \derN\cone\of\TyOne
  \]
  \justifies
  \derN\oprf\of\boxty{\TyOne}
  \]
  \justifies
  \derN \ABwhere x x \oprf \of \TyOne
  \]
  \justifies
  \derN \ABwhere x x \oprf = \cone \of \TyOne
  \]
  \justifies
  \derN \ABwhere x x \oprf = \cone \of \TyN
  \]
  \justifies
  \derN \czero = \cone \of \TyN
\end{prooftree}
\]
This derivation shows that mixing the rule \ruleref{sing-eq-el} with
erasure of proof-terms leads to inconsistencies. It is yet unclear how
to combine singleton types and erasure of proof-terms; we leave this
topic for a future work. On the other hand, there are no problems in
extending $(\vdash)$ with singletons types; in fact, we can construct
(see Rem. \ref{rem:proof-relevant-model}) a model where
$\semc{\czero}\not=\semc{\cone}$, which assures $\not\derN
\czero=\cone \of \TyN$.



%% file: extrules.tex

\para{Sigma types}  
Both $\TmU$ and $\mathsf{Type}$ are closed under (strong) sigma-type formation;
$\depair a b$ introduces a dependent pair and $\dfst t$ and $\dsnd t$
eliminate it.

\begin{gather*}
\tyrule{sum-u-i%
  }{A\in \term{\Gamma}{\TmU}\\ B\in \term{\ctxe{\Gamma}{A}}{\TmU}}
  {\DSum{A}{B}\in \term{\Gamma}{\TmU}}
\
\tyrule{sum-f%
  }{\REDUNDANT{\Gamma\in \ctx\\ }A \in \type{\Gamma} \\ 
    B \in \type{\ctxe{\Gamma}{A}}}
  {\DSum{A}{B}\in \type{\Gamma}}
\\[1ex]  
\tyrule{sum-in%
  }{
    \REDUNDANT{\Gamma \in \ctx\\ A\in \type{\Gamma}\\ }B\in \type{\ctxe{\Gamma}{A}}\\
      a \in \term{\Gamma}{A}\\ 
      b \in \term{\Gamma}{\subsTy{B}{\subid{}{a}}}}
    {\depair{a}{b} \in \term{\Gamma}{\DSum{A}{B}}}
\\[1ex]
  \tyrule{sum-el1%
  }{
    \REDUNDANT{\Gamma \in \ctx\\ A\in \type{\Gamma}\\ B\in \type{\ctxe{\Gamma}{A}}\\}
    t \in \term{\Gamma}{\DSum{A}{B}}}{\dfst{t} \in \term{\Gamma}{A}}
\qquad
  \tyrule{sum-el2%
  }{
    \REDUNDANT{\Gamma \in \ctx\\ A\in \type{\Gamma}\\ B\in \type{\ctxe{\Gamma}{A}}\\}
    t \in \term{\Gamma}{\DSum{A}{B}}}
  {\dsnd{t} \in \term{\Gamma}{\subsTy{B}{\subid{}{\dfst{t}}}}
  }
\end{gather*}

The $\beta$- and $\eta$-laws for pairs are given by the first three
equations to follow.  The remaining equations propagate substitutions
into the new term constructors.
\begin{align*}
  \dfst{\depair{a}{b}} &= a &
  \dsnd{\depair{a}{b}} &= b &
  \depair{\dfst{t}}{\dsnd{t}} &= t
\\
  \subsTm{(\dfst{t})}{\sigma} &= \dfst{(\subsTm{t}{\sigma})}&
  \subsTm{(\dsnd{t})}{\sigma} &= \dsnd{(\subsTm{t}{\sigma})}&
  \subsTm{\depair{a}{b}}{\sigma}  &=
  \depair{\subsTm{a}{\sigma}}{\subsTm{b}{\sigma}} &
\\
& &
    \subsTy{(\DSum{A}{B})}{\sigma} &=  \DSum{(\subsTy{A}{\sigma})}
                       {{(\subsTy{B}{\exsubs{\subsc{\sigma}{\p}}{\q}})}\hspace*{-6ex}} &
\end{align*}
Propagation laws can be obtained mechanically: to propagate $\sigma$
into $\mathsf{c}\,\vec t$, just compose it with each $t_i$ that is not
a binder (e.g., $A$ in $\DSum A B$), and compose its lifted version
$\exsubs{\subsc{\sigma}{\p}}{\q}$ with each $t_j$ that is a binder
(e.g., $B$ in $\DSum A B$).  Binders are those formed in an extended context
(here, $B \in \type{\ctxe \Gamma A}$).   In the following, we will
skip the propagation laws.

\para{Natural numbers}  We add an inductive type $\natty$ with
constructors $\ztm$ and $\suctm{}$ and primitive recursion
$\mathsf{natrec}$.
\begin{gather*}
  \tyrule{nat-u-i%
  }{\Gamma\in\ctx}
  {\natty\in \term{\Gamma}{\TmU}}
\
  \tyrule{nat-z-i%
  }{
    \Gamma \in \ctx}{\ztm \in \term{\Gamma}{\natty}}
\
  \tyrule{nat-s-i%
  }{
    \REDUNDANT{\Gamma \in \ctx\\} t\in \term{\Gamma}{\natty}
  }{\suctm{t} \in \term{\Gamma}{\natty} }
\\[1ex]
    \tyrule{nat-el%
  }{
    \REDUNDANT{\Gamma \in \ctx\\}
    B\in \type{\ctxe{\Gamma}{\natty}}
      \\ t \in \term{\Gamma}{\natty}
      \\ z\in \term{\Gamma}{\subsTy{B}{\subid{}{\ztm}} }
      \\ s\in \term{\Gamma}{\recty{B}}
  }{
    \natrec{B}{z}{s}{t}\in \term{\Gamma}{\subsTy{B}{\subid{}{t}}}}
\end{gather*}
Here, we used $\recty{B}$ as an abbreviation for
  $\F{\natty}{(\F{B}{\,(\subsTy{\subsTy{B}{\exsubs{\p}{\suctm{\q}}}}{\p})})}$
  which in conventional notation reads $\Pi x\!:\! \natty.\, B \to
  B[\suctm x/x]$.
Since $B$ is a big type, it can mention the universe $\TmU$, thus, we
can define small types by recursion via $\mathsf{natrec}$.  This so
called \emph{large elimination} excludes normalization
proofs which use induction on type expressions 
\cite{crary:lfmtp08,courant:itrs02}.  We add the usual computation
laws for primitive recursion.
\begin{align*}
\natrec{B}{z}{s}{\ztm} &= z \\
\natrec{B}{z}{s}{(\suctm{t})} &= \appTm{(\appTm{s}{t})}{(\natrec{B}{z}{s}{t})}
\end{align*}

\para{Enumeration sets}  The type $\enum n$ has the $n$ canonical
inhabitants $\const n 0$, \dots, $\const n {n-1}$, which can be
eliminated by the dependent case distinction  $\elim{n}{B}{t_{0} \cdots
  t_{n-1}}{t}$ with $n$ branches.
\begin{gather*}
  \tyrule{\enumrule{u-i}%
  }{\Gamma \in \ctx}{\enum{n}\in\term{\Gamma}{\TmU}}
\qquad
  \tyrule{\enumrule{i}%
  }{
    \Gamma \in \ctx\\ i<n
  }{
    \const{n}{i} \in \term{\Gamma}{\enum{n}}}
\\[1ex]
  \tyrule{\enumrule{e}
  }{
    \REDUNDANT{\Gamma \in \ctx\\}
    B\in \type{\ctxe{\Gamma}{\enum{n}}}
      \\ t\in \term{\Gamma}{\enum{n}}
      \\\\ t_{0}\in \term{\Gamma}{\subsTy{B}{\subid{}{\const{n}{0}}}}\ \cdots\ 
      t_{n-1}\in \term{\Gamma}{\subsTy{B}{\subid{}{\const{n}{n-1}}}}
  }{\elim{n}{B}{t_{0} \cdots t_{n-1}}{t}\in
    \term{\Gamma}{\subsTy{B}{\subid{}{t}}}}  
\end{gather*}
We add the usual computational law for case distinction, and weak
extensionality, which for booleans ($\enum 2$) reads
``$\mathsf{if}\;t\;\mathsf{then}\;\mathsf{true}\;\mathsf{else}\;\mathsf{false}
= t$''in sugared syntax.
\begin{align*}
\elim{n}{B}{t_{0} \cdots t_{n-1}}{\const{n}{i}}&= t_{i} \\
\elim{n}{\enum n}{\const{n}{0} \cdots \const{n}{n-1}}{t}&= t 
\end{align*}
For $\enum
0$ and $\enum 1$ we can formulate strong $\eta$-laws:  all their
inhabitants are considered equal, since there is at most one.  To
realize this, we introduce a new term $\oprf$ in $\enum 0$ \emph{if
  it already has an inhabitant $t$}; we consider $\oprf$ as normal
form of $t$.  Note that this seemingly paradoxical canonical form
$\oprf \in \enum 0$ does not threaten consistency, since it cannot
exist in the empty context $\Gamma = \ectx$; otherwise there would
have already been a term $t \in \term\ectx{\enum 0}$.
\begin{gather*}
  \tyrule{{\enumrule[0]tm}%
        }{t \in \term \Gamma {\enum 0}
        }{\oprf \in \term \Gamma {\enum 0}}
\quad
  \tyrule{{\enumrule[0]eq}%
        }{t,t' \in \term \Gamma {\enum 0}%
        }{t = t' \in \term \Gamma {\enum 0}}
\quad
   \tyrule{{\enumrule[1]eq}%
         }{t,t' \in \term \Gamma {\enum 1}%
         }{t = t' \in \term \Gamma {\enum 1}}
\end{gather*}
On the one hand, rule (\textsc{\enumrule[0]tm}) destroys decidability
of type checking: to check whether $\oprf \in \term \Gamma {\enum 0}$
we would have to decide the consistency of $\Gamma$ which is certainly
impossible in a theory with natural numbers.  On the other hand, it
allows us to decide equality by computing canonical forms.  We solve
this dilemma by forbidding $\oprf$ in the user syntax which is input
for the type-checker; $\oprf$ is only used internally in the NbE
algorithm and in the canonical forms it produces. Formally this is
reflected by having two calculi: one with the rule
\ruleref{\enumrule[0]tm} and one without it. For distinguishing the
calculi, we decorate the turnstile ($\vdashp$) in judgements of the
former and leave $(\vdash)$ for the calculus without
\ruleref{\enumrule[0]tm}. We also use the different turnstiles for
referring to each calculus. In Sect. \ref{sec:cons-vdashp} we prove
that $(\vdashp)$ is a conservative extension of $(\vdash)$.

Strong extensionality for booleans and larger enumeration sets is hard
to implement
\cite{altenkirchDybjerHofmannScott:lics01,balatDiCosmoFiore:popl04} 
and beyond the scope of this work.

In the sequel we use $\vec{t}$ for denoting the $n$ terms
$t_0 \cdots t_{n-1}$ in $\elim{n}{B}{t_0\cdots t_{n-1}}{r}$. We will
omit the superscript $n$ in $\constn{i}$, and in $\elimn{B}{\vec{t}}{r}$.

\para{Proof irrelevance}  
Our treatment of proof-irrelevance is based on
Awodey and Bauer \cite{awodeyBauer:propositionsAsTypes} 
and Maillard \cite{odalric-pi-sn}. 
The constructor $\boxty{}$ turns a type $A$
into the \emph{proposition} $\boxty A$ in the sense that only the fact
matters \emph{whether} $A$ is inhabited, not by \emph{what}.  An
inhabited proposition is regarded as \emph{true}, an uninhabited as
\emph{false}.  The proposition
$\boxty A$ still has all inhabitants of $A$, but now they are
considered equal.  If $A$ is not empty, we introduce a trivial proof
$\oprf$ in $\boxty A$ which we regard as the normal form of any $t \in
\term \Gamma {\boxty A}$.   
\begin{gather*}  
  \tyrule{prf-f%
  }{\REDUNDANT{\Gamma \in \ctx\\ } 
    A \in \term{\Gamma}{\TmU}
  }{\boxty{A} \in \term{\Gamma}{\TmU}}
\qquad
  \tyrule{prf-f%
  }{\REDUNDANT{\Gamma \in \ctx\\ } 
    A \in \type{\Gamma}
  }{\boxty{A} \in \type{\Gamma}}
\\
  \tyrule{prf-i%
  }{\REDUNDANT{\Gamma \in \ctx\\ A \in \type{\Gamma}\\ }
    a\in\term{\Gamma}{A}
  }{\boxtm{a} \in \term{\Gamma}{\boxty{A}}}
\qquad
  \tyrule{prf-tm%
  }{\REDUNDANT{\Gamma \in \ctx\\ A \in \type{\Gamma}\\ }
    a\in\term{\Gamma}{A}
  }{\oprf \in \term{\Gamma}{\boxty{A}}}
\\
  \tyrule{prf-eq%
  }{\REDUNDANT{\Gamma \in \ctx\\ } 
    A \in \type{\Gamma}\\
    t,t'\in\term{\Gamma}{\boxty{A}}
  }{t = t' \in \term{\Gamma}{\boxty{A}}}
\end{gather*}
Note that (\textsc{prf-tm}) is analogous to (\textsc{\enumrule[0]tm})
and the same remarks apply; in particular, (\textsc{prf-tm}) is also a
rule in $(\vdashp)$ but not in ($\vdash$).

We use Awodey and Bauer's \cite{awodeyBauer:propositionsAsTypes}
elimination rule for proofs.
\[
  \dfrac{\Gamma \derN t : \boxty A \qquad
            \Gamma \derN B \qquad
            \Gamma, x \of A \derN b : B \qquad
            \Gamma, x \of A, y \of A \derN b = b[y/x] : B
          }{\Gamma \derN \ABwhere b x t : B}
\]
The content $x : A$ of a proof $t : \boxty A$ can be used in $b$ via
the elimination
$b \; \mathsf{where}\; [x] = t$ if $b$ does not actually depend on it,
which is expressed via the hypothesis that $b$ should be equal to
$b[y/x]$ for an arbitrary $y$.  This elimination principle is stronger
than ``proofs can only be used inside of proofs'' which is witnessed by
the rule:
\[
  \dfrac{\Gamma \derN t : \boxty A \qquad 
            \Gamma \derN B \qquad
            \Gamma, x \of A \derN b : \boxty B 
          }{\Gamma \derN \ABwhere b x t : \boxty B}
\]
Note that this weaker elimination rule in the style of a \emph{bind}
operation for monads is an instance of the Awodey-Bauer rule, since
the equation $\Gamma, x \of A, y \of A \derN b = b[y/x] : \boxty B$
holds trivially due to proof irrelevance.  An example which is typable
with the Awoday-Bauer rule but not the monadic rule is the term
$\tmagicraw$ given in the next section.

The Awodey-Bauer $\mathsf{where}$ fulfills $\beta$-, $\eta$-, and
associativity laws analogous to the ones of a monad.
\begin{align*}
  \ABwhere b x {\boxtm a} & = b[a/x] \\
  \ABwhere {b[\boxtm x/y]} x t & =  b[t/y] \\
  \ABwhere a x (\ABwhere b y c) & = \ABwhere {(\ABwhere a x b)} y c &
   \mbox{if } y \not\in\FV(a)
\end{align*}
After this more readable presentation in named syntax, we add the
eliminator and its equations to our GAT in de Bruijn style:
\begin{gather*}
  \tyrule{prf-el%
  }{t \in \term{\Gamma}{\boxty{A}} \\
    B\in \type{\Gamma}\\ 
    b \in \term{\ctxe{\Gamma}{A}}{\subsTy{B}{\p}}\\
    \subsTm b \p = \subsTm b {\exsubs{\subsc \p \p}{\q}}
      \in \term{\ctxe{\ctxe \Gamma A} {\subsTy A \p}}{\subsTy B {\subsc \p \p}}
  }{\wheretm{b}{t}{B} \in \term{\Gamma}{B}}  
\end{gather*}
\begin{align*}
  \wheretm{b}{\boxtm{a}}{B} 
      & = \subsTm{b}{\subid{}{a}} 
    & \rulename{prf-$\beta$} \\
  \wheretm{\subsTm b {\exsubs \p {\boxtm \q}}}{t}{B}
    & = \subsTm b {\subid{} t} 
      & \rulename{prf-$\eta$}\\
  \wheretm a {(\wheretm b c B)} A 
    & = \wheretm {(\wheretm {\subsTm a {\exsubs {\subsc \p \p} \q}} b {\subsTy A \p})} c B
      & \rulename{prf-assoc}\\
\end{align*}


%% file: examples.tex
\section{Examples}
\label{sec:examples}

\subsection{Safe vector projection in \lambdaPI}

We give a short demonstration how to use proof irrelevance in
\lambdaPI: we define vectors and a type safe projection function.
While de Bruijn style is good for implementation and reasoning, it is
virtually unreadable for humans, so we allow ourselves named
$\lambda$-terms here which can be mechanically converted into actual
terms of \lambdaPI.  For instance, we write $\F A {(x.B)}$
instead of the de Buijn style $\F A {B}$.
For further convenience, let $\funT x A B = \F A {(x. B)}$ and
$\sigT x A B = \DSum A {(x. B)}$.  The non-dependent versions
are written $A \to B = \funT \_ A B$ and $A \times B = \sigT \_ A B$.
The type $\tVec A n$ of vectors of length $n$ over element type $A$ can
be defined recursively as follows.
\bprg
  \begin{array}{@{}lcl}
  \Vecraw & : & \funT A \TmU \funT n \natty \TmU \\
  \Vecraw & = & \lambda A \lambda n.\; \natrec \TmU {\enum 1} {(\lambda
    n'\lambda V.\, A \times V)} n \\
  \end{array}
\eprg
A more suggestive notation for definitions by recursion is
\emph{pattern matching}; in this the definition of $\Vecraw$ reads as
follows:
\bprg
  \begin{array}{@{}lcl}
  \multicolumn 3 {@{}l} {\Vecraw  :  \funT A \TmU \funT n \natty \TmU }\\
  \tVec A \ztm & = & \enum 1 \\
  \tVec A (\suctm n') & = & A \times \tVec A n' \\
  \end{array}
\eprg
In the following, we use pattern matching as syntactic sugar for
$\natrecraw$.  Our language already has booleans, so let us define
comparison of natural numbers.
\bprg
\begin{array}[t]{@{}l}
\begin{array}[t]{@{}l@{~}l@{~}lcl}
  \multicolumn 5 {@{}l} {\tleqraw  :  \funT m \natty \funT n \natty \Bool} \\
  \tleqraw & \ztm         &  n           & = & \ttrue \\
  \tleqraw & {(\suctm m)} &  \ztm        & = & \tfalse \\
  \tleqraw & {(\suctm m)} & {(\suctm n)} & = & \tleq m n \\
\end{array}
\qquad \qquad
\begin{array}[t]{lcl}
  \Bool & = & \enum 2 \\
  \ttrue & = & \const 2 0 \\
  \tfalse & = & \const 2 1 
\end{array}
\end{array}
\eprg
By reflecting booleans into $\TmU$ we can obtain witnesses of
propositions.  
\bprg
\begin{array}{@{}lcl}
  \multicolumn 3 {@{}l} {\Trueraw :  \funT b \Bool \TmU} \\
  \TrueC \ttrue  & = & \enum 1 \\
  \TrueC \tfalse & = & \enum 0 \\ 
\end{array}
\eprg
$\TrueC (\tleq m n)$ is inhabited if $m \leq n$, because then $\TrueC
(\tleq m n)$ simplifies to $\enum 1$ with trivial inhabitant $\const 1
0$.  If not $m \leq n$ then $\TrueC (\tleq m n)$ simplifies to the
empty type $\enum 0$.  A proposition $\Ltraw$ for ``less than'' is
obtained as:
\bprg
\begin{array}{@{}lcl}
  \multicolumn 3 {@{}l} {\Ltraw  :  \funT m \natty \funT n \natty \TmU} \\
  \Lt m n & = & \TrueC (\tleq {(\suctm m)} n) \\
\end{array}
\eprg
We are now ready to define a safe projection operation for vectors.
\bprg
\begin{array}{@{}l@{~}l@{~}l@{~}l@{~}l@{~}lcl}
  \multicolumn 8 {@{}l} {\tlookupraw  :  \funT A \TmU \funT n \natty \funT m \natty 
                    \funT p {\boxty (\Lt m n)} \funT v {\tVec A n} A
                    }\\
  \tlookupraw & A & \ztm         & m            & p & v & = & \tmagic A p \\
  \tlookupraw & A & {(\suctm n)} & \ztm         & p & v & = & \tfst\;v \\
  \tlookupraw & A & {(\suctm n)} & {(\suctm m)} & p & v & = & \tlookup A n m p
    (\tsnd\,v)  
\end{array}
\eprg
Since $\Lt {(\suctm m)} (\suctm n) = \Lt m n$ we can simply pass $p$
to the recursive call in the last equation.  In the first line we have
to magically conjure an element of $A$ from a proof $p : \boxty{(\Lt
  m \ztm)} = \boxty{(\TrueC (\tleq {(\suctm m)} \ztm))} = \boxty{(\TrueC
\tfalse)} = \boxty{\enum 0}$.
\bprg
\begin{array}{@{}l}
  \tmagicraw  :  \funT A \TmU \funT p {\boxty{\enum 0}} A \\
  \tmagic A p = \ABwhere {\elimraw^0\, A\, q} q p \\
\end{array}
\eprg
This is well typed since all inhabitants of $\enum 0$ are equal, thus,
the Awodey-Bauer rule \rulename{prf-el} is applicable.  

The benefit of proof irrelevance is that now for any $p,q : \Lt m n$,
$\tlookup A n m p v = \tlookup A n m q v : A$; for a more detailed
discussion consult Werner \cite{werner:strengthProofIrrelevance}.

\subsection{Isomorphisms in \lambdaPI}

Any already irrelevant type is isomorphic to its $\prfraw$ version,
i.e., for $A \in \{ \enum 0, \enum 1, \boxty B, \singTm t B \}$ we have
coercions
\bprg
\begin{array}{@{}lclcl}
  \phi & = & \lambda x.\, \boxtm x & : & A \to \boxty A \\
  \psi & = & \lambda x.\, \ABwhere y y x & : & \boxty A \to A \\ 
\end{array}
\eprg
with $\psi \circ \phi = \lambda x.\,x$ by \rulename{prf-$\beta$} and
$\phi \circ \psi = \lambda x.\,x$ by proof irrelevance.  How do these
coercions extend to higher types? For arbitrary $A,B$ we have.
\bprg
\begin{array}{@{}lcl}
  \phi^\Sigma & : & \sigT x {\boxty A} {\boxty B} \to \boxpar{\sigT x A B} 
\\
  \phi^\Sigma & = &\lambda p.\, \ABwhere 
    {\ABwhere {\boxtm{(a,b)}} a {\dfst p}} b {\dsnd p}
\\[1ex]
  \psi^\Sigma & : & \boxpar{\sigT x A B} \to \sigT x {\boxty A} {\boxty B}  
\\
  \psi^\Sigma & = & \lambda q.\, \ABwhere 
     {(\boxtm{\dfst p},\, \boxtm{\dsnd p})} p q
\\[2ex]
  \phi^\Pi & : & (\funT x A {\boxty B}) \to \boxpar{\funT x A B}
\\
  \phi^\Pi & = & \lambda f.\, \boxtm{\lambda x.\, \ABwhere ? y {f\,x}}
\\[1ex]
  \psi^\Pi & : & \boxpar{\funT x A B} \to \funT x A {\boxty B}
\\
  \psi^\Pi & = & \lambda f \lambda x.\, \ABwhere {\boxtm{g\,x}} g f
\end{array}
\eprg
We would like to put $y$ for the $?$ in $\phi^\Pi$, but this is not
well typed, since we do not have $y = z : B$ for arbitrary $z$.  It
seems that $\phi^\Pi$ is not definable in $\lambdaPI$, as it is not
definable in computational lambda-calculus \cite{complambda}
for an arbitrary monad
$\prfraw$.  Awodey and Bauer also have only $\phi^\Pi : (\funT x A {\boxty B})
\to \boxpar{\funT x A \boxty B}$, which is trivial.

For arbitrary $p : \sigT x {\boxty A} {\boxty B}$ we have
\[
\begin{array}{lcl}
  (\psi^\Sigma \circ \phi^\Sigma)(p) 
    & =_\beta & \ABwhere {\ABwhere {(\boxtm a,\, \boxtm b)} a {\dfst p}} b {\dsnd p}
\\  & =_\eta  & (\dfst p,\, \dsnd p) =_\eta  p .
\end{array}
\]
In the opposite direction, $\phi^\Sigma \circ \psi^\Sigma = \lambda
q.\, q$ by proof irrelevance.  Thus, $\phi$ and $\psi$ establish an
isomorphism, which means that $\prfraw$ distributes over $\Sigma$.

\subsection{On subtyping in \lambdaSing}

Subtyping can be defined in several ways, for instance, $A$ is a
subtype of $A'$ in $\Gamma$, written $\Gamma \derN A \subtype A'$,
iff $\Gamma, x \of A \derN x \of A'$.
Most presentations of singleton types include
subtyping \cite{aspinall:csl94,courant:itrs02,stoneHarper:tocl06}, so
it is natural to ask whether the usual rules hold in our calculus.
Using the principle $u = t : A$ iff $u : \singTm t A$, it is easy to
see that we have Aspinall's two axioms \cite{aspinall:csl94}:
\[
\begin{array}{lcl@{~}c@{~}ll}
  x : \singTm t A & \derN & x & : & A \\
  x : \singTm t A & \derN & x & : & \singTm t {\singTm t A}
   & \mbox{since } x = t : \singTm t A
\end{array}
\]
Also, singleton formation is compatible with subtyping, if $\Gamma
\derN A \subtype B$ then $\Gamma \derN \singTm t A \subtype \singTm t
B$.  Contravariant subtyping, however, only holds up to
$\eta$-equality.  If we relax the definition of subtyping $\Gamma
\derN A \subtype A'$ to $\Gamma, x \of A \derN \eta(x) \of A'$ where
$\eta(x)$ denotes any $\eta$-expansion of $x$, then we get
contravariant subtyping
\[
  \dfrac{\Gamma, x \of A' \derN \eta_A(x) : A \qquad
         \Gamma, x \of A', y \of B[\eta_A(x)/x] \derN \eta_B(y) : B'
       }{\Gamma, f \of \funT x A B \derN 
           \lambda x.\, \eta_B(f(\eta_A(x))) : \funT x {A'} B'}.
\]
Furthermore, we have following two axioms, which hold definitionally in Stone and
Harper's system \cite{stoneHarper:tocl06}
\[
\begin{array}{lcl@{~}c@{~}ll}
  f : \singTm t {\funS x A B} & \derN & 
    \lambda x.\,f\,x 
     & : &  \funT x A \singTm {t\,x} B 
     & \mbox{since } f = t : \funT x A B \\
  f : \funT x A \singTm {t\,x} B & \derN & 
  \lambda x.\,f\,x 
     & : & \singTm t {\funS x A B} 
     & \mbox{since } f\,x = t\,x : B \\
\end{array}
\]
The first axiom is one of Courant's subtyping rules \cite{courant:itrs02}. 


%% file: nbeprimer.tex

\section{From Untyped to Typed Normalisation by Evaluation}
\label{sec:nbe-primer}

\noindent In this section, we give a short introduction into
normalisation-by-evaluation for typed lambda calculi, with a special
emphasis on our novel method for generation of fresh identifiers during
reification.

\subsection{Fresh Name Generation in NbE}

The basic idea of NbE is to \emph{evaluate} a term of type $A$ into a
suitable semantics $\den A$ from which its \emph{normal form} can be
extracted by \emph{reification}.  In case of the simply-typed
lambda-calculus, this is possible if we choose for base types $\den o$
the set of terms of type $o$ and for function types $\den{A \to B}$ a
suitable subset of the function space $\den A \to \den B$.  During
reification of a function $f \in \den{A \to B}$ to a term $\lambda x.
t$, the identifier $x$ has to be chosen \emph{fresh} to avoid capture
of names in the body of the function $f$.  However, since $f$ is a
semantic object, it is a non-trivial problem to compute a name which
is fresh for $f$.  Garillot and Werner \cite{garillotWerner:tphols07}
solve it by first letting $x$ be a dummy identifier, computing the
free variables in the reified function body $t$, and then reify $f$
again with a name $x$ which is fresh for $t$.  This is, of course,
horribly inefficient, and there are other solutions.  In the original
publication on NbE by Berger and Schwichtenberg
\cite{bergerSchwichtenberg:lics91}, base types $\den o =
\widehat\Lambda$ are interpreted by \emph{term families}.  These are
functions $g$ from the natural numbers into a \emph{de Bruijn level}
representation of terms such that all instances $g(n)$ are
$\alpha$-equivalent but in $g(n)$ the bound variables are levels
starting with $n$. In this setting, the reification of a function $f
\in \den{A \to B}$ is not a term but a term family, mapping $f$ to the
term family $n \mapsto \lambda x_n . \phi (f(\widehat{x_n})) (n+1)$,
where $\phi$ denotes the reification function and $\widehat{x_n}$ is
the variable $x_n$ seen as an element in $\den A$. 
Note that every $\lambda$ in $\phi (f(\widehat{x_n})) (n+1)$, 
the body of the reified abstraction, will bind a variable 
from the set $\{x_{n+1}, x_{n+2}, \ldots\}$.

When considering NbE for the \emph{untyped} lambda calculus, the type
semantics collapses to a single domain $\D \cong \widehat\Lambda +
[\D \to \D]$ which contains terms and functions,\footnote{
Let us notice
here the \emph{tagging} introduced by the disjoint sum operator $+$.
Indeed, in the absence of a type structure, \emph{tagless
  normalisation} seems impossible.  
}
as observed by Filinski and Rhode \cite{filinskiRohde:untypedNbE}.  
Aehlig and Joachimski \cite{aehlig:nbe} replace term families by
functions $h$ from natural numbers to a \emph{de Bruijn index}
representation of terms, where $h(n)$ shifts all \emph{free} indices
by $n$. 

In this paper, instead of having term families $\widehat\Lambda$ in
the semantics, we have a notion of \emph{neutral (``term-like'')
  value} built up from free variables $x_i$ and application of the
free variables to sequences of values $\vec d$.  The free variables
are \emph{de Bruijn levels} in spirit, thus, no shifting is needed,
just like in the \emph{locally nameless} approach
\cite{pollack:alpha}.  The second author has given a semantics with
neutrals before \cite{coquand:type}, calling the free variables
\emph{generic values}.  Also, this approach has been used by the first
two authors together with Dybjer \cite{abelCoquandDybjer:mpc08} for
NbE without a reflection operation, and independently by L\"oh,
McBride, and Swierstra
\cite{loehMcBrideSwierstra:tutorialDependentlyLambda}.  In this
article, we put the technique to a novel use by defining typed
reification \emph{and reflection} for this semantics.

\subsection{Untyped NbE}
\label{sec:untyped-nbe}
Let $\Var = \{x_0, x_1, \dots\}$ be a denumerable set of variables.
We consider a set $\D$ and a notion of function space $[\D \to \D]$
with an embedding constructor $\Lam : [[\D \to \D] \to \D]$ and two
further constructors $\mathsf{Var} : \Var \to \D$ and $\App : [\D \times \D \to
\D]$ for neutral values.  An application function $\_\cdot\_ : [\D
\times \D \to \D]$ is given by
\[
\begin{array}{l@{~}l@{~}lll@{\quad}l}
  (\Lam\,f) & \cdot & d & = & f(d) \\
  e         & \cdot & d & = & \App(e,d) & \mbox{if } e \mbox{ not a }
    \Lam .\\ 
\end{array}
\]
Such a $\D$ can be realised by solving the recursive domain equation 
$\D \cong [\D \to \D] \oplus \Var_{\bot}  \oplus (\D \times \D)$ or, for
the practically minded, by defining a Haskell data type
\begin{verbatim}
    data D where
      Lam : (D -> D) -> D
      Var : Var -> D
      App : D -> D -> D
\end{verbatim}
and programming $\_\cdot\_$ by pattern matching.  Our definition of
$\D$ is a bit ``too big'' since it does not restrict $\App$ to the
construction of neutral values
$\App(\dots\App(\iVar i,d_1)\dots,d_n)$ but we have also
$\App(\Lam\,f,d)$.  However, we can ignore these unwanted elements
since our NbE algorithm never produces any. 

\begin{rem}
\label{rem:comp-adequacy}
  The relationship between the denotational model\/ $\D$ and the
  Haskell data type $\texttt{D}$ is not without subtleties.  Domain
  theoretic functions such as application $\_\cdot\_$ correspond to
  Haskell programs if our denotational semantics is computationally
  adequate for Haskell's operational semantics
  \cite{plotkin:adequacy}. Filinski and Rhode
  \cite{filinskiRohde:untypedNbE} formally relate a NbE
  function on a reflexive domain\/ $\D$ to a NbE program written in an
  ML-like, call-by-value language, by exploiting computational
  adequacy.  We do not formally prove this connection for Haskell in
  this article, this is deferred to future work.  
\end{rem}

Untyped NbE is now given by a standard evaluator $\Den t \rho \in \D$
of terms $t$ in environments $\rho$ and a readback function $\reify j
d$ from values $d$ at de Bruijn level $j$ to terms
\cite{gregoireLeroy:compiledReduction}.   For the sake of
readability, we use names instead of de Bruijn indices in the syntax
of untyped terms.
\[
\begin{array}{c@{\qquad}c}
  \begin{array}[t]{rll}
    \Den x \rho & = & \rho(x) \\
    \Den {r\, s} \rho & = & \Den r \rho \cdot \Den s \rho \\
    \Den {\lambda x. t} \rho & = & \Lam\,(d \mapsto 
      \Den t {\update \rho x d}) \\
  \end{array}
&
  \begin{array}[t]{lll}
    \reify j (\iVar i) & = & x_i \\
    \reify j (\App(r,s)) & = & (\reify j r) \, (\reify j s) \\
    \reify j (\Lam\, f) & = & \lambda x_j.\, \reify {j+1}
      (f(\iVar j)) \\
  \end{array}
\end{array}
\]
To normalise a closed term $t$, compute $\reify 0 \den t$.  To
normalise an open term $t$ with free variables $y_0,\dots y_{n-1}$
compute $\reify n {\Den t \rho}$ with environment $\rho(y_i) = \iVar i$.

To prepare for applying our method to $\lambdaSing$ and $\lambdaIrr$,
let us switch to de Bruijn representation.  Environments become tuples
and variables de Bruijn indices $\ind i$.
\[
\begin{array}{c@{\qquad}c}
  \begin{array}[t]{lll}
    \Den \tq (\rho,d) & = & d \\
    \Den {t\,\tp} (\rho,d) & = & \Den t \rho \\
    \Den {r\, s} \rho & = & \Den r \rho \cdot \Den s \rho \\
    \Den {\lambda t} \rho & = & \Lam\,(d \mapsto 
      \Den t (\rho,d)) \\
  \end{array}
&
  \begin{array}[t]{lll}
    \reify j (\iVar i) & = & \ind{j - (i+1)} \\
    \\
    \reify j (\App(r,s)) & = & (\reify j r) \, (\reify j s) \\
    \reify j (\Lam\, f) & = & \lambda \, (\reify {j+1}
      (f(\iVar j))) \\
  \end{array}
\end{array}
\]
To read back a de Bruijn level $\iVar i$ as a de Bruijn index, we have
to take the current length $j$ of the variable context into account.
While de Bruijn levels are absolute references, they are numbered
$x_0, x_1, \dots, x_{j-1}$ in a context of length $j$, de Bruijn
indices are relative to the length of the context, they are enumerated
from right to left: 
$\ind{j-1}, \dots, \ind 1, \ind 0$.  
The formula $j - (i+1)$ (assuming $i < j$) converts level $i$
into the corresponding index.

\subsection{Typed NbE}

While untyped NbE returns a $\beta$-normal form (if it exists), typed
normalisation by evaluation yields a $\beta\eta$-normal form, usually
the $\eta$-long form.  To obtain the $\eta$-long form, we have to
modify our reification procedure.  One method is to make read-back
type directed \cite{abelCoquandDybjer:mpc08}, which corresponds to
postponing $\eta$-expansion after $\beta$-normalisation. 
However, this strategy
is not sufficient in the case of $\lambdaSing$, because
$\eta$-expansions at singleton types can trigger new
$\beta$-reductions.   The other method is to divide $\eta$-expansion
into \emph{reflection} and \emph{``reification''}, the first expanding
variables to enable new reductions, and the second expanding the
result of $\beta$-normalisation to obtain an $\eta$-long
form.\footnote{In tagless normalisers
  \cite{bergerSchwichtenberg:lics91}, 
  reflection is necessary to
  inject variables $x : A$ of non-base types $A$ into the semantics
  $\den A$.  However, for languages beyond pure type systems it is
  hard to obtain tagless normalisation.  Classic is the problem of
  disjoint sum types \cite{altenkirchDybjerHofmannScott:lics01}: to
  display a free variable of type $A + B$ as either a left or a right
  injection, we need control structures
  \cite{balatDiCosmoFiore:popl04}.  Alternatively, one can replace
  data types by their Church encodings.  None of these approaches fit
  our purposes, thus, we are currently not aiming at tagless normalisation.}

The novel approach of this article is to do reflection $\uparrow$ and
``reification'' $\downarrow$, hence, $\eta$-expansion, completely at
the level of
the semantics $\D$.  Since our value domain $\D$ allows us to
construct functions via $\Lam$, the process of $\eta$-expansion is
independent of any fresh name considerations.
\[
\begin{array}{c@{\qquad}c}
  \begin{array}{l@{~}lll}
    \up o        {}& e & = & e \\
    \up{A \to B} {}& e & = & \Lam\,(d \mapsto \up B \App(e,\down A d)) \\
  \end{array}
&
\begin{array}{l@{~}lll}
  \down o         {}& d & = & d \\
  \down {A \to B} {}& d & = & \Lam\,(e \mapsto \down B (d \dapp (\up A
  e))) \\
\end{array}
\end{array}
\]
To compute the long normal form of a closed term $t$ of type $A$, 
run $\reify 0 (\down A {\den t})$.  For an open term $y_0 \of A_0,
\dots, y_{n-1} \of A_{n-1} \derN t : A$, execute $\reify n (\down A
{\Den t \rho})$, where $\rho(y_i) = \up{A_i} {(\iVar i)}$.


%% file: tlca09-model.tex
\section{Semantics}
\label{sec:semantics}

\noindent In this section we define a domain $D$ for denoting types,
  terms, and substitutions. Then we introduce a partial function
  $\reify j {}$ for
  reifying elements of the domain into the calculus; this function
  takes an extra argument $j \in \NN$ indicating the next free variable.  We
  continue by defining PERs over the domain; these PERs denote the
  axioms for types, terms, and substitutions. We need PERs for the
  evaluation function is defined over syntactical entities and not for
  typing judgements. We also introduce PERs $\pernf$ and $\perne$ 
  whose elements are
  invariably, in every context, reified as normal forms and
  neutral terms respectively.  Using these PERs we define a family
  (indexed by denotations of types) of functions for ``normalising''
  in the domain. We conclude this section proving completeness for
  this family of normalisation functions; here completeness means that
  two terms in the theory are read back as the same normal form.
  In this section we define a PER model of the calculus presented in
  the previous section. The model is used to define a normalisation
  function later.

\subsection{PER semantics}
\label{sec:persem}

\LONGVERSION{
  \input{per-model}

  \subsection{A concrete PER model} In this subsection we define a
  concrete PER model over a Scott domain. The definition of the
  evaluation function is post-poned to the next subsection after
  introducing the NbE machinery.}

\begin{defi}
  \label{def:domain}
  We define a domain 
  $$D = \mathbb{O} \oplus \Var_\bot \oplus  [D \into D]
  \oplus (D \times D) \oplus (D \times D)  \oplus \mathbb{O}
  \oplus (D \times [D\into D])
  \oplus (D\times D)\enspace,$$ where $\Var$ is a denumerable set
  of variables (as usual we write $x_i$ and assume $x_i\neq x_j$ if
  $i\neq j$, for $i,j\in \mathbb{N}$), 
  $E_\bot = E \cup \{ \bot \}$ is lifting, $\mathbb{O} = \{\top\}_\bot$ 
  is the Sierpinski
  space, $[D\into D]$ is the set of continuous functions from $D$ to
  $D$, $\oplus$ is the coalesced sum\LONGSHORT{ (this is the
    disjoint union where all the bottoms elements are identified),}{,}
  and $D\times D$ is the Cartesian product of $D$
  \cite{dom-theory}. 


\end{defi}
An element of
$D$ which is not $\bot$ can be of one of the forms:
\begin{align*}
&\iO           & \mbox{\hspace{1cm}}&\iPair{d}{d'}&
\mbox{\hspace{1cm}}&\mbox{for } d, d' \in D\\
&\iVar{i}      & \mbox{\hspace{1cm}}&\iU          &
\mbox{\hspace{1cm}}&\mbox{for } x_i \in \Var\\
&\iLam{f}      & \mbox{\hspace{1cm}}&\iPi{d}{f}    &
\mbox{\hspace{1cm}}&\mbox{for } d\in D\mbox{, and } f \in [D\into D]\\
&\iNe{d}{d'} & \mbox{\hspace{1cm}}&\iSing{d}{d'}  &\
\mbox{\hspace{1cm}}&\mbox{for } d, d' \in D .
\end{align*} 
Elements of the form $\iVar i$ and $\iNe d d'$ are called
\emph{neutral}; in this section, we reuse the letter $k$ to denote
neutral elements of $D$.

In order to define an environment model over $D$, we endow it with an
applicative structure. Note also that $D$ has pairing, letting us to
take the set of sequences over $D$ simply as $D^* = D$ with
$\tPair\,a\,b = (a,b)$. We define application $\_\cdot\_ : [D \times D
\to D]$ and the projections $\fst,\snd : [D \to D]$ by
\[
\begin{array}{lll}
  f \cdot d & = & 
   \mbox{if } f = \iLam f' \mbox{ then } f'\,d \mbox{ else } \bot
,\\
  \fst\,d & = & 
   \mbox{if } d = (d_1,d_2) \mbox{ then } d_1 \mbox{ else } \bot
,\\
  \snd\,d & = & 
   \mbox{if } d = (d_1,d_2) \mbox{ then } d_2 \mbox{ else } \bot
.\\
\end{array}
\] 
We define a partial function $\reify{\_}{\_}: \mathbb{N} \into D \into
\terms$ which reifies elements from the model into terms; this
function is similar to Gr\'egoire and Leroy's read-back function
\cite{gregoireLeroy:compiledReduction}. 
\begin{defi}[Read-back function]
  \label{def:reify}
\[
  \begin{array}[t]{r@{~}c@{~}l}
    \reify{j} {\iU}& = & \TmU \\
    \reify{j} {(\iPi{X}{F})} & = & \F{(\reify{j}{X}) 
                 }{(\reify{j+1}{(F(\iVar{j}))})} \\
    \reify{j} {(\iSing{d}{X})} & = &
    \singTm{\reify{j}{d}}{\reify{j}{X}} \\
  \end{array}
\quad
  \begin{array}[t]{r@{~}c@{~}l}
    \reify{j} {(\iNe{d}{d'})} & = &   \appTm{(\reify{j}{d})}{(\reify{j}{d'})}  \\
    \reify{j} {(\iLam{f})} & = & \lambda(\reify{j+1}{(f (\iVar{j}))}) \\
    \reify{j} {(\iVar{i})}& = & \ind{j - (i + 1)}
  \end{array}
\]
\end{defi}

As explained in Sect.~\ref{sec:untyped-nbe}, the reification of
variables turns de Bruijn levels into de Bruijn indices. 
Note that in case $j < i + 1$ we
return the $0$th de~Bruijn index just not to be undefined;
we will come back to this later. 


The next PERs contain those elements of the domain $D$ whose
reification is defined for any context length. Moreover, their
elements are reified as neutral terms and normal forms, respectively;
allowing us to reason semantically about normal forms. Remember that
$t\equiv_T t'$ denotes that $t$ is syntactically equal to $t'$ and $t
\in T$.

\begin{defi}[(Semantic) neutral terms and normal forms]
\label{def:pernf}
  \[
  \begin{array}{lll}
    d = d' \in \perne & :\iff &
    \text{for all\ } i\in \mathbb{N},\, \reify{i}{d} \mbox{ and }
    \reify{i}{d'} \mbox{ are defined and } \reify{i}{d}\equiv_{\neterms}
    \reify{i}{d'}
    \\
    d = d' \in \pernf & :\iff &
    \text{for all\ } i\in \mathbb{N},\, \reify{i}{d} \mbox{ and }
    \reify{i}{d'} \mbox{ are defined and } \reify{i}{d}\equiv_{\nfterms}
    \reify{i}{d'}
  \end{array}
  \]
\end{defi}\medskip

\noindent Notice that if the case $j < i+1$ were undefined in the clause for
variables in \ref{def:reify}, then for any $m$ the application
$\reify{0}{\iVar{m}}$ would be undefined; hence $\iVar{m} \not\in
\perne$ and, consequently, $\perne$ would be empty. Since we depend on
having a semantic representation of variables and neutrals we add the
case $j < i+1$. This case will not arise in our use of the readback
function.

\LONGVERSION{
\begin{rem}
  \label{rem:pers}
  These are clearly PERs over $D$: symmetry is trivial and
  transitivity follows from transitivity of the syntactical equality.
\end{rem}

\begin{lem}[Closure properties of $\perne$ and $\pernf$]
\label{rem:presnf} \bla
  \begin{enumerate}[\em(1)]
  \item $\iU = \iU \in \pernf$.
  \item Let $X = X' \in \perne$. If $F \cdot k = F' \cdot k' \in \pernf$ for all $k = k' \in
    \perne$, then $\iPi X F = \iPi {X'}{F'} \in \pernf$.
  \item If $d=d' \in \pernf$ and $X = X' \in \pernf$, then $\iSing{d}{X} =
  \iSing{d'}{X'} \in \pernf$.
  \item If $f \cdot k = f' \cdot k' \in \pernf$ for all $k = k' \in
    \perne$, then $f = f' \in \pernf$.
  \item $\iVar i = \iVar i \in \perne$ for all $i \in \mathbb N$.
  \item If $k=k' \in \perne$ and $d=d' \in
  \pernf$, then $\iNe{k}{d} = \iNe{k'}{d'} \in \perne$. 
  \end{enumerate}
\end{lem}
}


\SHORTVERSION{
  The following definitions are standard
  \cite{aspinall:csl94,coquandPollackTakeyama:fundinf05} (except for
  $\one$); they will be used in the definition of the model.  
\begin{defi} Let $\calX\in\perD$ and $\calF\in \calX\into \perD$.
  \begin{enumerate}[\em(1)]
  \item $\one = \{(\iO,\iO)\}$;
  \item $\sigD{\calX}{\calF} = \{(d,d')\ |\ \p\ d = \p\ d' \in \calX \mbox{ and }
    \q\ d = \q\ d' \in \calF\ (\p\ d)\}$;
  \item $\prod\,\calX\,\calF = \{(f,f')\ |\ f\cdot d = f'\cdot d' \in \calF\ d,
    \mbox{ for all } d = d' \in \calX\}$;
  \item $\singD{d}{\calX} = \{(e,e')\ |\ d = e \in \calX \mbox{ and } d = e'
    \in \calX\}$.
  \end{enumerate}
\end{defi}
}

We define $\perU, \perT\in \perD$ and $[\_] : \dom(\perT) \into
\perD$ using Dybjer's schema of inductive-recursive definition 
\cite{dybjer:jsl00}. We show then that $[\_]$ is a family of PERs
over $D$.

\begin{defi}[PER model]\label{def:peru}\hfill
\begin{enumerate}[(1)]
\item Inductive definition of $\perU \in \perD$.\hfill
  \begin{enumerate}[(a)]
  \item $\perne \subseteq \perU$,
  \item if $X = X' \in \perU$ and $d = d'\in [X]$, then $\iSing{d}{X}
    = \iSing{d'}{X'} \in \perU$,
  \item if $X = X'\in \perU$ and for all $d = d' \in [X]$, $F\ d = F'\
    d' \in \perU$ then $\iPi{X}{F} = \iPi{X'}{F'}\in \perU$.
  \end{enumerate}
\item Inductive definition of $\perT \in \perD$.\hfill
  \begin{enumerate}[(a)]
  \item $\perU \subseteq \perT$,
  \item $\iU = \iU \in \perT$,
  \item if $X = X' \in \perT$, and $d = d' \in [X]$ then $\iSing{d}{X}
    = \iSing{d'}{X'}\in \perT$,
  \item if $X = X' \in \perT$, and for all $d = d' \in [X]$, $F\ d =
    F'\ d'\in\perT$, then $\iPi{X}{F} = \iPi{X'}{F'}\in \perT$.
  \end{enumerate}
\item Recursive definition of $[\_] \in \vdom(\perT) \to \perD$.\hfill
  \begin{enumerate}[(a)]
  \item $[\iU] = \perU$,
  \item $[\iSing{d}{X}]= \singD{d}{[X]}$,
  \item $[\iPi{X}{F}] = \prod\,[X]\,(d\mapsto [F\ d])$,
  \item $[d] = \perne$, in all other cases.
  \end{enumerate}
\end{enumerate}
\end{defi}

\LONGVERSION{
\begin{rem}
  \label{rem:ord-pert}
  The generation order $\sqsubset$ on $\perT$ is well-founded. The
  minimal elements are $\iU$, and elements in $\perne$; $X \sqsubset
  \iPi{X}{F}$, and for all $d\in [X]$, $F\ d\sqsubset \iPi{X}{F}$;
  and, finally, $X \sqsubset \iSing{d}{X}$.
\end{rem}
}

\begin{lem}
  \label{lem:famperd}
\LONGSHORT{
  The function $[\_] : dom(\perT) \into \perD$ is a family of
  $\perD$ over $\perT$, i.e., $[\_] : \perT \into \perD$.
}{
  The function $[\_]$ is a family of $\perD$ over $\perT$.  
}
\end{lem}

\LONGVERSION{
  \begin{proof} By induction on $X = X' \in \perT$. 
    See Appendix~\ref{prf:famperd}.
  \end{proof}

The previous lemma leads us to the definition of a PER model over
$D$. Note also that $D$ has all the distinguished elements needed to
call it a syntactical applicative structure.

\begin{cor}
  The tuple $(D, \perU, \perT, [\_])$ is a PER model. \qed
\end{cor}
}
\newcommand{\etatitle}{\texorpdfstring{$\eta$}{eta}}
\subsection{Normalisation and \etatitle-Expansion in the Model}
\label{sec:norm-model}


In the following, we adopt the NbE algorithm outlined in
Section~\ref{sec:nbe-primer} to the dependent type theory
$\lambdaSing$.  Since read-back has already be defined, we only
require reflection, reification and evaluation functions.



\LONGSHORT{
\begin{defi}[Reflection and reification]
  \label{def:up-down}
  The partial functions $\upa{\_}{\_},\da{\_}{\_} : [D \into
  [D\into D]]$ and $\Da{} : [D \into D]$ are given as follows: 
\[
  \begin{array}{rcl}
    \upa{\iPi{X}{F}}{k} & = &  \iLam{(d \mapsto \upa{F\, d}{(\iNe {k}{\da{X}{d}})})}  \\
    \upa{\iSing{d}{X}}{k}& = & d \\
    \upa{\iU}{k} & = &  k \\
    \upa{X}{k} & = & k\mbox{, in all other cases.}
  \end{array}
\
  \begin{array}{rcl}
    \da{\iPi{X}{F}}{d} & = & \iLam{(e \mapsto \da{F\, \upa{X}{e}}{(d\,\cdot \upa{X}{e})})}  \\
    \da{\iSing{d}{X}}{e} & = & \da{X}{d} \\
    \da{\iU}{d} & = &  \Da{d} \\
    \da{X}{e} & = & e \mbox{, in all other cases.}
  \end{array}
\]
  \begin{align*}
    &\Da{(\iPi{X}{F})}   = \iPi{(\Da{X})}{(d \mapsto \Da{(F\, \upa{X}{d})})} & \\
    &\Da{(\iSing{d}{X})} = \iSing{(\da{X}{d})}{(\Da{X})}& \\
    &\Da{\iU} =  \iU& \\
    &\Da{X}  = X\mbox{, in all other cases.}&
  \end{align*}
\end{defi}
}{
\begin{defi}
  \label{def:up-down}
  The partial functions $\upa{\_}{\_},\da{\_}{\_} : D \into
  D\into D$ and $\Da{} : D \into D$ are given as follows: 
  \begin{align*}
    \upa{\iPi{X}{F}}{d}&= \iLam{(e \mapsto \upa{F\,e}{(\iNe
        {d}{\da{X}{e}})})} &
    \da{\iPi{X}{F}}{d}&= \iLam{(e \mapsto 
        \da{F\,\upa{X}{e}}{(d\,\cdot \upa{X}{e})})} & \\
    \upa{\iSing{d}{X}}{e}&= d&
    \da{\iSing{d}{X}}{e}&= \da{X}{d}& \\
    \upa{\iU}{d}&= d&
    \da{\iU}{d}&=\Da{d}& \\
    \upa{d}{e}&= e& \da{d}{e}&= e\mbox{, in all other cases.} &
  \end{align*}
  \begin{align*}
    \Da{(\iPi{X}{F})}   &= \iPi{(\Da{X})}{(d \mapsto \Da{(F\,\upa{X}{d})})} &
    \Da{\iU} &=  \iU& \\
    \Da{(\iSing{d}{X})} &= \iSing{(\da{X}{d})}{(\Da{X})}& 
    \Da{d} &= d\mbox{, in all other cases.}&
  \end{align*}
\end{defi}
}\medskip

\LONGVERSION{\noindent In the following lemma we show that reflection $\uparrow$
  corresponds to Berger and Schwichtenberg's ``make self evaluating''
  and both reification functions $\downarrow$ and
  $\Downarrow$ correspond to``inverse of the evaluation function''
  \cite{bergerSchwichtenberg:lics91}.
  Note that they are indexed by types values instead of 
  syntactic types, since we are dealing with dependent instead of
  simple types.
}

\newcommand{\uptitle}{\texorpdfstring{$\uparrow$}{\textuparrow}}
\newcommand{\downtitle}{\texorpdfstring{$\downarrow$}{\textdownarrow}}
\newcommand{\Downtitle}{\texorpdfstring{$\Downarrow$}{Down}}

\begin{lem}[Characterisation of \uptitle, \downtitle, and \Downtitle]
  \label{lem:reify}
  Let $X = X' \in \perT$, then\hfill
  \begin{enumerate}[\em(1)]
  \item if $k = k'\in \perne$ then $\upa{X}{k} = \upa{X'}{k'} \in
    [X]$;
  \item if $d=d'\in [X]$, then $\da{X}{d} = \da{X'}{d'} \in \pernf$;
  \item and also $\Da{X} = \Da{X'} \in \pernf$.
  \end{enumerate}
\end{lem}
\begin{proof}
  By induction on $X=X' \in \perT$. See \ref{prf:reify}.
\end{proof}




Let us recapitulate what we have achieved: we have defined a PER model
over the domain $D$; then we defined a family of functions $\da{X}{ }$
indexed over denotation of types with the property that when applied
to elements in the corresponding PER we get back elements which will
be reified as normal forms. In fact, we have the stronger result that
whenever we apply $\da{X}{ }$ to two related elements $d = d' \in [X]$
we get elements to be reified as the same term.

Now we define evaluation which clearly satisfies the environment model
conditions in Def.~\ref{def:syn-app-str}; hence, we have a model and, using 
Thm.~\ref{thm:soundness}, we conclude completeness for our normalisation
algorithm.

\begin{defi}[Semantics] Evaluation of substitutions and terms
  into $D$ is defined inductively by the following equations.
  \begin{align*}
    \intertext{Substitutions.}
    \semc{\ectx}d &= \iO& 
    \semc{\idsubs{}}d &= d& \\
    \semc{\exsubs{\gamma}{t}}d &=  \iPair{\semc{\gamma}d}{\semc{t}d}& 
    \semc{\p}d &=  \tfst\; d & \\
    \semc{\subsc{\gamma}{\delta}}d &= \semc{\gamma} (\semc{\delta} d)&
  \intertext{Terms (and types).}
    \semc{\TmU}d          & =  \iU& 
    \semc{\F{A}{B}}d      &=  \iPi{(\semc{A}d)}{(e\mapsto \semc{B}\iPair{d}{e})}& \\
    \semc{\singTm{a}{A}}d &=  \iSing{(\semc{a}d)}{(\semc{A}d)}& 
    \semc{\appTm{t}{u}}d  &=  \semc{t}d\cdot\semc{u}d& \\
    \semc{\lambda t}d     &=  \iLam{(d'\mapsto \semc{t}\iPair{d}{d'})}& 
    \semc{\subsTm{t}{\gamma}}d &=   \semc{t} (\semc{\gamma}d)& \\
    \semc{\q}d &= \tsnd\; d &
  \end{align*}
\end{defi}

\SHORTVERSION{
\begin{defi}[Validity] \hfill
  \begin{enumerate}[\em(1)]
  \item $\ectx \vDash$ iff true
  \item $\ctxe{\Gamma}{A}\vDash$ iff $\Gamma\vDash A$
  \item $\Gamma\vDash A$ iff $\Gamma\vDash A = A$
  \item $\Gamma\vDash A=A'$ iff $\Gamma\vDash$ and for all $d = d' \in
    \semc{\Gamma}$, $\semc{A}{d} = \semc{A'}{d'} \in \perT $
  \item $\Gamma\vDash t: A$ iff $\Gamma\vDash t = t:A$
  \item $\Gamma\vDash t=t':A$ iff $\Gamma\vDash A$ and for all $d =
    d'\in \semc{\Gamma}$, $\semc{t}{d} = \semc{t'}{d'} \in
    [\semc{A}{d}] $
  \item $\Gamma\vDash \sigma: \Delta$ iff $\Gamma\vDash \sigma =
    \sigma:\Delta$
  \item $\Gamma\vDash \sigma=\sigma':\Delta$ iff $\Gamma\vDash$,
    $\Delta\vDash$, and for all $d = d' \in \semc{\Gamma}$,
    $\semc{\sigma}{d} = \semc{\sigma'}{d'}\in \semc{\Delta}$.
  \end{enumerate}
\end{defi}
}

\SHORTVERSION{
\begin{thm}[Soundness of the Judgements]
  \label{thm:soundness}
  if $\Gamma \vdash J$, then $\Gamma\vDash J$.
\end{thm}
\begin{proof} By easy induction on $\Gamma \vdash J$. 
\end{proof}
}

\begin{thm}[Completeness of NbE]
  \label{thm:completeness}
  Let $\deqterm{\Gamma}{A}{t}{t'}$ and let also $d\in \semctx{\Gamma}$, then\linebreak
  $\da{\semc{A}d}{\semc{t}d}=\da{\semc{A}d}{\semc{t'}d}\in \pernf$.
\end{thm}
\begin{proof}
  By Thm.~\ref{thm:soundness} we have $\semc{t}d = \semc{t'}d \in [\semc{A}d]$
  and we conclude by Lem.~\ref{lem:reify}.
\end{proof}

\PrfIrrTitle

\label{sec:pi-model}

\LONGVERSION{
We extend all the definitions concerning the construction of the model.

\begin{defi}[Extension of domain $D$]
  \[
  \begin{split}
    D = \ldots &\oplus D \times [D \into D] \oplus D \oplus D\\ 
    &\oplus\mathbb{O} \oplus \mathbb{O} \oplus D\oplus \tnrecd\\
    &\oplus D \oplus \mathbb{O} \oplus \mathbb{N} \oplus \mathbb{N}
    \times \mathbb{N} \oplus \mathbb{N} \times [D\into D]\times
    D^\omega \times D\enspace .
  \end{split}
  \]
\end{defi}

}
\SHORTVERSION{
  We extend all the definition concerning the construction of the model;
  \[
  D = \ldots \oplus D \oplus \mathbb{O} \enspace ; 
  \]
  the new inhabitants will be written as $\prf{d}$, and $\dprf$,
  respectively.  The read-back function is extended by the equations
  $\reify{j}{(\prf{d})} = \boxty{(\reify{j}{d})}$ and
  $\reify{j}{\dprf} = \oprf$.  We add a new clause in the definition
  of $\perT$,
\[
\mbox{if } X = X' \in \perT \mbox{, then } \prf{X} = \prf{X'} \in
\perT \mbox{, and } [\prf{X}] = \{(\dprf,\dprf)\} \enspace .
\]
The definitions of normalisation and expansion are extended for
$\prf{X}$,
  \begin{align*}
    \upa{\prf{X}}{d} & = \dprf& \da{\prf{X}}{d} & = \dprf&
    \Da{\prf{X}}& = \prf{\Da{X}} \enspace .&
  \end{align*}
  The semantic equations for the new constructions are
  \begin{align*}
    \semc{\boxty{A}}d &= \prf{\semc{A}d}&
    \semc{\boxtm{a}}d &= \dprf&\\
    \semc{\wheretm{b}{t}{B}}d &= \dprf&
    \semc{\oprf}d &=  \dprf \enspace . &
  \end{align*}
} 

\LONGVERSION{
\noindent
We use the following notations for the injections into $D$:
\begin{align*}
  & \iDs{d}{F} &\mbox{\hspace{.2cm}}
  & \Fst\;d, ~~ \Snd\;d &
  \mbox{\hspace{1cm}}&\mbox{for } d\in D, F \in [D \into D]
\\
  & \iZero & \mbox{\hspace{.2cm}} 
  & \iNat&  \mbox{\hspace{1cm}}& \dprf\\
  & \iSuc{d}& \mbox{\hspace{.2cm}} 
  &\prf{d} & \mbox{\hspace{1cm}}& \mbox{for } d \in D\\
  &\enumD{n} & \mbox{\hspace{.2cm}}
  &\constD{n}{i} & 
  \mbox{\hspace{1cm}}&\mbox{for } i,n \in \mathbb{N}\\
  &\iNrec{F}{d}{g}{d'} & \mbox{\hspace{.2cm}}&{ } &
  \mbox{\hspace{1cm}}&\mbox{for } d, d' \in D, F\in [D \into D], g\in [D \into [D \into D]] \\
  &\elimD{n}{F}{\vec{d}}{d'} &  \mbox{\hspace{.2cm}}&{ } &
  \mbox{\hspace{1cm}}&\mbox{for } d, d' \in D,F\in [D\into D],
  \vec{d} \in D^\omega, n\in \mathbb{N}
\end{align*}

\noindent
In this extension, the injections $\Fst$, $\Snd$, $\Nrecraw$, and
$\Elimraw$ construct neutral elements $k$. Soundness for the calculus
$(\vdashp)$ requires the canonical element for proof-irrelevant types
$(\dprf)$ to be in every PER; thus we need to redefine application
$\_\cdot\_$ to have $\dprf\in [\iPi{X}{F}]$:
\[ \dprf \cdot d = \dprf \enspace. \] We also redefine the projections
$\fstnew$ and $\sndnew$ to account for neutrals and because they are
used in the definition of $\coprod{X}{F}$, which will be used as the denotation
of sigma types.
\[
\begin{array}{lll}
  \fstnew\,d =
  \begin{cases}
    d_1 &\mbox{ if } d=(d_1,d_2) \\
    \dprf&\mbox{ if } d=\dprf \\
    \iBoolF{d}&\mbox{ otherwise}
  \end{cases}
&\quad &
  \sndnew\,d = \begin{cases}
    d_2 &\mbox{ if } d=(d_1,d_2) \\
    \dprf&\mbox{ if } d=\dprf \\
    \iBoolT{d}&\mbox{ otherwise}
  \end{cases}
\end{array}
\] 

\setlength{\originalArrayColSep}{\arraycolsep} 
\setlength\arraycolsep{0.2em}

\begin{defi}[Read-back function]\hfill
\[
  \begin{array}[t]{rcl}
    \reify{j}{(\iDs{X}{F})} &=& \DSum{(\reify{j}{X})}{\\ &&\quad(\reify{j+1}{(F\ \iVar{j})})}\\
    \reify{j}{\iNat} &= &\natty \\
    \reify{j}{\iZero} &=& \ztm \\
    \reify{j}{(\iSuc{d})} &=& \suctm{(\reify{j}{d})}\\
    \reify{j}{(\prf{d})}& =& \boxty{(\reify{j}{d})}\\
    \reify{j}{(\enumD{n})} &=& \enum{n}\\
  \end{array}
\quad
\begin{array}[t]{rcl}
    \reify{j}{\iPair{d}{d'}} &=& \depair{\reify{j}{d}}{\reify{j}{d'}} \\
    \reify{j}{(\iBoolF{d})} &=& \dfst{(\reify{j}{d})}\\
    \reify{j}{(\iBoolT{d})} &=& \dsnd{(\reify{j}{d})}\\
    \reify{j}{(\iNrec{F}{d}{f}{e})} &=& \natrec{(\reify{j+1}{(F\;\iVar{j})})}{
      \\ && \quad(\reify{j}{d})}{(\reify{j}{f})}{(\reify{j}{e})}\\
    \reify{j}{\dprf} & = &\oprf \\
    \reify{j}{(\constD{n}{i})} &=& \const{n}{i}\\
  \end{array}
  \]
  \[
  \begin{array}[t]{rcl}
    \reify{j}{(\elimD{n}{F}{\tuple{d_0,\ldots,d_{n-1}}}{e})} &=& 
    \elim{n}{(\reify{j+1}{(F\ \iVar{j})})}{(\reify{j}{d_0})\cdots(\reify{j}{d_{n-1}})}{(\reify{j}{e})}\\
  \end{array}    
  \] 
\end{defi}
\setlength{\arraycolsep}{\originalArrayColSep}


\noindent
We define inductively new PERs for interpreting naturals and finite
types. Note that $\C 0$ and $\C 1$ are irrelevant, in this way we can
model $\eta$-expansion for $\enum{0}$ and $\enum{1}$; $\smashed \X$ is
also irrelevant, even when $\X$ distinguishes its elements.

\begin{defi}[More semantic types] \label{def:pernat}\bla
  \begin{enumerate}[(1)]
  \item $\perNat$ is the smallest PER over $D$, such that
    \begin{enumerate}[(a)]
    \item $\perne \subseteq \perNat$
    \item $\iZero = \iZero \in \perNat$
    \item $\iSuc{d} = \iSuc{d'} \in \perNat $, if $d=d' \in \perNat$
    \end{enumerate}
  \item If $\X \in \perD$ then $\smashed \X := \{ (d,d') \mid d,d' \in
     \dom(\X)\cup\{\star\} \} \in \perD$.
   \item $\C 0 = |\emptyset| = \{(\dprf,\dprf)\}$,
   \item $\C 1 = |\{\constD{1}{0}\}| = \{(d,d') \mid d,d' \in
     \{\dprf,\constD 1 0\}\}$,
  \item 
    $\mathcal{C}_n = \{(\constD{n}{i},\constD{n}{i})\ |\ i < n \} \cup
    \perne$, for $n \geq 2$.
  \end{enumerate}
\end{defi}

We add new clauses in the definitions of the partial equivalences for
universe and types, these clauses do not affect the well-foundedness
of the order $\sqsubset$ defined in \ref{rem:ord-pert}, but now we
have that $\enumD{n}$ and $\iNat$ are also minimal elements for that
order.

\begin{defi}[Extension of $\perU$ and $\perT$]\hfill
\begin{enumerate}[(1)]
 \item Inductive definition of $\perU \in \perD$.
  \begin{enumerate}[(a)]
  \item If $X = X' \in \perU$, and for all $d = d' \in [X]$, $F\,d =
    F'\,d' \in \perU$, then
    $\iDs{X}{F} =\iDs{X'}{F'} \in \perU$.
  \item $\iNat = \iNat \in \perU$,
  \item $\enumD{n} = \enumD{n} \in \perU$,
  \item if $X = X' \in \perU$, then $\prf{X} = \prf{X'} \in \perU$.
  \end{enumerate}
\item Inductive definition of $\perT \in \perD$.
  \begin{enumerate}[(a)]
  \item If $X = X' \in \perT$, and for all $d = d' \in [X]$, $F\,d =
    F'\,d' \in \perT$, then $\iDs{X}{F} = \iDs{X'}{F'}\in \perT$.
  \item if $X = X' \in \perT$, then $\prf{X} = \prf{X'} \in \perT$.
  \end{enumerate}
\item Recursive definition of $[\_] \in \vdom(\perT) \to \perD$.
  \begin{enumerate}[(a)]
  \item $[\iDs{X}{F}] = \sigD{[X]}{(d\mapsto [F\ d])}$,
  \item $[\enumD{n}] = \C n$ 
  \item $[\iNat] = \perNat$,
  \item if $X \in \dom(\perT)$, then $[\prf{X}] = \{(\dprf,\dprf)\}$.
  \end{enumerate}
\end{enumerate}
\end{defi}\medskip

\noindent Note that in the PER model, all propositions $\prf{X}$ are inhabited.
In fact, all types are inhabited, for there is a reflection from
variables into any type, be it empty or not.  So, the PER model is
unsuited for refuting propositions.  However, the logical relation we
define in the next section will only be inhabited for non-empty types.

\begin{rem}
  It can be proved by induction on $X\in\perT$ that $\dprf\in [X]$.
\end{rem}

\begin{defi}[Reflection and reification, \cf~\ref{def:up-down}]
\[
\begin{array}{rcl}
    \upa{\iDs{X}{F}}{k} &=& \iPair{\upa{X}{\Fst\, k}}{\upa{F\, (\upa{X}{\Fst\, k})}{\Snd\, k}} \\
    \upa{\iNat}{k} &=& k \\
    \upa{\enumD{0}}{k} &=& \dprf  \\
    \upa{\enumD{1}}{k} &=& \constD 1 0  \\
    \upa{\enumD{n}}{k} &=& k \qquad \mbox{for } n \geq 2\\
    \upa{\prf{X}}{k}   &=& \dprf \\
\end{array}
\qquad
\begin{array}{rcl}
    \da{\iDs{X}{F}}{d} &=&  \iPair{\da{X}{\fstnew d}}{\da{F\,(\fstnew d)}{\sndnew d}} \\
    \da{\iNat}{d} &=& d \\
    \da{\enumD{0}}{d} &=& \oprf  \\
    \da{\enumD{1}}{d} &=& \constD 1 0  \\
    \da{\enumD{n}}{d} &=&d \\
    \da{\prf{X}}{d} &=& \oprf \\
\end{array}
\]
  \begin{align*}
    \Da{\iDs{X}{F}} &= \iDs{(\Da{X})}{(d\mapsto \Da{(F\, \upa{X}{d})})}&
    \Da{\iNat} &= \iNat&\\
    \Da{\enumD{n}}& = \enumD{n}&
    \Da{\prf{X}}& = \prf{(\Da{X})}&
  \end{align*}
\end{defi}\medskip



\noindent For giving semantics to eliminators for data types we need to define
partial functions $\nrecraw : [D\into D]\times D \times D\times
D\into D$, and $\erecraw : [D\into D]\times D \times D\times
D\into D$.
\begin{defi}[Eliminations on $D$]\bla
  \label{def:recd}
  \begin{enumerate}[(1)]
  \item Elimination operator for naturals.
\[
    \begin{array}{lll}
      \drec{F}{d}{f}{\dprf}   & = & \dprf\\
      \drec{F}{d}{f}{\iZero}   & = & d\\
      \drec{F}{d}{f}{\iSuc{e}} & = & (f\cdot e)\cdot \drec{F}{d}{f}{e}\\
      \drec{F}{d}{f}{k}        & = & 
        \upa{F\ k}(\Nrecraw(
        \begin{array}[t]{@{}l@{}}
          d'\mapsto \Da{F\ d'},\\
          \da{F\ \iZero}{d},   \\
          \iLam{d'\mapsto (\iLam{e' \mapsto  \da{F\ (\iSuc{d'})}{f\cdot
                d' \cdot e'}})}, \\
          k))
    \end{array}
    \end{array}
\]
  \item Elimination operator for finite types.
  \begin{align*}
    \erec{n}{F}{\tuple{d_0,\ldots,d_{n-1}}}{\dprf} &= \dprf\\
    \erec{n}{F}{\tuple{d_0,\ldots,d_{n-1}}}{\constD{n}{i}} &= d_i\\
    \erec{n}{F}{\tuple{\constD n 0,\ldots,\constD n {n-1}}}{d} &= d \\
    \erec{n}{F}{\tuple{d_0,\ldots,d_{n-1}}}{k} &= \upa{F\
      k}{\elimD{n}{e\mapsto \Da{F\ e}} {\tuple{\da{F\
            \constD{n}{0}}{d_0},\ldots,\da{F\
            \constD{n}{n-1}}{d_{n-1}}}}{k}}&
  \end{align*}
\end{enumerate}
\end{defi}

\begin{rem}
  If for all $d=d'\in\perNat$, $F\ d= F' d'\in \perT$, and $z = z'\in
  [F\ \iZero]$, and for all $d=d' \in \perNat$ and $e =e' \in [F\
  d]$, $s\cdot d\cdot e=s'\cdot d'\cdot e' \in [F (\iSuc{d})]$, and $d
  =d' \in \perNat$ then $\drec{F}{z}{s}{d}= \drec{F}{z}{s}{d'} \in [F\
  d]$.
\end{rem}

With these new definitions we can now give the semantic equations for
the new constructs. 

\begin{defi}[Extension of interpretation]
  \begin{align*}
    \semc{\DSum{A}{B}}d &=  \iDs{(\semc{A}d)}{(d'\mapsto \semc{B}\iPair{d}{d'})}& 
    \semc{\enum{n}}d & = \enumD{n}\\
    \semc{\natty}d &=  \iNat& 
    \semc{\boxty{A}}d &= \prf{\semc{A}d}&\\
    \semc{\dfst{t}}d &=  \fstnew \semc{t}d& 
    \semc{\dsnd{t}}d &=  \sndnew \semc{t}d& \\
    \semc{\depair{t}{t'}}d &= \iPair{\semc{t}d}{\semc{t'}d}&
    \semc{\ztm}d &= \iZero& \\
    \semc{\natrec{B}{z}{s}{t}}d &= \drec{e\mapsto \semc{B}{\iPair{d}{e}}}{\semc{z}d}{\semc{s}d}{\semc{t}d}&
    \semc{\suctm{t}}d &=  \iSuc{\semc{t}d}& \\
    \semc{\boxtm{a}}d &= \dprf&
    \semc{\oprf}d &=  \dprf &\\
    \semc{\wheretm{b}{t}{B}}d &= \semc b {\iPair{d}{\dprf}} 
    &
    \semc{\const{n}{i}}{d}& = \constD{n}{i}&
  \end{align*}
  \[
  \semc{\elim{n}{B}{t_0\cdots t_{n-1}}{t}}{d} = 
    \erec{n}{e\mapsto\semc{B}\iPair{d}{e}}{\tuple{\semc{t_0}d,\ldots,\semc{t_{n-1}}d}}{\semc{t}d}
    \]
\end{defi}

\begin{lem}[Laws of proof elimination]
  \label{lem:soundness-irr} $\beta$, $\eta$, and
  associativity for $\whereraw$ are modeled by the extended
  applicative structure.
\end{lem}
\begin{proof}See \ref{prf:soundness-irr}.
\end{proof}
}

\begin{rem}
\LONGSHORT{
 All of lemmata \ref{lem:famperd}, \ref{lem:reify}, and theorems
  \ref{thm:soundness}, and \ref{thm:completeness} are valid for the
  extended calculus.
}{
  All of lemmata \ref{lem:famperd}, \ref{lem:reify}, and theorems
  \ref{thm:soundness}, and \ref{thm:completeness} are valid for the
  calculus with proof-irrelevance.
}
\end{rem}

Note that we have defined a \emph{proof-irrelevant}
semantics for $(\vdashp)$ that collapses
all elements of $\boxty A$ to $\dprf$, which leads to a more efficient
implementation of the normalisation function. However, this semantics
is not sound if \lambdaIrr\ is extended with singleton types
interpreted analogously to $\C{1}$, i.e., $[\iSing{d}{X}] =
\smashed{\singD{d}{X}}$, because it does not model
\ruleref{sing-eq-el}.  (We have $d = \dprf \in [\iSing{d}{X}]$ for all
$d \in [X]$, but not necessarily $d = \dprf \in [X]$.)
On the other hand, \lambdaIrr\ without $\oprf$
can be extended to singleton types as explained in the following
remark.

\begin{rem}[Extending \lambdaIrr\ by singleton types] 
  \label{rem:proof-relevant-model}
  Singleton types can be added straightforwardly if we employ a
  \emph{proof-relevant} semantics:

  The domain $D$ is not changed; in particular we have $\dprf \in D$,
  and it is readback as before, $\reify{j}{\dprf} = \oprf$; hence
  $\dprf \in dom(\pernf)$.


  All the enumerated types are modelled in a uniform way:
  $[\enumD{n}] = \{(\constD{n}{i},\constD{n}{i})\ |\ i < n \} \cup
  \perne$; proof-irrelevance types $\boxty{A}$ are interpreted as
  the irrelevant PER with the same domain as the PER for $A$: 
  $[ \prf{X} ] = \{(d,d')\ |\ d,d' \in dom([X]) \}$. Reflection
  and reification for $\prf{X}$ are defined respectively as
  $$ \up{\prf{X}}{d} = \up{X}{d} \qquad\mbox{ and } \qquad \down{{\prf{X}}}{d} = \dprf \enspace .$$
  With these definitions it is clear that the corresponding
  result for Lem.~\ref{lem:reify} is still valid.
  
  Since $\dom([\prf{X}]) = \dom([X])$, introduction and elimination of
  proofs can be interpreted as follows
  $$ \semc{\boxtm{a}}{d} = \semc{a}{d} \quad \mbox { and } \quad
  \semc{\wheretm{b}{t}{B}}d = \semc b {\iPair{d}{\semc{t}d}} \enspace ;
  $$
  \noindent this model is sound with respect to the calculus
  $(\vdash)$ extended with singleton types; hence
  Thm. \ref{thm:completeness} is valid.
\end{rem}

\begin{rem}
  As was previously said we cannot use this PER model for proving
  that there is no closed term in $\enum{0}$. Instead, one can build
  up a PER model, in the sense of \ref{sec:persem}, of closed values,
  where $[\enum{0}]
  = \emptyset$. By soundness (Thm.~\ref{thm:soundness}) it follows
  that there is no possible derivation of $\dterm{}{\enum{0}}{t}$.
\end{rem}


%% file: per-model.tex
\newcommand{\pairapp}[2]{\mathsf{Pair}\ #1\ #2}
\newcommand{\singapp}[2]{\mathsf{Sing}\ #1\ #2}
\newcommand{\per}[1]{\mbox{PER}(#1)}
\newcommand{\calA}{\mathcal{A}}
\newcommand{\calM}{\mathcal{M}}
\newcommand{\calS}{\mathcal{S}}
\newcommand{\semcs}[1]{{}^{\mathsf{s}}\kern-0.2ex\semc{#1}}
\newcommand{\semctx}[1]{(\![{#1}]\!)}

In this subsection we introduce the abstract notion of PER models for
our theory. This subsection does not introduce any novelty (except for
some notational issues). We refer the reader to \cite{319872} for a
short report on the historical developments of PER models.

\begin{defi}[Partial Equivalence Relations]\label{sec:pers}
A partial equivalence relation (PER) over a set $\calA$ is a binary
relation over $\calA$ which is symmetric and transitive.

If $\calR$ is a PER over $\calA$, and $(a,a') \in \calR$ then it is
clear that $(a,a) \in \calR$. We define $dom(\calR) = \{ a \in \calA\
|\ (a,a) \in \calR \}$; clearly, $\calR$ is an equivalence relation
over $dom(\calR)$.  If $(a,a') \in \calR$, sometimes we will write $a
= a' \in \calR$, and $a\in \calR$ if $a \in dom(\calR)$. We denote
with $\per{\calA}$ the set of all PERs over $\calA$. Given two PERs $\calR$
and $\calR'$ over $\calA$, we say $\calR$ is included in $\calR'$ if
$(a,a') \in \calR$ implies $(a,a') \in \calR'$; we denote this
inclusion with $\calR \subseteq \calR'$.

If $\calR \in \per{\calA}$ and $\calF \colon dom(\calR) \into
\per{\calA}$, we say that $\calF$ is {\em a family of PERs indexed by
  $\calR$}\/ iff for all $a = a' \in \calR$, $\calF\ a = \calF\
a'$. If $\calF$ is a family indexed by $\calR$, we write $\calF \colon
\calR \into \per{\calA}$.
\end{defi}

\begin{defi}[Applicative structure]
  An applicative structure is given by a pair $\calA = (\calA,\cdot)$,
  where $\calA$ is a set and $\cdot$ is a binary operation on $\calA$.
\end{defi}

The following definitions are standard (\eg\
\cite{aspinall:csl94,coquandPollackTakeyama:fundinf05}) in definitions
of PER models for dependent types. The first one is even standard for
non-dependent types (\cf\ \cite{DBLP:books/el/leeuwen90/Mitchell90})
and ``F-bounded polymorphism'' (\cite{143230}); its definition clearly
shows that equality is interpreted extensionally for dependent
function spaces. The second one is the PER corresponding to the
interpretation of singleton types; it has as its domain all the
elements related to the distinguished element of the singleton, and it
relates everything in its domain.

\begin{defi} Let $\calA$ be an applicative structure,
  $\calX\in\per{\calA}$, and $\calF\in \calX\into \per{\calA}$.
  \begin{enumerate}[(1)]
  \item $\prod\,\calX\,\calF = \{(f,f')\ |\ f\cdot a = f'\cdot a' \in \calF\ a,
    \mbox{ for all } a = a' \in \calX\}$;
  \item $\singD{a}{\calX} = \{(b,b')\ |\ a = b \in \calX \mbox{ and } a = b'
    \in \calX\}$.
  \end{enumerate}
\end{defi}

 Besides interpreting function spaces and singletons we need PERs for
the denotation of the universe of small types, and for the set of
large types; jointly with these PERs we need functions assigning a PER
for each element in the domain of these universe PERs. Note that this
forces the applicative structure to have some distinguished elements.

\begin{defi}[Universe]
  Given an applicative structure $\calA$ with distinguished elements
  $\mathsf{Fun}$ and $\mathsf{Sing}$, a universe $(\perU, [\_])$ is a
  PER $\perU$ over $\calA$ and a family $[\_] \colon \perU \into \Per(\calA)$
  with the condition that $\perU$ is closed under function and 
  singleton types.  This means:
  \begin{enumerate}[(1)]
  \item Whenever $X = X' \in \perU$ and
    for all $a = a' \in [X]$, $F\ a = F'\ a' \in \perU$, then $\F{X}{F}
    = \F{X'}{F'} \in \perU$, with $[\F{X}{F}] = \prod\,[X]\,(a \mapsto [F\ a])$.
  \item Whenever $X = X' \in \perU$ and $a = a' \in [X]$, then $\iSing a
    X = \iSing{a'}{X'} \in \perU$ and $[\iSing a X] = \singD{a}{[X]}$.
  \end{enumerate}
\end{defi}\medskip

\noindent An applicative structure paired with one universe for small
types and one universe for large types is the minimal structure needed
for having a model of our theory.

\begin{defi}[PER model]
  Let $\calA$ be an applicative structure with distinguished 
  elements $\iU,\mathsf{Fun}$, and $\mathsf{Sing}$; a {\em
    PER model} is a tuple $(\calA,\perU,\perT,[\_])$ satisfying:
  \begin{enumerate}[(1)]
  \item $\perU\subset \perT \in PER(\calA)$, such that
    $(\perT,[\_])$ and $(\perU, [\_] \!\restriction_{\perU})$ are
    both universes, and
  \item $\iU \in dom(\perT)$, with $[\iU] = \perU$.
  \end{enumerate}
\end{defi}

In the following definition we introduce an abstract concept for
environments: since variables are represented as projection functions
from lists (think of $\q$ as taking the head of a list, and $\p$ as
taking the tail), it is enough having sequences together with
projections.

\begin{defi}[Sequences]
  Given a set $\calA$, a set $\Aseq$ \emph{has sequences} over $\calA$ if
  there are distinguished operations $\iO : \Aseq$, $\tPair : \Aseq
  \times \calA \to \Aseq$, $\tfst : \Aseq \to \Aseq$,  and $\tsnd : \Aseq
  \to \calA$ such that
  \begin{eqnarray*}
    &\tfst\, (\tPair\, a\, b)= a \\
    &\tsnd\, (\tPair\, a\, b)= b .
  \end{eqnarray*}
\end{defi}

Now we need to extend the notion of PERs over $\calA$ to PERs over
$\Aseq$ for interpreting substitutions.\footnote{The reader is invited
  to think of $\one$ as the terminal object of the category of PERs
  over $\Aseq$ and PER preserving morphisms; looked this way our
  definition for $\one$ does not differ very much from others
  \cite{aspinall:csl94,coquandPollackTakeyama:fundinf05}.} 

\begin{defi} Let $\calA$ be an applicative structure and let $\Aseq$
  have sequences over $\calA$; moreover let $\calX\in\per{\Aseq}$ and $\calF\in \calX\into
  \per{\calA}$.
  \begin{enumerate}[(1)]
   \item $\one = \{(\iO,\iO)\}$;
   \item $\sigD{\calX}{\calF} = \{(a,a')\ |\ \tfst\ a = \tfst\ a' \in \calX \mbox{ and }
     \tsnd\ a = \tsnd\ a' \in \calF\ (\tfst\ a)\}$;
  \end{enumerate}
\end{defi}

Until here we have introduced semantic concepts. Now we are going to
axiomatise the notion of evaluation, connecting the syntactic realm
with the semantic one.

\begin{defi}[Environment model]\label{def:syn-app-str}
  Let $(\calA,\perU,\perT,[\_])$ be a PER model and let $\Aseq$ have
  sequences over $\calA$.  We call $\calM = (\calA, \perU, \perT,
  [\_], \Aseq, \semc{\,},\semcs{\,}{})$ an {\em environment model} if
  the {\em evaluation functions} $\semc{\_}\_\colon \terms\times
  \Aseq \to \calA$ and $\semcs{\_}{\_}\colon \terms\times \Aseq \to
  \Aseq$ satisfy:
  \newlength{\originalArrayColSep} 
  \setlength{\originalArrayColSep}{\arraycolsep} 
  \setlength\arraycolsep{0.3em}
  \[
  \begin{array}[t]{rcl}
    \semcs{\idsubs{}}a &=& a\\
    \semcs{\esubs{}}a &=& \iO\\
    \semcs{\subsc{\sigma}{\delta}}a &=& \semc{\sigma}(\semc{\delta}a)\\
    \semcs{\exsubs{\sigma}{t}}a &=& \pairapp{(\semc{\sigma}a)}{(\semc{t}a)}\\
    \semcs{\p}a &=& \tfst\ a
  \end{array}
  \quad
  \begin{array}[t]{rcl}
  \semc{\TmU}a &=& \iU\\
  \semc{\F{A}{B}}a &=& \F{(\semc{A}a)}{F}\text{, where }F\;b = \semc{B}{(\pairapp{a}{b})}\\
  \semc{\singTm{t}{A}}a &=& \singapp{(\semc{t}a)}{(\semc{A}a)}\\
  \semc{\subsTm{t}{\sigma}}a &=& \semc{t}(\semcs{\sigma}a)\\
  \semc{\lambda{t}}a &=& f \text{, where } f\cdot b = \semc{t}{(\pairapp{a}{b})}\\
  \semc{\appTm{t}{u}}a &=& (\semc{t}a)\cdot (\semc{u}a)\\
  \semc{\q}a &=& \tsnd \ a
 \end{array}
 \]
  \setlength\arraycolsep{\originalArrayColSep}
Since no ambiguities arise, we shall henceforth
write $\semc\sigma$ instead of $\semcs\sigma$.
\end{defi}

Once we have an environment model $\calM$ we can define the denotation
for contexts. The second clause in the next definition is not
well-defined a priori; its totality is a corollary of
Thm. \ref{thm:soundness}.

\begin{defi}Given an environment model $\calM$, we define recursively
  the semantic of contexts $\semctx{\_}\colon\ctx \into\per{\Aseq}$:
  \begin{enumerate}[(1)]
  \item $\semctx{\ectx} =\one$,
  \item $\semctx{\ctxe{\Gamma}{A}} = \sigD{\semctx{\Gamma}(a\mapsto [\semc{A}a])}$.
  \end{enumerate}  
\end{defi}\medskip

\noindent We use PERs for validating equality judgements and the domain of each
PER for validating typing judgements.

\begin{defi}[Validity] Let $\calM$ be an environment model. We
  define inductively the predicate of satisfability of judgements by
  the model, denoted with $\Gamma\vDash^{\calM} J$:
  \begin{enumerate}[(1)]
  \item $\ectx \vDash$ iff true
  \item $\ctxe{\Gamma}{A}\vDash$ iff $\Gamma\vDash A$
  \item $\Gamma\vDash A$ iff $\Gamma\vDash A = A$
  \item $\Gamma\vDash A=A'$ iff $\Gamma\vDash$ and for all $d = d' \in
    \semctx{\Gamma}$, $\semc{A}{d} = \semc{A'}{d'} \in \perT $
  \item $\Gamma\vDash t: A$ iff $\Gamma\vDash t = t:A$
  \item $\Gamma\vDash t=t':A$ iff $\Gamma\vDash A$ and for all $d =
    d'\in \semctx{\Gamma}$, $\semc{t}{d} = \semc{t'}{d'} \in
    [\semc{A}{d}] $
  \item $\Gamma\vDash \sigma: \Delta$ iff $\Gamma\vDash \sigma =
    \sigma:\Delta$
  \item $\Gamma\vDash \sigma=\sigma':\Delta$ iff $\Gamma\vDash$,
    $\Delta\vDash$, and for all $d = d' \in \semctx{\Gamma}$,
    $\semc{\sigma}{d} = \semc{\sigma'}{d'}\in \semctx{\Delta}$.
  \end{enumerate}
\end{defi}

\begin{thm}[Soundness of the Judgements]
  \label{thm:soundness} Let $\calM$ be a model. If $\Gamma \vdash J$,
  then $\Gamma\vDash^{\calM} J$.
\end{thm}
\begin{proof} By easy induction on $\Gamma \vdash J$. 
\end{proof}


%% file: tlca09-nbe.tex
\section{Correctness of NbE}
\label{sec:logrel}

\LONGSHORT{

  \noindent In Thm.~\ref{thm:completeness} we have proved that the NbE algorithm
  is complete with respect to the judgemntal equality of our calculi;
  a corollary of that fact is totality of NbE.

  \begin{rem} Let $\dterm{\Gamma}{A}{t}$. Given some $d \in
    \semctx{\Gamma}$, we can conclude
    $\reify{i}{(\da{\semc{A}d}{\semc{t}d})}$ is a well-defined term 
    in normal form.
  \end{rem}

  In this section we prove correctness, with respect to the typing
  rules, for our NbE algorithm. This means that given a typing
  $\dterm{\Gamma}{A}{t}$, when NbE is applied to $t$, the resulting
  normal form $v$, is provable equal to $t$; i.e.
  $\deqterm{\Gamma}{A}{t}{v}$.

  Let us anticipate the main results of this section. As a corollary of
  Thm.~\ref{thm:logrel} we show that a term is related to its
  denotation with respect to some canonical environment (to be defined
  in Def.~\ref{def:canonical-env}). Previously we prove in Lem.
  \ref{lem:judgeq} that if a term is logically related with some
  semantic element, then its reification will be judgmentally equal
  to the term. Composing these facts we obtain correctness. As a
  consequence of having correctness and completeness for NbE, one gets
  decidability for judgmentally equality: normalise both terms and
  check they are syntactically the same. Another important corollary
  is injectivity for constructors.

  \subsection{Logical relations}
  \label{sec:logrel-sing}
  In this subsection we define logical relations and prove some
  technical lemmas about them. As is standard with logical relations
  one defines them by induction on types (here we define by induction
  on semantics of types, \ie\  elements of $\perT$) and for basic types
  they are defined by prescribing the property to be proved; while for
  higher order types they are defined using the relations of the 
  domain and image types.

}{ 

  In order to prove soundness of our normalisation algorithm we
  define logical relations \cite{kripke-models} between types and
  elements in the domain of $\perT$, and between terms and elements in
  the domain of the PER corresponding to elements of $\perT$.

}

\begin{defi}[Logical relations]
  \label{def:logrel} We define simultaneously two families of binary
   relations:
  \begin{enumerate}[(a)]
  \item If $\dctx{\Gamma}$ then $(\Gamma \derN \_ \rel \_ \in \perT) \subseteq \{ A \mid
    \dtype{\Gamma}{A}\} \times \perT$ 
    shall be a $\Gamma$-indexed family of relations
    between well-formed syntactic types $A$ and type values $X$.
  \item If $\dtype{\Gamma}{A}\rel X\in \perT$ then $(\Gamma \derN \_ :
    A \rel \_ \in [X]) \subseteq
    \{ t \mid \dterm{\Gamma}{A}{t}\} \times [X]$
    shall be a $(\Gamma,A,X)$-indexed family of relations between
    terms $t$ of type $A$ and values $d$ in PER $[X]$.
  \end{enumerate}
  These relations are defined simultaneously by induction on $X \in
  \perT$.
  \begin{enumerate}[(1)]
  \item Neutral types: $X \in \perne$.
    \begin{enumerate}[(a)]
    \item $\dtype{\Gamma}{A}\rel X \in \perT$ iff
      for all $\Delta \leqslant^i \Gamma$,
      $\deqtype{\Delta}{\lift{i}{A}}{\reifyC{\Delta}{\Da{X}}}$.
    \item $\dterm{\Gamma}{A}{t}\rel d\in [X]$ iff  $\dtype{\Gamma}{A}\rel X \in \perT$, and
      for all $\Delta \leqslant^i \Gamma$,
      $\deqterm{\Delta}{\lift{i}{A}}{\lift{i}{t}}{\reifyC{\Delta}{\da{X}{d}}}$.
    \end{enumerate}
  \item Universe.
    \begin{enumerate}[(a)]
    \item $\dtype{\Gamma}{A}\rel \iU \in \perT$ iff
      $\deqtype{\Gamma}{A}{\TmU}$.
    \item $\dterm{\Gamma}{A}{t}\rel X\in [\iU]$ iff
      $\deqtype{\Gamma}{A}{\TmU}$, and $\dtype{\Gamma}{t}\rel
      X\in\perT$.
    \end{enumerate}
  \item Singletons.
    \begin{enumerate}[(a)]
    \item $\dtype{\Gamma}{A}\rel \iSing{d}{X} \in \perT$ iff
      $\deqtype{\Gamma}{A}{\singTm{a}{A'}}$ and
      $\dterm{\Gamma}{A'}{a}\rel d\in [X]$.
    \item $\dterm{\Gamma}{A}{t}\rel d'\in [\iSing{d}{X}]$ iff
      $\deqtype{\Gamma}{A}{\singTm{a}{A'}}$ and
      $\dterm{\Gamma}{A'}{t}\rel d\in [X]$, and
      $\dtype{\Gamma}{A'}\rel X\in\perT$.
    \end{enumerate}
  \item Function spaces.
    \begin{enumerate}[(a)]
    \item $\dtype{\Gamma}{A}\rel \iPi{X}{F} \in \perT$ iff
      $\deqtype{\Gamma}{A}{\F{A'}{B}}$, and $\dtype{\Gamma}{A'}\rel
      X\in \perT$, and
      $\dtype{\Delta}{\subsTy{B}{\exsubs{\p^i}{s}}}\rel F\,d\in \perT$
      for all $\Delta\leqslant^i\Gamma$ and
      $\dterm{\Delta}{\lift{i}{A'}}{s}\rel d\in [X]$.
    \item $\dterm{\Gamma}{A}{t}\rel f\in [\iPi{X}{F}]$ iff
      $\deqtype{\Gamma}{A}{\F{A'}{B}}$, $\dtype{\Gamma}{A'}\rel X$,
      and $\dterm{\Delta}{\subsTy{B}{\exsubs{\p^i}{s}}}
      {\appTm{(\lift{i}{t})}{s}}\rel f\cdot d \in [F\,d]$ for all
      $\Delta\leqslant^i\Gamma$ and
      $\dterm{\Delta}{\lift{i}{A'}}{s}\rel d\in [X]$.
    \end{enumerate}
  \end{enumerate}
\end{defi}\medskip

\noindent The following technical lemmata show that the logical relations are
preserved by judgmental equality, weakening of the judgement, and the
equalities on the corresponding PERs.   These lemmata are proved
simultaneously for types and terms.

\begin{lem}[Closure under conversion]
  \label{lem:logrelEqTy} Let $\dtype{\Gamma}{A}\rel X\in \perT$ 
  and $\deqtype{\Gamma}{A}{A'}$.  Then,
  \begin{enumerate}[\em(a)]
  \item $\dtype{\Gamma}{A'}\rel X\in\perT$, and
  \item if $\dterm{\Gamma}{A}{t}\rel d\in [X]$ and $\deqterm{\Gamma}{A}{t}{t'}$ 
  then $\dterm{\Gamma}{A'}{t'}\rel d\in[X]$.
  \end{enumerate}
\end{lem}
\LONGVERSION{
\begin{proof}By induction on $X\in \perT$. See \ref{prf:logrelEqTy}.
\end{proof}
}

\begin{lem}[Monotonicity]
  \label{lem:logrelMon}
  Let $\Delta\leqslant^i\Gamma$, then
  \begin{enumerate}[\em(a)]
  \item if $\dtype{\Gamma}{A}\rel X\in\perT$, then
    $\dtype{\Delta}{\lift{i}{A}}\rel X\in\perT$; and
  \item if $\dterm{\Gamma}{A}{t}\rel d\in [X]$, then
    $\dterm{\Delta}{\lift{i}{A}}{\lift{i}{t}}\rel d\in[X]$.
  \end{enumerate}
\end{lem}
\LONGVERSION{
  \begin{proof}By induction on $X \in \perT$. See \ref{prf:logrelMon}.
  \end{proof}
}

\begin{lem}[Closure under PERs]
  \label{lem:logrelEqD}
  Let $\dtype{\Gamma}{A}\rel X\in \perT$, then
  \begin{enumerate}[\em(a)]
  \item if $X = X' \in\perT$, then $\dtype{\Gamma}{A}\rel X'\in\perT$;
    and
  \item if $\dterm{\Gamma}{A}{t}\rel d\in [X]$ and 
    $d=d' \in [X]$, then $\dterm{\Gamma}{A}{t}\rel d' \in [X]$.
  \end{enumerate}
\end{lem}
\LONGVERSION{
\begin{proof}By induction on $X = X' \in \perT$. See \ref{prf:logrelEqD}.
\end{proof}
}

The following lemma plays a key r\^ole in the proof of soundness. It
proves that if a term is related to some element in (some PER), then
it is convertible to the reification of the corresponding element in
the PER of normal forms.

\begin{lem}
  \label{lem:judgeq}
  Let $\dtype{\Gamma}{A}\rel X\in\perT$.  Then, 
  \begin{enumerate}[\em(a)]
  \item $\deqtype{\Gamma}{A}{\reifyC{\Gamma}{\Da{X}}} $,
  \item if $\dterm{\Gamma}{A}{t}\rel d \in [X]$ then 
     $\deqterm{\Gamma}{A}{t}{\reifyC{\Gamma}{\da{X}{d}}}$; and
  \item if 
          $k \in \perne$ and for all $\Delta \leqslant^i \Gamma$,
    $\deqterm{\Delta}{\lift{i}{A}}{\lift{i}{t}}{\reifyC{\Delta}{k}}$,
    then $\dterm{\Gamma}{A}{t} \rel \upa{X}{k} \in [X]$.
  \end{enumerate}
\end{lem}
\LONGVERSION{
\begin{proof}By induction on $X\in \perT$. See \ref{prf:judgeq}.
\end{proof}
}

In order to finish the proof of soundness we have to prove that each
well-typed term (and each well-formed type) is logically related to
its denotation; with that aim we extend the definition of logical
relations to substitutions and prove the fundamental theorem of
logical relations.

\begin{defi}[Logical relation for substitutions] \hfill
  \label{logrelsubs}
  \begin{enumerate}[(1)]
  \item $\dsubs{\Gamma}{\ectx}{\sigma}\rel d \in\one$ always holds.
  \item $\dsubs{\Gamma}{\ctxe{\Delta}{A}}{\exsubs{\sigma}{t}} \rel
    \iPair{d}{d'} \in \sigD{\calX}{(d \mapsto [F\ d])}$ iff
    $\dsubs{\Gamma}{\Delta}{\sigma} \rel d \in \calX$,
    $\dtype{\Gamma}{\subsTy{A}{\sigma}}\rel F\ d\in \perT$, and
    $\dterm{\Gamma}{\subsTy{A}{\sigma}}{t} \rel d'\in [F\ d]$.
  \end{enumerate}
\end{defi}\medskip

\LONGSHORT{

\noindent By the way this relation is defined, the counterparts of
~\ref{lem:logrelEqTy},~\ref{lem:logrelMon}, and~\ref{lem:logrelEqD} are
easily proved by induction on the co-domain of the substitutions.

\begin{rem}
  \label{lem:logrelEqSub}
  If $\deqsubs{\Gamma }{\Delta}{\gamma}{\delta}$, and $\dsubs{\Gamma
  }{\Delta}{\gamma} \rel d \in \calX$, then $\dsubs{\Gamma
  }{\Delta}{\delta}\rel d \in \calX $.
\end{rem}
\begin{rem}
  \label{lem:monsubs}
  If $\dsubs{\Gamma}{\Delta}{\delta} \rel d \in \calX$, then
  for any $\Theta\leqslant^i\Gamma$,
  $\dsubs{\Theta}{\Delta}{\subsc{\delta}{\p^i}} \rel d \in
  \calX$.
\end{rem}
\begin{rem}
  \label{lem:logrelSubEqD}
  If $\dsubs{\Gamma }{\Delta }{\gamma}\rel d\in \calX $, and $d = d' \in
  \calX$, then $\dsubs{\Gamma }{\Delta }{\gamma} \rel d' \in \calX$.
\end{rem}
}{
  After proving the counterparts of \ref{lem:logrelEqTy},
  \ref{lem:logrelMon} and \ref{lem:logrelEqD} for substitutions, we
  can proceed with the proof of the main theorem of logical relations.
}

\begin{thm}[Fundamental theorem of logical relations]
  \label{thm:logrel}
  Let $\dsubs{\Delta}{\Gamma}{\delta}\rel d\in \semctx{\Gamma}$.
  \begin{enumerate}[\em(1)]
  \item If $\dtype{\Gamma}{A}$, then
    $\dtype{\Delta}{\subsTy{A}{\delta}}\rel \semc{A}{d}\in \perT$;
  \item if $\dterm{\Gamma}{A}{t}$, then
    $\dterm{\Delta}{\subsTy{A}{\delta}}{\subsTm{t}{\delta}} \rel
    \semc{t}{d} \in [\semc{A}{d}]$; and
  \item if $\dsubs{\Gamma}{\Theta}{\gamma}$ then
    $\dsubs{\Delta}{\Theta}{\subsc{\gamma}{\delta}} \rel
    \semc{\gamma}{d} \in \semctx{\Theta}$.
  \end{enumerate}
\end{thm}
\LONGVERSION{
\begin{proof}By mutual induction on the derivations. See \ref{prf:logrel}.
\end{proof}
}

\LONGSHORT{ We define for each context $\Gamma$ an element
  $\rho_\Gamma$ of $D$. This environment will be used to define the
  normalisation function.

}{
We define for each context $\Gamma$ an element $\rho_\Gamma$ of $D$,
that is, by construction, logically related to $\idsubs{\Gamma}$. This
environment will be used to define the normalisation function; also
notice that if we instantiate Thm.~\ref{thm:logrel} with
$\rho_\Gamma$, then a well-typed term under $\Gamma$ will be logically
related to its denotation.
}
\begin{defi}[Canonical environment]
  \label{def:canonical-env}
\LONGSHORT{
We define $\rho_\Gamma$ by induction on $\Gamma$ as follows:
      \label{eq:env-ctx}
    \begin{align*}
      \rho_\ectx &=  \iO&\\
      \rho_{\ctxe{\Gamma}{A}} & =
      \iPair{d'}{\upa{\semc{A}{d'}}{\iVar{n}}}&    \mbox{ where } n =
      |\Gamma|, \mbox{ and } d' = \rho_\Gamma.
    \end{align*}
}{
  Let $\rho_\Gamma = P_\Gamma\ \iO$, where
      $P_\ectx \ d = d$ and
      $P_{\ctxe{\Gamma}{A}}\ d =
      \iPair{d'}{\upa{\semc{A}{d'}}{\iVar{|\Gamma|}}}$ with $d' = P_\Gamma\ d$.

Then $\dsubs{\Gamma}{\Gamma}{\idsubs{\Gamma}}\rel{\rho_\Gamma} \in
\semc{\Gamma}{}$ for $\Gamma \in \ctx$.
}
\end{defi}\medskip

\LONGVERSION{
\noindent By an immediate induction on contexts we can check the following.

\begin{lem}\label{lem:logrelIdSub} If $\dctx{\Gamma}$ then 
  $\dsubs{\Gamma}{\Gamma}{\idsubs{\Gamma}}\rel{\rho_\Gamma} \in
  \semctx{\Gamma}$. 
\end{lem}
\begin{proof}By induction on $\dctx{\Gamma}$. See \ref{prf:logrelIdSub}
\end{proof}

\subsection{Main results}

Now we can define concretely the normalisation function as the
composition of reification with normalisation after evaluation under
the canonical environment. The following corollaries just instantiate
previous lemmata and theorems concluding correctness of NbE.

}

\begin{defi}[Normalisation algorithm] 
  \label{def:nbe-alg}
  Let $\dtype{\Gamma}{A}$, and $\dterm{\Gamma}{A}{t}$.
  \begin{align*}
    \nbety{\Gamma}{A} &= \reifyC{\Gamma}{\Da{\semc{A}{\rho_\Gamma}}} &\\
    \nbetm{\Gamma}{A}{t} &= \reifyC{\Gamma}{\da{\semc{A}{\rho_\Gamma}}{\semc{t}\rho_\Gamma}} &
  \end{align*}
\end{defi}\medskip

\LONGSHORT{

  \noindent Notice that if we instantiate Thm.~\ref{thm:logrel} with
  $\rho_\Gamma$, then a well-typed term $t$ under $\Gamma$ will be
  logically related to its denotation. Finally, using the key lemma
  \ref{lem:judgeq} we conclude correctness for NbE.

  \begin{cor}
    Let $\dtype{\Gamma}{A}$, and $\dterm{\Gamma}{A}{t}$, then by
    fundamental theorem of logical relations (and Lem.~\ref{lem:logrelEqTy}),
    \begin{enumerate}[\em(1)]
    \item $\dtype{\Gamma}{A}\rel \semc{A}\rho_\Gamma \in
      \perT$; and
    \item $\dterm{\Gamma}{A}{t}\rel \semc{t}\rho_\Gamma \in
      [\semc{A}\rho_\Gamma]$,
    \end{enumerate}
  \end{cor}

  \begin{cor}[Soundness of NbE]\label{rem:nbe-eq}
    By way of Lem.~\ref{lem:judgeq}, it follows immediately
    \begin{enumerate}[\em(1)]
    \item $\deqtype{\Gamma}{A}{\nbe{A}}$, and
    \item $\deqterm{\Gamma}{A}{t}{\nbe{t}}$.
    \end{enumerate}
  \end{cor}

}{ The first point of soundness is a direct consequence of
  Thm.~\ref{thm:logrel} and Lem.~\ref{lem:logrelEqTy}; and the second
  point is obtained using Lem.~\ref{lem:judgeq}.
  \begin{cor}[Soundness of NbE]
      \label{rem:nbe-eq}
      Let $\dtype{\Gamma}{A}$, and $\dterm{\Gamma}{A}{t}$, then
      \begin{enumerate}
      \item $\dtype{\Gamma}{A}\rel \semc{A}\rho_\Gamma \in \perT$, and
        $\dterm{\Gamma}{A}{t}\rel \semc{t}\rho_\Gamma \in
        [\semc{A}\rho_\Gamma]$; and
      \item $\deqtype{\Gamma}{A}{\nbe{A}}$, and
        $\deqterm{\Gamma}{A}{t}{\nbe{t}}$.
      \end{enumerate}
  \end{cor}
}
\LONGSHORT{

  We have now a decision procedure for judgmental equality; for
  deciding $\deqterm{\Gamma}{A}{t}{t'}$, put both terms in normal
  formal and check if they are syntactically equal.

  \begin{cor}
    \label{cor:termrel}
    If $\dtype{\Gamma}{A}$, and $\dtype{\Gamma}{A'} $, then we can decide
    $\deqtype{\Gamma}{A}{A'}$. Also if $\dterm{\Gamma}{A}{t}$, and
    $\dterm{\Gamma}{A}{t'}$, we can decide $\deqterm{\Gamma}{A}{t}{t'}$.
  \end{cor}

  As a byproduct we can conclude that type constructors are injective;
  this result is exploited in the next section where we introduce the
  type-checking algorithm.  Injectivity of $\F{\_}{\_}$ plays a key
  r\^ole in all versions of dependent type theory with equality as
  judgement;  \cf\ Adams' \cite{adams} proof of equivalence between
  PTS with equality as a judgement and equality taken as a relation
  between untyped terms, improved by Siles and 
  Herbelin~\cite{silesHerbelin:lics10}.
  
  \begin{rem}
    \label{rem:nbe-hom}
    By expanding definitions, we easily check 
    \begin{enumerate}[(1)]
    \item $\nbety{\Gamma}{\F{A}{B}} = \F{(\nbety{\Gamma}{A})}{(\nbety{\ctxe{\Gamma}{A}}{B})}$, and
    \item $\nbety{\Gamma}{\singTm{a}{A}} =
      \singTm{\nbetm{\Gamma}{A}{a}}{\nbety{\Gamma}{A}}$.
    \end{enumerate}
  \end{rem}

  \begin{cor}[Injectivity of $\F{\_}{\_}$ and of $\singTm{\_}{\_}$]
    \label{injcons}
    If $\deqtype{\Gamma}{\F{A}{B}}{\F{A'}{B'}}$, then
    $\deqtype{\Gamma}{A}{A'}$, and
    $\deqtype{\ctxe{\Gamma}{A}}{B}{B'}$.  Also
    $\deqtype{\Gamma}{\singTm{t}{A}}{\singTm{t'}{A'}}$, then
    $\deqtype{\Gamma}{A}{A'}$, and $\deqterm{\Gamma}{A}{t}{t'}$.
  \end{cor}

}{

  \begin{rem}
    \label{rem:nbe-hom}
    By expanding the definitions, we easily check 
    \begin{enumerate}
    \item $\nbety{\Gamma}{\F{A}{B}} = \F{(\nbety{\Gamma}{A})}{(\nbety{\ctxe{\Gamma}{A}}{B})}$, and
    \item $\nbety{\Gamma}{\singTm{a}{A}} =
      \singTm{\nbetm{\Gamma}{A}{a}}{\nbety{\Gamma}{A}}$.
    \end{enumerate}
  \end{rem}

  \begin{cor}
    \label{cor:termrel}
    If $\dtype{\Gamma}{A}$, and $\dtype{\Gamma}{A'} $, then we can decide
    $\deqtype{\Gamma}{A}{A'}$. Also if $\dterm{\Gamma}{A}{t}$, and
    $\dterm{\Gamma}{A}{t'}$, we can decide $\deqterm{\Gamma}{A}{t}{t'}$.
  \end{cor}

  \begin{cor}[Injectivity of $\F{\_}{\_}$ and of $\singTm{\_}{\_}$]
    \label{injcons}
    If $\deqtype{\Gamma}{\F{A}{B}}{\F{A'}{B'}}$, then
    $\deqtype{\Gamma}{A}{A'}$, and
    $\deqtype{\ctxe{\Gamma}{A}}{B}{B'}$.  Also
    $\deqtype{\Gamma}{\singTm{t}{A}}{\singTm{t'}{A'}}$, then
    $\deqtype{\Gamma}{A}{A'}$, and $\deqterm{\Gamma}{A}{t}{t'}$.
  \end{cor}
}


\PrfIrrTitle
\label{sec:pi-nbe}

\LONGSHORT{ 

  In this section we introduce the logical relations for the new types
  in \lambdaIrr. We skip the re-statement of the results given for
  \lambdaSing\ in \ref{sec:logrel-sing}, instead we present in Appendix
  \ref{sec:proofs} the proof for some of the new cases arising in this
  calculus for each of lemmata \ref{lem:logrelEqTy},
  \ref{lem:logrelMon}, \ref{lem:logrelEqD}, \ref{lem:judgeq} and
  theorem \ref{thm:logrel}.

\begin{defi}[\cf\ \ref{def:logrel}]\hfill
  \begin{enumerate}[(1)]
  \item Sigma types.
    \begin{enumerate}[(a)]
    \item $\dtype{\Gamma}{A}\rel \iDs{X}{F}$ iff
      $\deqtype{\Gamma}{A}{\DSum{A'}{B'}}$ and
      $\dtype{\Gamma}{A'}\rel X$ and for all $\Delta \leqslant^i
      \Gamma$ and $\dterm{\Delta}{\lift{i}{A'}}{s}\rel d\in [X]$,
      $\dtype{\Delta}{\subsTy{B'}{\exsubs{\p^i}{s}}}\rel F\,d$.
    \item $\dterm{\Gamma}{A}{t}\rel d \in [\iDs{X}{F}]$ iff
      $\deqtype{\Gamma}{A}{\DSum{A'}{B'}}$ and
      $\dterm{\Gamma}{A'}{\dfst{t}}\rel \fstnew d\in [X]$ and
      $\dterm{\Gamma}{\subsTy{B'}{\subid{\Gamma}{\dfst{t}}}}
      {\dsnd{t}}\rel \sndnew d\in [F\, (\fstnew d)]$.
    \end{enumerate}
  \item Natural numbers.
    \begin{enumerate}[(a)]
    \item $\dtype{\Gamma}{A}\rel \iNat$ iff
      $\deqtype{\Gamma}{A}{\iNat}$.
    \item $\dterm{\Gamma}{A}{t}\rel d \in [\iNat]$ iff
      $\dtype{\Gamma}{A}\rel \iNat$ and for all $\Delta \leqslant^i
      \Gamma$,
      $\deqterm{\Delta}{\iNat}{\lift{i}{t}}{\reifyC{\Delta}{d}}$.
    \end{enumerate}
  \item Finite types.
    \begin{enumerate}[(a)]
    \item $\dtype{\Gamma}{A}\rel \enumD{n}$ iff
      $\deqtype{\Gamma}{A}{\enum{n}}$.
    \item $\dterm{\Gamma}{A}{t}\rel d \in [\enumD{n}]$ iff
      $\dtype{\Gamma}{A}\rel \enumD{n}$ and for all $\Delta
      \leqslant^i \Gamma$,
      $\deqterm{\Delta}{\enum n}{\lift{i}{t}}{\reifyC{\Delta}{d}}$.
    \end{enumerate}
  \item Proof-irrelevance types.
    \begin{enumerate}[(a)]
    \item $\dtype{\Gamma}{A}\rel \prf{X}\in \perT$ iff
      $\deqtype{\Gamma}{A}{\boxty{A'}}$ and $\dtype{\Gamma}{A'}\rel
      X\in \perT$.
    \item $\dterm{\Gamma}{A}{t}\rel d \in [\prf{X}]$ iff
      $\dtype{\Gamma}{A}\rel \prf{X}$.
    \end{enumerate}
  \end{enumerate}
\end{defi}
}{
  We add the corresponding cases in the definition of logical relations,
  \begin{align*}
    &\dtype{\Gamma}{A}\rel \prf{X}\in \perT , \mbox{ iff }
    \deqtype{\Gamma}{A}{\boxty{A'}} , \mbox{ and }\dtype{\Gamma}{A'}\rel
    X\in \perT ; \mbox{ and}\\
  &\dterm{\Gamma}{A}{t}\rel d \in [\prf{X}], \mbox{ iff }
    \dtype{\Gamma}{A}\rel \prf{X} \in \perT .
  \end{align*}
}

\LONGSHORT{

  \begin{rem}\hfill
    \begin{enumerate}[(1)]
    \item $\nbety{\Gamma}{\DSum{A}{B}}=\DSum{\nbety{\Gamma}{A}}{\nbety{\ctxe{\Gamma}{A}}{B}}$;
    \item $\nbetm{\Gamma}{\DSum{A}{B}}{\depair{t}{b}}=\depair{\nbetm{\Gamma}{A}{t}}
      {\nbetm{\Gamma}{\subsTy{B}{\subid{}{t}}}{b}}$;
    \item $\nbe{\suctm{t}} = \suctm{\nbe{t}}$.
    \item $\nbe{\boxty{A}} = \boxty{\nbe{A}}$.
    \end{enumerate}
  \end{rem}
}{
  \begin{rem}
    All the lemmata \ref{lem:logrelEqTy}, \ref{lem:logrelMon},
    \ref{lem:logrelEqD}, \ref{lem:judgeq}, theorem \ref{thm:logrel},
    and remarks \ref{rem:nbe-eq}, \ref{rem:nbe-hom} are still
    valid. Moreover we also have $\nbe{\boxty{A}} = \boxty{(\nbe{A})}$.
  \end{rem}
}

\LONGVERSION{
\begin{cor}
  If $\deqtype{\Gamma}{\DSum{A}{B}}{\DSum{A'}{B'}}$, then
  $\deqtype{\Gamma}{A}{A'} $, and $\deqtype{\ctxe{\Gamma}{A}}{B}{B'}$.
\end{cor}
}

%% file: tlca09-tc.tex
  \section{Type-checking algorithm}
  \label{sec:proof-alg}
\LONGSHORT{
  \noindent In this section, we define a couple of judgements that represent a
  bidirectional type checking algorithm for terms in normal form; its
  implementation in Haskell can be found in the appendix.
  The algorithm is similar to previous ones 
  \cite{coquand:type,abelCoquandDybjer:flops08},
  in that it proceeds by analysing the possible types for each normal
  form, and succeeds only if the type's shape matches the one required
  by the introduction rule of the term.
  The only difference is introduced by the presence of singleton types;
  now we should take into account that a normal form can also have
  a singleton as its type. 

  This situation can be dealt in two possible ways; either one checks
  that the deepest tag of the normalised type (see Def.~\ref{def:erase})
  has the form of the type of the introductory rule; or one adds a rule
  for checking any term against singleton types. The first approach
  requires to have more rules (this is due to the combination of
  singletons and a universe). We take the second approach, which
  requires to compute the eta-long normal form of the type before
  type-checking. We also note that the proof of completeness is more
  involved, because now the algorithm is not only driven by the term
  being checked, but also by the type.

  Our algorithm depends on having a \emph{good} normalisation function;
  note that this function does not need to be based on normalisation
  by evaluation. Also note that the second point asks for having correctness
  and completeness of the normalisation function.

  \begin{defi}[Good normalisation function]\hfill
    \label{def:nbe-prop}
    \begin{enumerate}[(1)]
    \item $\norm{\singTm{a}{A}} = \singTm{\norm{a}}{\norm{A}}$, and
      $\norm{\F{A}{B}} = \F{\norm{A}{\norm{B}}}$; 
    \item $\normty{\Gamma}{A} = \normty{\Gamma}{B}$ if and only if 
      $\deqtype{\Gamma}{A}{B}$, and 
      $\normtm{\Gamma}{A}{t}= \normtm{\Gamma}{A}{t'}$, if and only if
      $\deqterm{\Gamma}{A}{t}{t'}$.
    \end{enumerate}
  \end{defi}\medskip

  \noindent From these properties we can prove the injectivity of $\F{}{}$which
  is crucial for completeness of type checking $\lambda$-abstractions.
}{ 

  In this section we define a bi-directional type-checking algorithm
  for terms in normal form, and a type-inference algorithm for neutral
  terms. We prove its correctness and completeness.

  The algorithm is similar to previous ones
  \cite{coquand:type,abelCoquandDybjer:flops08}.  The only difference is
  due to the presence of singleton types. We deal with this by
  $\eta$-normalising the type, and considering first if the normalised
  type is a singleton (side-condition in type-checking of neutrals);
  in that case we check that the term is typeable with the tag of the
  singleton type, and that it is equal to the term of the singleton.

  We stress the importance of having a normalisation function with the
  property stated in Rem.~\ref{rem:nbe-hom}, and also to have
  decidability of equality. In fact, it is enough to have a function
  $\nbe{\_}$ such that:
  
    \begin{enumerate}
    \item $\norm{\singTm{a}{A}} = \singTm{\norm{a}}{\norm{A}}$, and
      $\norm{\F{A}{B}} = \F{\norm{A}{\norm{B}}}$; 
    \item $\normty{\Gamma}{A} = \normty{\Gamma}{B}$ if and only if 
      $\deqtype{\Gamma}{A}{B}$, and 
      $\normtm{\Gamma}{A}{t}= \normtm{\Gamma}{A}{t'}$, if and only if
      $\deqterm{\Gamma}{A}{t}{t'}$.
    \end{enumerate}  
}


  

\subsection{Type-checking \lambdaSing}

In this section, let $V,V',W,v,v',w \in \nfterms$, and $k \in
\neterms$.
\LONGSHORT{
For obtaining the deepest tag of a singleton type, 
we define an operation on types, which
is essentially the same as the one defined by Aspinall~\cite{aspinall:csl94}.
  \begin{defi}[Singleton's tag] 
    \begin{equation*}
    \label{def:erase}
      \ertype{V} = 
      \begin{cases}
        \ertype{W} & \mbox{ if } V\equiv\singTm{w}{W}  \\
        V & \mbox{ otherwise.}
      \end{cases}
    \end{equation*}
  \end{defi}
}{ 
  We define a function to get the deepest tag of a singleton, that is
  essentially the same as in \cite{aspinall:csl94},
    \begin{equation*}
    \label{def:erase}
      \ertype{V} = 
      \begin{cases}
        \ertype{W} & \mbox{ if } V\equiv\singTm{w}{W}  \\
        V & \mbox{ otherwise} .
      \end{cases}
    \end{equation*}
}\medskip

  \noindent The predicates for type-checking are defined mutually inductively,
  together with the function for inferring types.

  \begin{defi}[Type-checking and type-inference]
    \label{alg:typecheking}
We define three mutually inductive algorithmic judgements
\[
\begin{array}{ll}
  \chktype \Gamma V   & \mbox{in context $\Gamma$, normal type $V$ checks} \\
  \chkterm \Gamma V v & \mbox{in context $\Gamma$, normal term $v$ checks
    against type $V$} \\
  \inftype \Gamma k V & \mbox{in context $\Gamma$, the type of neutral
    term $k$ is
    inferred as $V$}.
\end{array}
\]
All three judgements presuppose and maintain the invariant the
input $\Gamma$ is a well-formed context. 
The procedures $\chktype \Gamma V$ and $\chkterm \Gamma V v$ expect
their inputs $V$ and $v$ in $\beta$-normal form.  Inference $\inftype
\Gamma k V$ expects a neutral term $k$ and returns its principal type
$V$ in long normal form. 
\vspace{0.7em}

\localpara{Well-formedness checking of types $\chktype \Gamma V$}
      \begin{gather*}
        \algrule{\ }{\chktype{\Gamma}{\TmU}}
\quad
        \algrule{\chktype{\Gamma}{V}\\ \chktype{\ctxe{\Gamma}{V}}{W}}{
          \chktype{\Gamma}{\F{V}{W}}}
\quad
        \algrule{\chktype{\Gamma}{V}\\ \chkterm{\Gamma}{\nbe{V}}{v}}{\chktype{\Gamma}{\singTm{v}{V}}}
\quad
        \algrule{\chkterm{\Gamma}{\TmU}{k}}{\chktype{\Gamma}{k}}
      \end{gather*}    
\localpara{Type checking terms $\chkterm\Gamma V v$} 
      \begin{gather*}
        \algrule{\chkterm{\Gamma}{\TmU}{V} \\
          \chkterm{\ctxe{\Gamma}{V}}{\TmU}{W}}{\chkterm{\Gamma}{\TmU}{\F{V}{W}}}
\qquad 
        \algrule{\chkterm{\ctxe{\Gamma}{V}}{W}{v}
               }{\chkterm{\Gamma}{\F{V}{W}}{\lambda v}}
\\
        \algrule{\chkterm{\Gamma}{\TmU}{V} \\
          \chkterm{\Gamma}{\nbe{V}}{v}}{\chkterm{\Gamma}{\TmU}{\singTm{v}{V}}}
\qquad
        \algrule{\chkterm{\Gamma}{V'}{v} \\ \deqterm{\Gamma}{V'}{v'}{v}}
        {\chkterm{\Gamma}{\singTm{v'}{V'}}{v}}
\\
        \inferrule*[right=\ensuremath{V\not\equiv\singTm{w}{W}}]
          {\inftype{\Gamma}{k}{V'}\\ \deqtype{\Gamma}{\ertype{V'}}{V}}
          {\chkterm{\Gamma}{V}{k}}
      \end{gather*}
\localpara{Type inference $\inftype\Gamma k V$} 
      \begin{gather*}
        \algrule{\ }{\inftype{\Gamma.A_i.\ldots A_0}{\ind i}{\nbe{\lift{i+1}{A_i}}}}\quad
        \algrule{\inftype{\Gamma}{k}{V}\\ \deqtype{\Gamma}{\ertype{V}}{\F{V'}{W}}\\ \chkterm{\Gamma}{V'}{v}}
        {\inftype{\Gamma}{\appTm{k}{v}}{\nbe{\subsTy{W}{\subid{}{v}}}}}
      \end{gather*}
    \end{defi}\medskip
  
\noindent Bidirectional type checking for dependent function types is
well-understood 
\cite{coquand:type,loehMcBrideSwierstra:tutorialDependentlyLambda};
let us illustrate briefly how it works for singleton types, by
considering the type checking problem
$\chkterm{\singTm{\ztm}{\natty}}{\singTm{\ztm}{\natty}}{\q}$.  Here is
a skeletal derivation of this judgement, which is at the same time an
execution trace of the type checker:
\[
\begin{prooftree}
  \[
  \[ \justifies
  \inftype{\singTm{\ztm}{\natty}}{\q}{\singTm{\ztm}{\natty}}
  \]
  \qquad
  \deqtype{\singTm{\ztm}{\natty}}{\ertype{\singTm{\ztm}{\natty}}}{\natty} 
  \justifies
  \chkterm{\singTm{\ztm}{\natty}}{\natty}{\q}
  \] 
  \qquad
  \deqterm{\singTm{\ztm}{\natty}}{\natty}{\q}{\ztm}
  \justifies
  \chkterm{\singTm{\ztm}{\natty}}{\singTm{\ztm}{\natty}}{\q}
\end{prooftree}
\]
    Since the type to check against is a singleton, 
    the algorithm proceeds by checking
    $\chkterm{\singTm{\ztm}{\natty}}{\natty}{\q}$ and
    $\deqterm{\singTm{\ztm}{\natty}}{\natty}{\q}{\ztm}$. Now the type
    of the neutral $\q$ is inferred and its tag compared to the given
    type $\natty$; as the tag is also $\natty$, the check succeeds. 
    The remaining equation
    $\deqterm{\singTm{\ztm}{\natty}}{\natty}{\q}{\ztm}$ is derivable
    by \ruleref{sing-eq-el}. Of course, the equations are checked
    by the $\norm{\_}$ function; for example, by using
    our own function for normalisation we have
    $\normtm{\singTm{\ztm}{\natty}}{\natty}{\q}
     = \ztm = \normtm{\singTm{\ztm}{\natty}}{\natty}{\ztm}$.

  \begin{thm}[Correctness of type-checking]\hfill
    \label{thm:corr-tc}
    \begin{enumerate}[\em(1)]
    \item If 
          $\chktype{\Gamma}{V}$, then $\dtype{\Gamma}{V}$.
    \item If 
          $\chkterm{\Gamma}{V}{v}$, then $\dterm{\Gamma}{V}{v}$.
    \item If 
          $\inftype{\Gamma}{k}{V}$, then $\dterm{\Gamma}{V}{k}$.
    \end{enumerate}
  \end{thm}
\LONGSHORT{
  \begin{proof}By simultaneous induction on 
    $\chktype{\Gamma}{V}$, $\chkterm{\Gamma}{V}{v}$, and $\inftype
    \Gamma V k$. See \ref{prf:corr-tc}.
  \end{proof}
}
{
  \begin{proof}
    By simultaneous induction on the type-checking judgement.
  \end{proof}
}

In order to prove completeness we define a lexicographic order on
pairs of terms and types, in this way we can make induction over the
term, and the type.
  
\begin{defi}
  \label{def:wf-compl}
  Let $v, v' \in \nfterms$, and $A, A' \in \type{\Gamma}$, then $(v,A)
  \prec (v',A')$ is the lexicographic order on $\nfterms \times
  \type{\Gamma}$. The corresponding orders are $v \prec v'$ iff $v$ is
  an immediate sub-term of $v'$; and $A \prec^{\Gamma} A'$, iff
  $\nbe{A'} \equiv \singTm{w}{\nbe{A}}$.
\end{defi}
  
  \begin{thm}[Completeness of type-checking]\hfill
    \label{thm:compl-tc}
    \begin{enumerate}[\em(1)]
    \item If $\dtype{\Gamma}{V}$, then $\chktype{\Gamma}{V}$.
    \item If $\dterm{\Gamma}{A}{v}$, then
      $\chkterm{\Gamma}{\nbe{A}}{v}$.
    \item If $\dterm{\Gamma}{A}{k}$, and $\inftype{\Gamma}{k}{V'}$,
      then $\deqtype{\Gamma}{\ertype{\nbe{A}}}{\ertype{V'}}$.
    \end{enumerate}
  \end{thm}
\LONGSHORT{
  \begin{proof}
    We prove these three statements simultaneously by well-founded
    induction on the order $\prec$.  The respective measures are (1)
    $(V,\TmU)$, (2) $(v,A)$, and (3) $(k,A)$.
    Details are in the Appendix~\ref{prf:compl-tc}.
  \end{proof}
}
{
  \begin{proof}
    By simultaneous induction on $V$, and well-founded induction on
    $(v,A)$.
  \end{proof}
}


\PrfIrrTitle

\LONGSHORT{

  We give additional rules for type-checking and type-inference
  algorithms for the constructs added in
  Sect.~\ref{sec:pi-calc}. Remember that we distinguished two calculi:
  the calculus $(\vdashp)$ has rules (\textsc{\enumrule[0]tm}) and
  (\textsc{prf-tm}); while $(\vdash)$ lacks those rules.



}{
}

\begin{defi}[Type-checking and type-inference]
  \label{alg:nat-tc}

\LONGSHORT{
\renewcommand{\para}[1]{\noindent #1.} 

\para{$\Sigma$-types}

\begin{gather*}
    \algrule{\chktype{\Gamma}{V}\\
        \chktype{\ctxe{\Gamma}{V}}{W}}{\chktype{\Gamma}{\DSum{V}{W}}}
\qquad
    \algrule{\chkterm{\Gamma}{\TmU}{V}\\
      \chkterm{\ctxe{\Gamma}{V}}{\TmU}{W}}{\chkterm{\Gamma}{\TmU}{\DSum{V}{W}}}
\\
    \algrule{\chkterm{\Gamma}{V}{v}\\
      \chkterm{\Gamma}{\nbe{\subsTy{W}{\subid{\Gamma}{v}}}}{v'} \\ 
    }
    {\chkterm{\Gamma}{\DSum{V}{W}}{\depair{v}{v'}}}
\\
  \algrule{\inftype{\Gamma}{k}{\DSum{V}{W}}}
  {\inftype{\Gamma}{\dfst{k}}{V}} 
\qquad
  \algrule{\inftype{\Gamma}{k}{\DSum{V}{W}}}
  {\inftype{\Gamma}{\dsnd{k}}
    {\nbe{\subsTy{W}{\subid{\Gamma}{\dfst{k}}}}}}
\end{gather*}

\para{Natural numbers}
\begin{gather*}
    \algrule{\ }{\chktype{\Gamma}{\natty}}
\qquad
    \algrule{\ }{\chkterm{\Gamma}{\TmU}{\natty}}
\qquad
    \algrule{\ }{\chkterm{\Gamma}{\natty}{\ztm}}
\qquad
    \algrule{\chkterm{\Gamma}{\natty}{v}}{\chkterm{\Gamma}{\natty}{\suctm{v}}}
\\
   \algrule{\chktype{\ctxe{\Gamma}{\natty}}{V}\\
         \inftype{\Gamma}{k}{\natty}\\
         \chkterm{\Gamma}{\nbe{\subsTy{V}{\subid{\Gamma}{\ztm}}}}{v} \\ 
         \chkterm{\Gamma}{\F \natty {(\F V 
            \nbe{\subsTy{V}{\exsubs{\subsc \p \p}{\suctm{(\subsTm{\q}{\p})}}}})}}{v'}}
       {\inftype{\Gamma}{\natrec{V}{v}{v'}{k}}{\nbe{\subsTy{V}{\subid{}{k}}}}}
\end{gather*}

\para{Finite types}
  \begin{gather*}
    \algrule{\ }{\chktype{\Gamma}{\enum{n}}}
\qquad
    \algrule{\ }{\chkterm{\Gamma}{\TmU}{\enum{n}}}
\qquad
    \algrule{i<n}{\chkterm{\Gamma}{\enum{n}}{\const{n}{i}}}\
\\
    \algrule{\chktype{\ctxe{\Gamma}{\enum{n}}}{V}\\
      \inftype{\Gamma}{k}{\enum{n}}\\
      \chkterm{\Gamma}{\nbe{\subsTy{V}{\subid{\Gamma}{\const{n}{i}}}}}{v_i}}
    {\inftype{\Gamma}{\elim{n}{V}{v_0 \cdots v_{n-1}}{k}}{\nbe{\subsTy{V}{\subid{}{k}}}}}
  \end{gather*}

\para{Proof types}
  \begin{gather*}  
    \algrule{\chktype{\Gamma}{V}}{\chktype{\Gamma}{\boxty{V}}}
\qquad
    \algrule{\chkterm{\Gamma}{\TmU}{V}}{\chkterm{\Gamma}{\TmU}{\boxty{V}}}
\qquad
    \algrule{\chkterm{\Gamma}{V}{v}}{\chkterm{\Gamma}{\boxty{V}}{\boxtm{v}}}
\\
    \algrule{
        \chktype \Gamma W \\
        \inftype{\Gamma}{k}{\boxty{V}}\\
        \chkterm{\ctxe{\Gamma}{V}}{\nbe{\subsTm{W}{\p}}}{v} \\
        \deqterm{\ctxe{\ctxe \Gamma V} {V \p}}{W \p \p}{v \p}{v (\p \p, \q)}  
      }{\inftype{\Gamma}{\wheretm{v}{k}{W}}{W}}
  \end{gather*}

}{  
  \begin{gather*}
    \algrule{\chktype{\Gamma}{V}}{\chktype{\Gamma}{\boxty{V}}}
    \
    \algrule{\chkterm{\Gamma}{V}{v}}{\chkterm{\Gamma}{\boxty{V}}{\boxtm{v}}}
    \
    \algrule{
      \inftype{\Gamma}{k}{\boxty{V'}}\\
      \chkterm{\ctxe{\Gamma}{V'}}{\nbe{\subsTm{V}{\p}}}{v}}
    {\inftype{\Gamma}{\wheretm{v}{k}{V}}{\boxty{V}}}
  \end{gather*}

}
\end{defi}\medskip


\LONGSHORT{

\noindent  We do not show the proof for correctness, because nothing is to be
  gained from it; suffice it to say that we can prove correctness with
  respect to $(\vdashp)$.

\begin{thm}
\label{thm:corr-tc-pi}
  The type-checking algorithm is sound with respect to the calculus
  $\vdashp$.
\end{thm}  
\begin{proof}
  By simultaneous induction on the derivability of the type-checking
  judgements.
\end{proof}

It is clear that the given rules are not complete for checking
$(\vdashp)$, because there is no rule for checking
$\chkterm{\Gamma}{A}{\oprf}$. Note that it is not possible to have a
sound and complete type-checking algorithm with respect to
$(\vdashp)$, for it would imply the decidability of
type-inhabitation. Since type checking happens always \emph{before}
normalisation, we can still use a good normalisation function with
respect to the calculus $(\vdashp)$ for normalising types or deciding
equality. Indeed, if the term to type-check does not contain $\oprf$,
the need of checking $\chkterm{\Gamma}{V}{\oprf}$ will never arise;
this is clearly seen by verifying that only sub-terms are
type-checked in the premises.

\begin{thm}
  The type-checking algorithm is complete with respect to the calculus
  $(\vdash)$.
\end{thm}  
\begin{proof}
  By simultaneous induction on the normal form of types and terms,
  using inversion on the typing judgement and correctness of $\norm{\_}$.
\end{proof}

}{ 
  \begin{rem}
    Thm.~\ref{thm:corr-tc} is still valid for the calculus with
    \ruleref{prf-tm}. Moreover, Thm.~\ref{thm:compl-tc} is valid if
    we add the axiom $\chkterm{\Gamma}{\boxty{V}}{\oprf}.$
  \end{rem}
}


\begin{cor}
  The type-checking algorithm is correct (by Thm.~\ref{thm:corr-tc-pi}
  and Cor.~\ref{cor:cons-prf-tm}) and complete with respect to the
  calculus $(\vdash)$. 
\end{cor}

%


%% file: related.tex

\subsection{Related and Further Work on Singleton Types}
\label{sec:related}

\noindent Singleton types are used to model the SML module system and records
with manifest fields \cite{coquandPollackTakeyama:fundinf05}.

Aspinall \cite{aspinall:csl94} presents a logical framework with
singleton types and subtyping and shows its consistency via a PER
model, yet not decidability.  
The second author, Pollack, and Takeyama
\cite{coquandPollackTakeyama:fundinf05} extend the Aspinall's
framework by $\eta$-equality and records and a type checking algorithm
which is correct wrt.\ the PER model. This work is unconventional since
there is no complete syntactical specification of the LF in terms of
syntax, typing and equality rules.  Instead, in the style of
Martin-L\"of meaning explanations, they list a number of inference
rules which are valid in the semantics and prove that type-checked
expressions evaluate to values of the correct semantic type.

Courant \cite{courant:itrs02} shows strong normalization for a variant of
Aspinall's system with equality defined by reduction.  He uses a
typed Kripke model of strongly normalizing terms, a variant of
Goguen's typed operational semantics \cite{goguen:PhD}.  

Stone and Harper \cite{stoneHarper:tocl06} extend Aspinall's framework
by sigma types and eta-equality, which allows them to reduce
singletons at higher types to singletons at base type.  Their decision
procedure is type-directed, its completeness is shown via a Kripke
model.  Crary \cite{crary:lfmtp08} gives a simplified decision
procedure via hereditary substitutions and proves its correctness in
Twelf, without the need for a model construction.  His purely
syntactical approach does not scale to universes, since he cannot
handle types defined by recursion.  Goguen \cite{goguen:syntacticEta}
follows a similar agenda, he shows decidability for singleton types
in the presence of eta by an eta-expanding translation into a logical
framework with beta-equality only.  He works with fully annotated
terms in the sense of Streicher \cite{streicher:PhD}.  He stresses
that his approach does not scale to computation on the type level.

In the continuation of this work we want to investigate whether our
type-checking algorithm can be simplified if we implement Stone and
Harper's insight that singleton types at higher types can be defined
in terms of singleton base types.  Further, we would like to integrate
subtyping in our calculus, which should not be too difficult, since
the PER model already supports subtyping
\cite{aspinall:csl94,coquandPollackTakeyama:fundinf05}.

\subsection{Related and Further Work on Proof Irrelevance}

Pfenning \cite{pfenning:intextirr} presents a logical framework with
proof irrelevance that supports irrelevant function arguments,
with function introduction rule 
(writing $\funT x {\boxty A} B$ in our syntax):  
\[
  \dfrac{\Gamma, x \div A \derN B \qquad \Gamma, x \div A \derN t : B
    }{\Gamma \derN \lambda x t :  \funT x {\boxty A} B}
\]
He proves decidability using erasure, mentioning that
his technique does not scale to universes.  
Elimination of irrelevance is implicitly handled by annotating
variables to ensure proof variables $(x \div A)$
appear only in proofs,
in contrast to our explicit use of 
$\ABwhere \_ \_ \_$ in the style of Awodey and Bauer
\cite{awodeyBauer:propositionsAsTypes}.  However, we believe that
Pfenning's proof irrelevance can be modeled via bracket types 
$\boxty A$, with the weaker ``monadic'' rule for $\whereraw$ (see
section~\ref{sec:cwf-gat}). 

Barras and Bernardo's \cite{DBLP:conf/fossacs/BarrasB08} presentation of
proof irrelevant functions
\[
  \dfrac{\Gamma, x : A \derN B \qquad 
      \Gamma, x : A \derN t : B \qquad x \not\in\FV(t^*)
    }{\Gamma \derN \lambda x t :  \funT x {\boxty A} B}
\]
diverges from Pfenning's that they allow irrelevant variables $x$ to
be relevant in types $B$.  (In $t$ the variable $x$ might only
appear irrelevantly, expressed by the side condition that $x$ may not be
free in the relevant parts $t^*$ of $t$.)
Barras and Bernardo justify their calculus by erasing into Miquel's
Implicit Calculus of Constructions (ICC) \cite{miquel:tlca01}. The ICC
style irrelevance seems more expressive than Awodey and Bauer's or
Pfenning's, but the exact relationship is unclear to us.


Berger's Uniform Heyting Algebra \cite{berger:uniformHA} features
uniform quantification $\{\forall x\}A$  (and $\{\exists x\}A)$ to
obtain optimized programs by extraction from proofs.  A proof of a
uniform universal
\[
  \dfrac{\Gamma \derN M : A
    }{\Gamma \derN \{\forall\}^+(\lambda x M) : \{\forall x\}A}
\]
may not mention term variable $x$ in a computational relevant
position.  Since the shape of formulas does not depend on terms,
Berger's calculus can be seen as logical counterpart of
either Pfenning's or Bruno and Bernardo's type system.

We see two interesting questions about the different approaches to
proof irrelevance above:
\begin{enumerate}[(1)]
\item How can Barras and Bernardo's ICC$^*$ be understood in
  terms of judgmental equality \`a la Pfenning?
\item How can ICC$^*$ and the calculus of Pfenning be extended to full
  bracket types \`a la Awodey and Bauer without explicit use of $\whereraw$.
\end{enumerate}


%% file: lmcs-proofs.tex
\vfill\eject

\section{Proofs}
\label{sec:proofs}

\newcounter{prfcnt}
\renewcommand{\theprfcnt}{\Alph{section}.\arabic{prfcnt}}

\newcommand{\DisplayProof}[4]{
\refstepcounter{prfcnt}\label{prf:#1}
\proof[\theprfcnt .\ Proof of #2 \ref{#3}.]#4
}

\DisplayProof{famperd}{Lemma}{lem:famperd}{%
    By induction on $X = X' \in \perT$. We do not show the base cases,
  for they are trivial.
  \begin{enumerate}[(1)]
  \item Let $\iSing{d}{X} = \iSing{d'}{X'} \in \perT$.
    \begin{align*}
      \proofLine{[X] = [X']}{by ind. hyp.}\\
      \proofLine{d = d' \in [X]}{by ind. hyp.}\\
      \proofLine{e = d \in [X]\text{ and }e' =d \in [X]}{hypothesis}\\
      \proofLine{e = d' \in [X]\text{ and }e' = d' \in [X]}{by transitivity}\\
      \proofLine{\singD{d}{X} = \singD{d'}{X'}}{by definition}.
    \end{align*}
  \item Let $\iPi{X}{F} = \iPi{X'}{F'} \in \perT$.
    \begin{align*}
      \proofLine{[X] = [X']}{by ind. hyp.}\\
      \proofLine{\text{ for all } d \in dom([X]), F\ d= F'\ d \in \perT \tag{*}\label{eq:per-fam-fun} }{by definition}\\
      \proofLine{\text{ for all } d=d' \in [X], f\cdot d = f'\cdot d' \in [F\ d] }{hypothesis}\\
      \proofLine{f\cdot d = f'\cdot d' \in [F'\ d]}{by ind. hypothesis
        in \eqref{eq:per-fam-fun}}.\rlap{\hbox to 38.3pt{\hfill\qEd}}
    \end{align*}
  \end{enumerate}
}

\DisplayProof{reify}{Lemma}{lem:reify}{By induction on $X=X'
  \in \perT$.
  \begin{enumerate}[(a)]
  \item Case $\iSing{d}{X} = \iSing{d'}{X'}\in \perT$.
    \begin{enumerate}[(1)]
    \item The partial function $\upa{}{}$ maps neutrals to related elements
      in the corresponding PER.
      \begin{align*}
        \proofLine{k = k\in \perne}{hypothesis}\\
        \proofLine{d = d' \in [X]\text{ and } X = X' \in \perT}{by inversion}\\
        \proofLine{\upa{\iSing{d}{X}}{k}= d \text{ and } 
                   \upa{\iSing{d'}{X'}}{k'} = d' }{by def.}\\
        \proofLine{d = d' \in \singD{d}{X}}{by def. of this PER}.
      \end{align*}
    \item The partial function $\da{}{}$ maps related elements to related
      normal forms.
      \begin{align*}
        \proofLine{d_1 = d_2 \in \singD{d}{X}}{hypothesis}\\
        \proofLine{d_1 = d_2 = d = d' \in [X] \text{ and } 
                   X = X' \in \perT}{by inversion}\\
        \proofLine{\da{X}{d}=\da{X'}{d'}\in\pernf }{by ind. hyp.} \\
        \proofLine{\da{\iSing d X}{d_1}=\da{\iSing{d'}{X'}}{d_2}\in\pernf }{by def.}
        \end{align*}
    \item The function $\Da$ maps related elements in $\perT$ to normal forms.
      \begin{align*}
        \proofLine{\Da{\iSing{d}{X}} = \iSing{(\da{X}{d})}{(\Da{X})}}{by def.}\\
        \proofLine{\Da{\iSing{d'}{X'}} = \iSing{(\da{X'}{d'})}{(\Da{X'})}}{by def.}\\
        \proofLine{\da{X}{d}=\da{X'}{d'}\in\pernf }{by ind. hyp.}\\
        \proofLine{\Da{X}=\Da{X'}\in\pernf }{by ind. hyp.}\\
        \proofLine{\iSing{(\da{X}{d})}{(\Da{X})}=\iSing{(\da{X}{d})}{(\Da{X})}\in \pernf }{by Lem.~\ref{rem:presnf}}.
      \end{align*}
    \end{enumerate}
  \item Case $\iPi{X}{F} = \iPi{X'}{F'}\in \perT$.
    \begin{enumerate}[(1)]
    \item The partial function $\upa{}{}$ maps neutrals to related elements
      in the corresponding PER.
      \begin{align*}
        \proofLine{k = k' \in \perne}{hypothesis} \\
        \proofLine{d=d'\in [X]}{hypothesis}\tag{*}\label{eq:eta-fun-1}\\
        \proofLine{X=X'\in\perT}{by inversion}\tag{\textdagger}\label{eq:eta-fun-2}\\
        \proofLine{F\ d=F'\ d'\in\perT}{by inversion}\tag{**}\label{eq:eta-fun-3}\\
        \proofLine{\da{X}{d} = \da{X'}{d'} \in \pernf}%
        {by ind. hyp. on \eqref{eq:eta-fun-1} and \eqref{eq:eta-fun-2}}\\
          \proofLine{\iNe{k}{(\da{X}{d})} =\iNe{k'}{(\da{X'}{d'})}\in \perne}%
        {by Lem.~\ref{rem:presnf}}\tag{\textdaggerdbl}\label{eq:eta-fun-4}\\
          \proofLine{\upa{F\,d}{(\iNe{k}{(\da{X}{d})})} = 
                     \upa{F'\,d'}{(\iNe{k'}{(\da{X'}{d'})})} \in [F\ d]}{by ind. hyp. on \eqref{eq:eta-fun-3} and \eqref{eq:eta-fun-4}} \\
        \proofLine{\upa{\iPi X F}{k} = \upa{\iPi{X'}{F'}}{k'} 
          \in [\iPi X F]}{by def.}
        \end{align*}
      \item The partial function $\da{}{}$ maps related elements to related
      normal forms.
        \begin{align*}
        \proofLine{X=X'\in\perT}{by inversion}\tag{*}\label{eq:mse-fun-1}\\
        \proofLine{f = f' \in [\iPi{X}{F}]}{hypothesis}\\
        \proofLine{k = k' \in \perne}{hypothesis}\\
        \proofLine{\upa{X}{k} = \upa{X'}{k'}\in [X]}%
          {by ind. hyp. on \eqref{eq:mse-fun-1}}\tag{\textdagger}\label{eq:mse-fun-2}\\
        \proofLine{d := \upa{X}{k}}{abbreviation}\\
        \proofLine{d' := \upa{X'}{k'}}{abbreviation}\\
        \proofLine{F\ d=F'\ d'\in\perT}%
          {by inversion and \eqref{eq:mse-fun-2}}\tag{**}\label{eq:mse-fun-3}\\
        \proofLine{f\cdot d = f'\cdot d' \in [F\ d]}%
          {definition of $[\iPi{X}{F}]$}\tag{\textdaggerdbl}\label{eq:mse-fun-4}\\
        \proofLine{\da{F\ d}{(f\cdot d)} = \da{F'\ d'}{(f'\cdot d')} \in \pernf}%
          {by ind. hyp. on \eqref{eq:mse-fun-4}} \\
        \proofLine{(\da{\iPi X F}{f}) \cdot k = (\da{\iPi{X'}{F'}}{f'})
          \cdot k' \in \pernf}{by def. } \\
        \proofLine{\da{\iPi X F}{f} = \da{\iPi{X'}{F'}}{f'}
          \in \pernf}{by Lem.~\ref{rem:presnf} } \\
      \end{align*}      
    \item The function $\Da$ maps related elements in $\perT$ to normal forms.
      \begin{align*}
        \proofLine{X=X'\in\perT}{by inversion}\tag{*}\label{eq:eta-t-fun-1}\\
        \proofLine{\Da{X}=\Da{X'}\in\pernf}{by ind. hyp. on \eqref{eq:eta-t-fun-1}}%
          \tag{**}\label{eq:eta-t-fun-2}\\
        \proofLine{k = k' \in \perne}{hypothesis.}\\
        \proofLine{\upa{X}{k} = \upa{X'}{k'}\in [X]}%
          {by ind. hyp. on \eqref{eq:eta-t-fun-1}}\tag{\textdagger}\label{eq:eta-t-fun-3}\\
        \proofLine{d := \upa{X}{k}}{abbr.}\\
        \proofLine{d' := \upa{X'}{k'}}{abbr.}\\
        \proofLine{F\,d=F'\,d'\in\perT}%
          {by inversion and \eqref{eq:eta-t-fun-3}}\tag{\textdaggerdbl}\label{eq:eta-t-fun-4}\\
        \proofLine{\Da{(F\,d)}=\Da{(F'\,d')}\in\pernf}{by ind. hyp. on \eqref{eq:eta-t-fun-4}}%
          \label{eq:eta-t-fun-5}\\
        \proofLine{\Da{(\iPi X F)} = \Da{(\iPi{X'}{F'})} \in \pernf}{by Lem.~\ref{rem:presnf}}\rlap{\hbox to 93pt{\hfill\qEd}}
        \end{align*}
    \end{enumerate}
  \end{enumerate}
}

\DisplayProof{soundness-irr}{Lemma}{lem:soundness-irr}{The proofs of
  soundness for \rulename{prf-$\beta$} and \rulename{prf-$\eta$} have
  the same structure, so we show only the first one.
  \begin{enumerate}[\rulename{prf-assoc}]
  \item[\rulename{prf-$\beta$}] 
    $\wheretm{b}{\boxtm{a}}{B} = \subsTm{b}{\subid{}{a}}$
    \begin{align*}
      \proofLine{\semc{\wheretm{b}{\boxtm{a}}{B}}d }{}\\
       \proofLine{= \semc b{\iPair d \dprf} }{def. of semantics for $\wheretm{b}{\boxtm{a}}{B}$}\\
       \proofLine{= \semc b{\iPair d {\semc{a}d}} }{ind. hypothesis on 
         $\deqterm{\ctxe{\ctxe{\Gamma}{A}}{\subsTy{A}{\p}}}{\subsTy{B}{\subsc \p \p}}
         {\subsTm{b}{\p}}{\subsTm b {\exsubs{\subsc \p \p}{\q}}}$}\\
       \proofLine{= \semc b{(\semc {\subid{} a} d)} }{def. of semantics for substitutions}\\
       \proofLine{= \semc{\subsTm{b}{\subid{}{a}}}d }{}
     \end{align*}
  \item[\rulename{prf-assoc}]
  $\wheretm a {(\wheretm b c B)} A 
    = \wheretm {(\wheretm {\subsTm a {\exsubs {\subsc \p \p} \q}} b
      {\subsTy A \p})} c B$
\[
\begin{array}{lll}
  \semc {\wheretm a {(\wheretm b c B)} A} d
  & = & \semc a {\iPair d \dprf}
\\ &= & \semc a {\iPair d {\semc b {\iPair d \dprf}}}
\\ &= & \semc {\subsTm a {\exsubs {\subsc \p \p} \q}} 
              {\iPair {\iPair d {\dprf}} 
                      {\semc b {\iPair d {\dprf}}} }
\\ &= & \semc {\wheretm {\subsTm a {\exsubs {\subsc \p \p} \q}} b 
                                                          {\subsTy A
                                                            \p})}
              {\iPair d \dprf}
\\ &= & \semc {\wheretm {(\wheretm {\subsTm a {\exsubs {\subsc \p \p} \q}} b
      {\subsTy A \p})} c B} d\rlap{\hbox to 74pt{\hfill\qEd}}
\end{array}
\]
\end{enumerate}
}

\DisplayProof{logrelEqTy}{Lemma}{lem:logrelEqTy}{By induction on $X\in \perT$. 
  \begin{enumerate}[(a)]
  \item Types; in all cases we use symmetry and transitivity to show 
    the conditions. We only show the case for $\iPi{X}{F}$.
    \begin{enumerate}[(1)]
    \item $X = \iPi{X'}{F} $: 
      \begin{align*}
        \proofLine{\deqtype{\Gamma}{A}{\F{B}{C}}}{by definition}\tag{*}\label{eq:cong-j-eq}&\\
        \proofLine{\dtype{\Gamma}{B}\rel X'}{by definition}\\
        \proofLine{\dtype{\Delta}{\subsTy{C}{\exsubs{\p^i}{s}}} \rel F\ d \in
           \perT}{by definition}\\
         &\text{\hspace{3cm} for all }\Delta\leqslant^i\Gamma \text{ and }
         \dterm{\Delta}{\lift{i}{B}}{s} \rel d\in [X']&&\text{}\\
         \proofLine{\deqtype{\Gamma}{A'}{\F{B}{C}}}{by sym. and trans. on \eqref{eq:cong-j-eq}} 
      \end{align*}
    \item $\enumD{n}\in \perT$.
    \begin{align*}
      \proofLine{\dterm{\Gamma}{A}{t}\rel d\in [\enum{1}]}{hypothesis}\tag{*}\label{eq:eq-j-en-1}\\
      \proofLine{\deqterm{\Gamma}{A}{t}{t'}}{hypothesis}\tag{\textdagger}\label{eq:eq-j-en-2}\\
      \proofLine{\deqterm{\Delta}{\lift{i}{A}}{\lift{i}{t}}{\reify{i}{d}}}%
            {by inversion on \eqref{eq:eq-j-en-1}}\tag{**}\label{eq:eq-j-en-3}\\
       \proofLine{\deqterm{\Delta}{\lift{i}{A}}{\lift{i}{t}}{\lift{i}{t'}}}%
            {by congruence on \eqref{eq:eq-j-en-2}}\tag{\textdaggerdbl}\label{eq:eq-j-en-4}\\
       \proofLine{\deqterm{\Delta}{\lift{i}{A}}{\lift{i}{t'}}{\reify{i}{d}}}%
            {by sym. and trans. on \eqref{eq:eq-j-en-3} and \eqref{eq:eq-j-en-4}}\\
    \end{align*}
    \end{enumerate}
  \item Terms. As in the case for types, we use symmetry and
    transitivity. We show only the case for singletons and functions.
    \begin{enumerate}[(1)]
    \item $X = \iSing{d}{X'}$: 
      \begin{align*}
      \proofLine{\deqtype{\Gamma}{A}{\singTm{b}{B}}}{by hypothesis}\tag{*}\label{eq:cong-j-eq-s1}\\
      \proofLine{\dtype{\Gamma}{B}\rel X' \in \perT}{by hypothesis}\\
      \proofLine{\dterm{\Gamma}{B}{t}\rel d\in [X']}{by hypothesis}  \tag{\textdagger}\label{eq:cong-j-eq-s2}\\
      \proofLine{\deqtype{\Gamma}{A'}{\singTm{b}{B}}}{by sym. and trans. on \eqref{eq:cong-j-eq-s1}}\\
      \proofLine{\dterm{\Gamma}{B}{t'}\rel d\in [X']}{By i.h. on \eqref{eq:cong-j-eq-s2}} 
    \end{align*}
  \item $X = \iPi{X'}{F}$: 
      \begin{align*}
        \proofLine{\deqtype{\Gamma}{A}{\F{B}{C}}}{by hypothesis} \tag{*}\label{eq:cong-j-eq-f1}&\\
        \proofLine{\dtype{\Gamma}{B}\rel X'}{by hypothesis}\\
        \proofLine{\dterm{\Delta}{\subsTy{C}{\exsubs{\p^i}{s}}}{\appTm{\lift{i}{t}}{s}} \rel f\cdot d \in [F\ d]}
          {by hypothesis}\\
        & \text{\hspace{2cm} for all }\Delta\leqslant^i\Gamma \text{ and }
        \dterm{\Delta}{\lift{i}{B}}{s} \rel d\in [X']  &&\\
        \proofLine{\deqtype{\Gamma}{A'}{\F{B}{C}} }{by sym. and trans. on \eqref{eq:cong-j-eq-f1}}
        \tag{\textdagger} \label{eq:cong-j-eq-f4}\\
        \proofLine{\deqterm{\Delta}{\subsTy{C}{\exsubs{\p^i}{s}}}{\appTm{\lift{i}{t}}{s}}{\appTm{\lift{i}{t'}}{s}}}%
          {by congruence on \eqref{eq:cong-j-eq-f4}}\tag{\textdaggerdbl} \label{eq:cong-j-eq-f5}\\
        \proofLine{\dterm{\Delta}{\subsTy{C}{\exsubs{\p^i}{s}}}{\appTm{\lift{i}{t'}}{s}} \rel f\cdot d \in [F\ d]}
          {by i.h. on \eqref{eq:cong-j-eq-f5}}\rlap{\hbox to 86pt{\hfill\qEd}}
      \end{align*}
    \end{enumerate}
  \end{enumerate}
}

\DisplayProof{logrelMon}{Lemma}{lem:logrelMon}{By induction on $X \in \perT$.  This property is
    trivial for the base cases; for singletons is obtained by applying
    the i.h. We show two cases.
    \begin{enumerate}
    \item Let $X = \iPi{X'}{F} $.
      \begin{align*}
        \proofLine{\deqtype{\Gamma}{A}{\F{B}{C}}}{by hypothesis} \tag{*}\label{eq:cong-m-f1}&\\
        \proofLine{\dtype{\Gamma}{B}\rel X'}{}\tag{\textdagger}\label{eq:cong-m-f2}&\\
        \proofLine{\dtype{\Theta}{\subsTy{C}{\exsubs{\p^i}{s}}} \rel
          F\ d \in
          \perT}{by hypothesis}\\
        & \text{\hspace{2cm} for all }\Theta\leqslant^i\Gamma \text{
          and }
        \dterm{\Theta}{\lift{i}{B}}{s} \rel d\in [X'] &&\\
        \proofLine{\deqtype{\Delta}{\lift{i}{A}}
          {\F{(\lift{i}{B})}{(\subsTy{C}{\exsubs{\subsc{\p^i}{\p}}{\q}})}}}
        {by congruence on \eqref{eq:cong-m-f1}}\\
        \proofLine{\dtype{\Delta}{\lift{i}{B}}\rel X}{by i.h. on \eqref{eq:cong-m-f2}}\\
        \proofLine{\dterm{\Theta'}{\lift{j}{(\lift{i}{B})}}{s}\rel d\in [X],\text{with } \Theta'\leqslant^j\Delta}{hypothesis}\\
        \proofLine{\dterm{\Theta'}{\lift{i+j}{B}}{s}\rel d\in [X]}{by
          rem. \ref{rem:lift} and \ref{lem:logrelEqTy}}
        \tag{\textdaggerdbl}\label{eq:cong-m-f4}\\
        \proofLine{\dtype{\Theta'}{\subsTy{C}{\exsubs{\subsc{\p^{i+j}}{\q}}{s}}}\rel F\ d }{by hyp. using \eqref{eq:cong-m-f4}}\\
        \proofLine{\dtype{\Theta'}{\subsTy{\subsTy{C}{\exsubs{\subsc{\p^i}{\p}}{\q}}}{\exsubs{\p^j}{s}}}
          \rel F\ d}{By congruence and \ref{lem:logrelEqTy}}
      \end{align*}
    \item $\prf{X} \in \perT$.  As mentioned earlier if
      $\dtype{\Gamma}{A}\rel \prf{X}\in \perT$ then
      $\dterm{\Gamma}{A}{\_}\rel \_ \in [\prf{X}]$ is non-empty if and
      only if $\dterm{\Gamma}{A}{\_}$ is not empty.
      \begin{align*}
        \proofLine{\dterm{\Gamma}{A}{t} \rel d\in [\prf{X}]}{hypothesis}\tag{*}\label{eq:mon-prf-1} \\
        \proofLine{\dterm{\Gamma}{A}{t}}{by inversion on
          \eqref{eq:mon-prf-1}}
        \tag{\textdagger}\label{eq:mon-prf-2}\\
        \proofLine{\dtype{\Gamma}{A}\rel \prf{X}\in \perT}{by
          inversion on \eqref{eq:mon-prf-1}}
        \tag{**}\label{eq:mon-prf-3}\\
        \proofLine{\dterm{\Delta}{\lift{i}{A}}{\lift{i}{t}}}{by weakening on \eqref{eq:mon-prf-2}}\\
        \proofLine{\dtype{\Delta}{\lift{i}{A}}\rel \prf{X}\in \perT}{by monotonicity for types on \eqref{eq:mon-prf-3}}\\
        \proofLine{\dterm{\Delta}{\lift{i}{A}}{\lift{i}{t}}\rel d\in
          [\prf{X}]}{by definition of log. rel.}
      \end{align*}
    \end{enumerate}
  We do not show proofs for the second part, since the most involved
  case is dealt analogously to the case for $\F{X'}{F}$.\qed  
}

\DisplayProof{logrelEqD}{Lemma}{lem:logrelEqD}{By induction on $X = X'
  \in \perT$. Note that the first part for the base cases is trivial;
  the second point is also trivial for $X \in \perne$. Thus we do not
  show those parts of the proof.
  \begin{enumerate}[(a)]
  \item Types.
    \begin{enumerate}[(1)]
    \item $\iSing{d}{X} = \iSing{d'}{X'}$. 
      \begin{align*}
        \proofLine{\deqtype{\Gamma}{A}{\singTm{b}{B}}}{by hypothesis}\tag{*}\label{eq:cong-d-eq-s1}\\
        \proofLine{\dtype{\Gamma}{B}\rel X \in \perT}{by hypothesis}\\
        \proofLine{\dterm{\Gamma}{B}{t}\rel d\in [X]}{by hypothesis}\tag{\textdagger}\label{eq:cong-d-eq-s2}\\
        \proofLine{\dterm{\Gamma}{B}{t}\rel d'\in [X']}{By i.h. on \eqref{eq:cong-d-eq-s1} and \eqref{eq:cong-d-eq-s2}}
      \end{align*}
    \item $\iPi{X}{F} = \iPi{X'}{F'}$. 
      \begin{align*}
        \proofLine{\deqtype{\Gamma}{A}{\F{B}{C}}}{by hypothesis}\\
        \proofLine{\dtype{\Gamma}{B}\rel X'\tag{*}\label{eq:cong-d-eq-f1}}{by hypothesis}\\
        \proofLine{\dtype{\Theta}{\subsTy{C}{\exsubs{\p^i}{s}}} \rel F\ d \in
          \perT}{by hypothesis}\tag{\textdagger}\label{eq:cong-d-eq-f2}\\ 
        & \text{\hspace{3cm} for
          all }\Theta\leqslant^i\Gamma \text{ and }
        \dterm{\Theta}{\lift{i}{B}}{s} \rel d\in [X'] &&\\
        \proofLine{\dtype{\Gamma}{B}\rel X' \in \perT}{By i.h. on \eqref{eq:cong-d-eq-f1}} \\
        \proofLine{\dtype{\Theta}{\subsTy{B}{\exsubs{\p^i}{s}}}\rel F'\ d \in
          \perT}{by i.h. on \eqref{eq:cong-d-eq-f2} }
      \end{align*}
    \end{enumerate}
  \item Terms.
    \begin{enumerate}[(1)]
    \item $e = e' \in [\iSing{d}{X}]$.
      \begin{align*}
        \proofLine{\deqtype{\Gamma}{A}{\singTm{b}{B}}}{by hypothesis} \\
        \proofLine{\dtype{\Gamma}{B}\rel X \in \perT}{by hypothesis} \tag{*}\label{eq:cong-d-eq-s1t}\\
        \proofLine{\dterm{\Gamma}{B}{t}\rel d\in [X]}{by hypothesis} \tag{\textdagger}\label{eq:cong-d-eq-s2t}\\
        \proofLine{e' = d \in [X]}{by def. of $e=e' \in [\iSing{d}{X}]$} \tag{**}\label{eq:cong-d-eq-s3t}\\
        \proofLine{\dterm{\Gamma}{B}{t}\rel e'\in [X]}{by i.h. on \eqref{eq:cong-d-eq-s1t}, \eqref{eq:cong-d-eq-s2t}, and
        \eqref{eq:cong-d-eq-s3t}, }
      \end{align*}      
  \item $f = f' \in [\iPi{X}{F}]$.
      \begin{align*}
      &\deqtype{\Gamma}{A}{\F{B}{C}} &\\
      &\dtype{\Gamma}{B}\rel X \tag{*}\label{eq:cong-d-eq-f1t}&\\
      &\dterm{\Delta}{\subsTy{C}{\exsubs{\p^i}{s}}}{\appTm{\lift{i}{t}}{s}} \rel f\cdot d \in [F\ d]
      \tag{\textdagger}\label{eq:cong-d-eq-f2t}\\
      & \text{\hspace{3cm} for all }\Delta\leqslant^i\Gamma \text{ and }
      \dterm{\Delta}{\lift{i}{B}}{s} \rel d\in [X] \tag{**}\label{eq:cong-d-eq-f3t} \enspace .
      \end{align*}
      By i.h. on \eqref{eq:cong-d-eq-f1t} and \eqref{eq:cong-d-eq-f3t}
      and monotonicity \ref{lem:logrelMon}
      \[\dterm{\Delta}{\lift{i}{A'}}{s} \rel d' \in [X'] \enspace .\]
      By i.h. on \eqref{eq:cong-d-eq-f2t}
      \[\dterm{\Delta}{\subsTy{B}{\exsubs{\p^i}{s}}}
      {\appTm{(\lift{i}{t})}{s}}\rel f'\cdot d' \in [F\ d'] \enspace .\]
  \item $d = d' \in [\iDs{X}{F}]$.
    \begin{align*}
      \proofLine{d = d' \in [\iDs{X}{F}]}{hypothesis}\tag{*}\label{eq:per-sum-0}\\
      \proofLine{\dterm{\Gamma}{A}{t} \rel d\in [\iDs{X}{F}]}{hypothesis}\tag{**}\label{eq:per-sum-1} \\
      \proofLine{\deqtype{\Gamma}{A}{\DSum{A'}{B}}}{by inversion on \eqref{eq:per-sum-1}}\\
      \proofLine{\dtype{\Gamma}{\DSum{A'}{B}}\rel \iDs{X}{F}\in \perT}{by inversion on \eqref{eq:per-sum-1}}\\
      \proofLine{\dterm{\Gamma}{A'}{\dfst{t}} \rel \fstnew{d}\in [X]}{by inversion on \eqref{eq:per-sum-0}}
        \tag{\textdagger}\label{eq:per-sum-3}\\
      \proofLine{\dterm{\ctxe{\Gamma}{A'}}{\subsTy{B}{\subid{}{\dfst{t}}}}{\dsnd{t}}
        \rel \sndnew{d}\in [F\ \fstnew{d}]}{by inversion on \eqref{eq:per-sum-1}}
        \tag{\textdaggerdbl}\label{eq:per-sum-4}\\
      \proofLine{\fstnew{d}=\fstnew{d'}\in [X]}{by definition of \eqref{eq:per-sum-0}}\tag{\textdbldagger}\label{eq:per-sum-5}\\
      \proofLine{\sndnew{d}=\sndnew{d'}\in [F\ \fstnew{d}]}{by definition of \eqref{eq:per-sum-0}}
        \tag{\textdbldaggerdbl}\label{eq:per-sum-6}\\
      \proofLine{\dterm{\Gamma}{A'}{\dfst{t}} \rel \fstnew{d'}\in [X]}{by ind. hyp. on \eqref{eq:per-sum-3} and \eqref{eq:per-sum-5}}\\
      \proofLine{\dterm{\ctxe{\Gamma}{A'}}{\subsTy{B}{\subid{}{\dfst{t}}}}{\dsnd{t}}
        \rel \sndnew{d'}\in [F\ \fstnew{d'}]}{by ind. hyp. on \eqref{eq:per-sum-4} and \eqref{eq:per-sum-6}}.\rlap{\hbox to 29pt{\hfill\qEd}}
    \end{align*}
  \end{enumerate}
\end{enumerate}
}

\DisplayProof{judgeq}{Lemma}{lem:judgeq}{By induction on $X\in \perT$.
  By induction on $X\in \perT$. For a better organisation of the proof
  we show the proofs for each point separately.
  \begin{enumerate}[(a)] 
  \item $\deqtype{\Gamma}{A}{\reifyC{\Gamma}{\Da{X}}}$. We skip the
    part for the minimal elements in $\perT$.
    \begin{enumerate}[(1)] 
    \item $\iSing{d}{X}$:
      \begin{align*}
        \proofLine{\deqtype{\Gamma}{A'}{\reifyC{\Gamma}{\Da{X}}}}{by ind. hyp. }\\
        \proofLine{\deqterm{\Gamma}{A'}{t}{\reifyC{\Gamma}{\da{X}{d}}}}{by ind. hyp. }\\
        \proofLine{\deqtype{\Gamma}{\singTm{a}{A'}}
          {\singTm{\reifyC{\Gamma}{\da{X}{d}}}{\reifyC{\Gamma}{\Da{X}}}}}{by congruence and
          transitivity}
      \end{align*}
    \item $\iPi{X}{F}$:
      \begin{align*}
        \proofLine{\deqtype{\Gamma}{A'}{\reifyC{\Gamma}{\Da{X}}}}{by ind. hyp.} \\
        \proofLine{\deqtype{\Delta}{\subsTy{B}{\exsubs{\p^i}{s}}}{\reifyC{\Delta}{\Da{F\
                d}}}}{}\tag{*}\label{eq:der-f-eq-1}\\
        &\text{\hspace{2cm} for any } \Delta\leqslant^i \Gamma \text{ and }
        \dterm{\Delta}{\lift{i}{A'}}{s} \rel d \in [X]&&\\
        \proofLine{\dterm{\ctxe{\Gamma}{A'}}{\subsTy{A'}{\p}}{\q}\rel
          \upa{X}{\iVar{|\Gamma|}}}{by ind. hyp.}\tag{\textdagger}\label{eq:der-f-eq-2}\\
        \proofLine{\deqtype{\ctxe{\Gamma}{A'}}{\subsTy{B}{\exsubs{\p}{\q}}}
          {\reifyC{\ctxe{\Gamma}{A'}}{\Da{F\ \upa{X}{\iVar{|\Gamma|}}}}}}
          {by instantiating \eqref{eq:der-f-eq-1} with \eqref{eq:der-f-eq-2}}\\
        \proofLine{\deqtype{\ctxe{\Gamma}{A'}}{B}
          {\reifyC{\ctxe{\Gamma}{A'}}{\Da{F\ \upa{X}{\iVar{|\Gamma|}}}}}}{by \ref{lem:logrelEqTy}}.
      \end{align*}
    \end{enumerate}
  \item $\deqterm{\Gamma}{A}{t}{\reifyC{\Gamma}{\da{X}{d}}}$. We skip
    the part for the minimal elements in $\perT$.
    \begin{enumerate}[(1)]
    \item $d' \in [\iSing{d}{X}]$: 
      \begin{align*}
        \proofLine{\deqtype{\Gamma}{A}{\singTm{b}{B}}}{} \tag{*}\label{eq:der-eq-s1}\\
        \proofLine{\dtype{\Gamma}{B}\rel X \in \perT}{}\\
        \proofLine{\dterm{\Gamma}{B}{t}\rel d\in [X]}{} \tag{\textdagger}\label{eq:der-eq-s2} \\
        \proofLine{\deqterm{\Gamma}{B}{t}{\reifyC{\Gamma}{\da{X}{d}}}}{by ind. hyp.  in \eqref{eq:der-eq-s2}}\\
        \proofLine{\deqterm{\Gamma}{\singTm{t}{B}}{t}{\reifyC{\Gamma}{\da{X}{d}}}}{by conversion and \ruleref{sing-eq-i}}\\
        \proofLine{\deqterm{\Gamma}{A}{t}{\reifyC{\Gamma}{\da{X}{d}}}}{by conversion}
      \end{align*}
      \item $f \in [\iPi{X}{F}]$:
        \begin{align*}
          \proofLine{\dterm{\ctxe{\Gamma}{A'}}{\subsTy{A'}{\p}}{\q}\rel
            \upa{X}{\iVar{|\Gamma|}} \in [X]}{by ind. hyp.  on the third part}\\
          \proofLine{d := \upa{X}{\iVar{|\Gamma|}}}{abbreviation}\\
          \proofLine{\dterm{\ctxe{\Gamma}{A'}}{\subsTy{B}{\exsubs{\p}{\q}}}
            {\appTm{(\subsTm{t}{\p})}{\q}}\rel f\cdot d\in [F\ d]}
          {by definition of the logical relation}\\
          \proofLine{\deqterm{\ctxe{\Gamma}{A'}}{\subsTy{B}{\exsubs{\p}{\q}}}
            {\appTm{(\subsTm{t}{\p})}{\q}}{\reifyC{\ctxe{\Gamma}{A}}{\da{F\
                  d}{f\cdot d}}}}{by ind. hyp. }\\
          \proofLine{\deqterm{\ctxe{\Gamma}{A'}}{B}
            {\appTm{(\subsTm{t}{\p})}{\q}}{\reifyC{\ctxe{\Gamma}{A}}{\da{F\
                  d}{f\cdot
                  d}}}}{by \ruleref{conv}} \\
          \proofLine{\deqterm{\Gamma}{\F{A'}{B}} {\lambda
              (\appTm{(\subsTm{t}{\p})}{\q})}{\lambda
              (\reifyC{\ctxe{\Gamma}{A}}{\da{F\ d}{f\cdot
                  d}})}}{by congruence} \\
          \proofLine{\deqterm{\Gamma}{\F{A'}{B}} {t}{\lambda
              (\appTm{(\subsTm{t}{\p})}{\q})}}{by \ruleref{eta}}\\
            \proofLine{\deqterm{\Gamma}{\F{A'}{B}}
              {t}{\reifyC{\Gamma}{\da{\iPi{X}{F}}{f}}}}{by trans.}
        \end{align*}
    \end{enumerate}
  \item\hfill
    \begin{enumerate}[(1)]
    \item $\iSing{d}{X}$:
      \begin{align*}
        \proofLine{\deqtype{\Gamma}{A}{\singTm{a}{B}}}{by hypothesis} \\
        \proofLine{\dtype{\Gamma}{B}\rel X\in \perT}{by hypothesis} \\
        \proofLine{\dterm{\Gamma}{B}{t} \rel d\in [X]}{by hypothesis}\\
        \proofLine{\dtype{\Gamma}{B}\rel X\in \perT}{by monotonicity \ref{lem:logrelMon}} \\
        \proofLine{\dterm{\Delta}{\lift{i}{B}}{\lift{i}{t}} \rel d\in [X]}{by monotonicity \ref{lem:logrelMon}}\\
        \proofLine{\deqtype{\Delta}{\lift{i}{A}}{\singTm{\lift{i}{a}}{\lift{i}{B}}}}{by congruence} 
      \end{align*}
    \item $\iPi{X}{F}$: 
      \begin{align*}
        \proofLine{\dterm{\Delta}{\lift{i}{A'}}{s}\rel d'\in [X]}{hypothesis} \tag{*}\label{eq:der-eq-f-1}\\
        \proofLine{\deqterm{\Delta}{\lift{i}{A'}}{s}{\reifyC{\Delta}{\da{X}{d'}}}}{by ind. hyp. on \eqref{eq:der-eq-f-1}}\\
        \proofLine{\deqterm{\Delta}{\subsTy{B}{\exsubs{\p^i}{s}}}{\appTm{(\lift{i}{t})}{s}}
          {\appTm{(\lift{i}{(\reifyC{\Gamma}{d})})}{(\reifyC{\Delta}{(\da{X}{d'}})})}}{by congruence} \\
        \proofLine{\reifyC{\Delta} {\iNe{d}{\da{X}{d'}}}=\appTm{(\lift{i}{(\reifyC{\Gamma}{d})})}{(\reifyC{\Delta}{(\da{X}{d'}})})}
        {by definition}\\
        \proofLine{\dterm{\Delta}{\subsTy{B}{\exsubs{\p^i}{s}}}{\appTm{(\lift{i}{t})}{s}}
          \rel \upa{F\ d'}{\iNe{d}{\da{X}{d'}}}\in [F\ d']}{by ind. hyp.}\rlap{\hbox to 56pt{\hfill\qEd}}
      \end{align*}
    \end{enumerate}
  \end{enumerate}
}

\DisplayProof{logrelIdSub}{Lemma}{lem:logrelIdSub}{By induction on
  $\dctx{\Gamma}$; we show only the inductive case.  Let
  $\dctx{\ctxe{\Gamma}{A}}$.
      \begin{align*}
        \proofLine{d := \rho_{\Gamma}}{definition}\\
        \proofLine{\dsubs{\Gamma}{\Gamma}{\idsubs{}}\rel
          d\in \semctx{\Gamma}}{by inversion and i.h. }
        \tag{*}\label{eq:ex-idsubs-logrel-1}\\
        \proofLine{\dsubs{\ctxe{\Gamma}{A}}{\Gamma}{\subsc{\idsubs{}}{\p}}
        \rel d\in \semctx{\Gamma}}
        {from \eqref{eq:ex-idsubs-logrel-1} by Rem. \ref{lem:monsubs}}
        \tag{\textdagger}\label{eq:ex-idsubs-logrel-2}\\
        \proofLine{\dsubs{\ctxe{\Gamma}{A}}{\Gamma}{\p}
        \rel d\in \semctx{\Gamma}}
        {from \eqref{eq:ex-idsubs-logrel-2} by Rem. \ref{lem:logrelEqSub}}
        \tag{**}\label{eq:idsubs-logrel-3}\\
        \proofLine{\dtype{\ctxe{\Gamma}{A}}{\subsTy{A}{\p}}
          \rel \semc{A}{d} \in \perT}{by inversion and Thm. \ref{thm:logrel}}\\
        \proofLine{\dterm{\ctxe{\Gamma}{A}}{\subsTy{A}{\p}}{\q}
          \rel \upa{\semc{A}{d}}{\iVar{n}} \in [\semc{A}{d}]}{by Thm. \ref{thm:logrel}}\\
        \proofLine{\dsubs{\ctxe{\Gamma}{A}}{\ctxe{\Gamma}{A}}{\exsubs{\p}{\q}}\rel\\
         &\quad\quad (d,\upa{\semc{A}{d}}{\iVar{n}})
            \in \sigD{\semctx{\Gamma}}{(e\mapsto [\semc{A}{e}])}}
        {by Def. \ref{logrelsubs}}
        \tag{\textdaggerdbl}\label{eq:ex-idsubs-logrel-4}\\
        \proofLine{\dsubs{\ctxe{\Gamma}{A}}{\ctxe{\Gamma}{A}}{\idsubs{}}\rel
          \rho_{\ctxe{\Gamma}{A}} \in \semctx{\ctxe{\Gamma}{A}}}{from 
           \eqref{eq:ex-idsubs-logrel-4} by Rem. \ref{lem:logrelEqSub}}\rlap{\hbox to 63pt{\hfill\qEd}}
      \end{align*}
}

\DisplayProof{logrel}{Theorem}{thm:logrel}{
 We note that for
  terms we show only the cases when the last rule used was the
  introductory rule, or the rule for introducing elements in
  singletons; for the case of the conversion rule,
  we can conclude by i.h., and lemma \ref{lem:logrelEqTy}.
  \begin{enumerate}[(a)]
  \item Types. We show only the case for \ruleref{fun-f}.
      \begin{align*}
        \proofLine{\dterm{\Delta}{\lift{i}{A'}}{s}\rel e\in [X]}{hypothesis } \tag{*}\label{eq:fund-f-1-t}\\
        \proofLine{\dsubs{\Theta}{\Gamma}{\subsc{\delta}{\p^i}} \rel d \in
          \semctx{\Gamma} }{By monotonicity for substitutions \ref{lem:monsubs}} \tag{\textdagger}\label{eq:fund-f-2-t}\\
      \proofLine{\dsubs{\Theta}{\ctxe{\Gamma}{A}}{\exsubs{\subsc{\delta}{\p^i}}{s}}
        \rel \iPair{d}{e} \in \semctx{\ctxe{\Gamma}{A}}}{From \eqref{eq:fund-f-1-t} and \eqref{eq:fund-f-2-t} } \\
      \proofLine{\dtype{\Theta}{\subsTy{B}{\exsubs{\subsc{\delta}{p^i}}{s}}} \rel
        \semc{B}{\iPair{d}{e}} \in \perT}{by ind. hyp. on $\dtype{\ctxe{\Gamma}{A}}{B}$ and using
         \ref{lem:logrelEqTy} and \ref{lem:logrelEqD}}
      \end{align*}
    \item Terms. We show the case for application
      \ruleref{fun-el} and for \ruleref{\enumrule{e}}. The case for abstraction \ruleref{fun-i} is
      analogous to \ruleref{fun-f}.\hfill
    \begin{enumerate}[(1)]
    \item \ruleref{fun-el}
      \begin{align*}
        \proofLine{\dterm{\Gamma}{\subsTy{B}{\subid{}{r}}}{\appTm{t}{r}}}{hypothesis} \\
        \proofLine{\dterm{\Delta}{\subsTy{A}{\delta}}{\subsTm{r}{\delta}}\rel
          \semc{r}{d} \in [\semc{A}{d}]}{by ind. hyp. }\tag{*}\label{eq-fund-app-1}\\
        \proofLine{\dterm{\Delta}{\subsTy{\F{A}{B}}{\delta}}
          {\subsTm{t}{\delta}}
          \rel \semc{t}{d} \in [\semc{\F{A}{B}}{d}]}{by ind. hyp.  }\tag{\textdagger}\label{eq-fund-app-2}\\
        \proofLine{
          \dterm{\Delta}{\subsTy{B}{\subid{}{\subsTm{r}{\delta}}}}
          {\appTm{(\subsTm{t}{\delta})}{(\subsTm{r}{\delta}})} \rel}{}\\
          \proofLine{\qquad
          \semc{t}{d} \cdot \semc{r}{d} \in
          [\semc{B}{\iPair{d}{\semc{r}{d}}}]}{by def. of log. rel. for
          \eqref{eq-fund-app-2} with \eqref{eq-fund-app-1}}\\
        \proofLine{\dterm{\Delta}{\subsTy{B}{\subid{}{\subsTm{r}{\delta}}}}
          {\subsTm{(\appTm{t}{r})}{\delta}} \rel}{}\\
        \proofLine{\qquad\semc{\appTm{t}{r}}{d} \in
          [\semc{B}{\iPair{d}{\semc{r}{d}}}]}{by \ref{lem:logrelEqTy} and \ref{lem:logrelEqD}}
      \end{align*}
    \item \ruleref{\enumrule{e}}
      \begin{align*}
        \proofLine{\dterm{\Gamma}{\subsTy{B}{\subid{}{t}}}{\elim{B}{t_0\cdots t_{n-1}}{t}}}{hypothesis} \\
        \proofLine{\dterm{\Delta}{\enum{n}}{\subsTm{t}{\delta}}\rel \semc{t}d\in [\enum{n}]}{by inversion and by ind. hyp.}\\
        \proofLine{\deqterm{\Delta}{\enum{n}}{\subsTm{t}{\delta}}{\reifyC{\Delta}{\semc{t}d}}}{by \ref{lem:judgeq} }\\
        \proofLine{\dterm{\Delta}{\subsTy{B}{\exsubs{\delta}{\constn{i}}}}
           {\subsTm{t_i}{\delta}}\rel \semc{t_i}d\in [\semc{B}\depair{d}{\semc{t}d}]}{by inversion and by ind. hyp.}\\
        &\text{if } \reifyC{\Delta}{\semc{t}d}\equiv \constn{i}\text{:}&&\\
        \proofLine{\quad \deqterm{\Delta}{\subsTy{B}{\subid{}{t}}}
          {\subsTm{(\elimn{B}{t_0\cdots t_{n-1}}{\constn{i}})}{\delta}}{\subsTm{t_i}{\delta}}}{by subst.} \\
        \proofLine{\quad \dterm{\Delta}{\subsTy{B}{\subid{}{t}}}
          {\subsTm{(\elimn{B}{t_0\cdots t_{n-1}}{\constn{i}})}{\delta}}\rel}{}\\
        \proofLine{\quad\quad\quad\quad
          \semc{\elim{B}{t_0\cdots t_{n-1}}{t}}d\in [\semc{\subsTy{B}{\subid{}{t}}}d]}
        {by \ref{lem:logrelEqTy} and \ref{lem:logrelEqD}}\\
       &\text{if }  \reifyC{\Delta}{\semc{t}d}\in \neterms\text{:}&&\\
        \proofLine{\quad\deqtype{\ctxe{\Delta}{\enum{n}}}{\subsTy{B}{\exsubs{\subsc{\delta}{\p}}{\q}}}
          {\reify{|\Delta|+1}{\Da{\semc{B}\depair{d}{\iVar{|\Delta|}}}}}}{}\\
        \proofLine{\quad\deqterm{\Delta}{\subsTy{B}{\exsubs{\delta}{\constn{i}}}}
           {\subsTm{t_i}{\delta}}{\reifyC{\Delta}{\semc{t_i}d}}}{by \ref{lem:judgeq}}\\
        \proofLine{\quad t'_i:=\reifyC{\Delta}{ \semc{t_i}d}}{abbreviation}\\
        \proofLine{\quad t':= \reifyC{\Delta}{\da{\semc{B}\depair{d}{\constn{i}}}{\semc{t}d}}}{abbreviation}\\
        \proofLine{\quad B' := \reify{|\Delta|+1}{\semc{B}\depair{d}{\iVar{|\Delta|}}}}{abbreviation}\\
        \proofLine{\quad\Delta\vdash
           \subsTm{(\elimn{B}{t_0\cdots t_{n-1}}{t})}{\delta}=}{}\\
         \proofLine{\quad\qquad\elimn{B'}{t'_0\cdots t'_{n-1}}{t'}:\subsTy{B}{\exsubs{\delta}{t}}}{congruence}\\
    \end{align*}\begin{align*}
         \proofLine{\quad
           \dterm{\Delta}{\subsTy{B}{\exsubs{\delta}{t}}}{\subsTm{(\elimn{B}{t_0\cdots t_{n-1}}{t})}{\delta}}
          \rel}{}\\
        \proofLine{\qquad
          \upa{\semc{B}\depair{d}{\semc{t}d}}{
        (\elimn{B'}{t'_0\cdots t'_{n-1}}{t'}\in \semc{B}\depair{d}{\semc{t}d}})}{by \ref{lem:judgeq} and \ref{lem:logrelMon}\qquad\qquad\qquad}\\
        \proofLine{\quad \dterm{\Delta}{\subsTy{B}{\subid{}{t}}}
          {\subsTm{(\elimn{B}{t_0\cdots t_{n-1}}{\constn{i}})}{\delta}}\rel}{}\\
        \proofLine{\qquad\semc{\elim{B}{t_0\cdots t_{n-1}}{t}}d\in [\semc{\subsTy{B}{\subid{}{t}}}d]}
        {by \ref{lem:logrelEqTy} and \ref{lem:logrelEqD}}
    \end{align*}

    \end{enumerate}    
  \item Substitutions. Only the proof for \ruleref{ext-subs} is shown.
      \begin{align*}
        \proofLine{\dsubs{\Gamma}{\exsubs{\gamma}{t}}{\ctxe{\Theta}{A}}}{hypothesis}\\
        \proofLine{\dsubs{\Delta}{\Theta}{\subsc{\gamma}{\delta}}\rel
          \semc{\gamma}{d}\in \semctx{\Theta}}{by ind. hyp.} \tag{*}\label{eq-fund-subs-eq-1}\\
        \proofLine{\dterm{\Delta}{\subsTy{(\subsTy{A}{\gamma})}{\delta}}{\subsTm{t}{\delta}
          }\rel \semc{t}{d} \in [\semc{\subsTy{A}{\gamma}}{d}]}{} \tag{\textdagger}\label{eq-fund-subs-eq-2}\\
        \proofLine{\iPair{\semc{\gamma}{d}}{\semc{t}{d}} \in
          \sigD{\semctx{\Theta}}{(e \mapsto [\semc{\subsTy{A}{\gamma}
            }{e}])}}{from \eqref{eq-fund-subs-eq-1} and \eqref{eq-fund-subs-eq-2}} \\
        \proofLine{      \dsubs{\Delta}{\ctxe{\Theta}{A}}{\subsc{\exsubs{\gamma}{t}}{\delta}}\rel
          \semc{\exsubs{\gamma}{t}}{d} \in \semctx{\ctxe{\Theta}{A}}}{by \ref{lem:logrelEqSub} and \ref{lem:logrelSubEqD} }\rlap{\hbox to 60pt{\hfill\qEd}} 
      \end{align*}
  \end{enumerate}
}


\DisplayProof{corr-tc}{Theorem}{thm:corr-tc}{By simultaneous induction on 
    $\chktype{\Gamma}{V}$, $\chkterm{\Gamma}{V}{v}$, and $\inftype
    \Gamma V k$. 
    \begin{enumerate}[(1)]
    \item Types:
      \begin{enumerate}[$\bullet$]
      \item the case for $\F{V}{W}$ is also obtained directly from the
        derivations we get using the i.h.\ on $\chktype{\Gamma}{V}$,
        and $\chktype{\ctxe{\Gamma}{V}}{W}$; and use them for deriving
        $\dtype{\Gamma}{\F{V}{W}}$
      \item for $\singTm{v}{V}$, we can apply the same reasoning as
        before: by i.h.\ on $\chktype{\Gamma}{V}$, and
        $\chkterm{\Gamma}{\nbe{V}}{v}$ we know that there are, respectively,
        derivations with conclusions $\dtype{\Gamma}{V}$, and
        $\dterm{\Gamma}{V}{v}$; from which we can conclude
        $\dtype{\Gamma}{\singTm{v}{V}} $
      \item here we'll consider the three cases when $V$ is a neutral
        term, because the reasoning is the same. By i.h.\ on
        $\chkterm{\Gamma}{\TmU}{V}$, we have a derivation with
        conclusion $\dterm{\Gamma}{\TmU}{V}$; hence we use
        \ruleref{u-el}.
      \end{enumerate}
    \item Terms:
      \begin{enumerate}[$\bullet$]
      \item let $V = \TmU$, and $v=\F{V'}{W}$. By i.h.
        $\dterm{\Gamma}{\TmU}{V'}$, and
        $\dterm{\ctxe{\Gamma}{V'}}{\TmU}{W}$, and using both
        derivations we can derive $\dterm{\Gamma}{\TmU}{\F{V}{W}}$.
      \item consider $V = \TmU$, and $v=\singTm{v'}{V'}$. by i.h.\ on
        $\chkterm{\Gamma}{\TmU}{V'}$, and
        $\chkterm{\Gamma}{\nbe{V}}{v'}$, we have
        $\dterm{\Gamma}{\TmU}{V}$, and $\dterm{\Gamma}{\nbe{V}}{v'}$,
        and using conversion we derive $\dterm{\Gamma}{V}{v'}$; and
        these are the premises we need to show
        $\dterm{\Gamma}{\TmU}{\singTm{v'}{V}}$.
      \item $ V= \F{V'}{W}$, and $v =\lambda v'$: we have
        $\chkterm{\ctxe{\Gamma}{V'}}{W}{v'}$. From this we can
        conclude by i.h.\ $\dterm{\ctxe{\Gamma}{V'}}{W}{v'}$; and this is
      the key premise for concluding
      $\dterm{\Gamma}{\F{V'}{W}}{\lambda v'}$.
    \item $V=\singTm{w}{W}$: by hypothesis we know
      $\dterm{\Gamma}{W}{w}$, and $\chkterm{\Gamma}{W}{v}$, and
      $\deqterm{\Gamma }{W}{w}{v}$; by the i.h. on the second one we
      get $\dterm{\Gamma}{W}{v}$; then we can conclude using
      \ruleref{sing-i}.
    \item $v = k\in \neterms$, and $V \not\equiv \singTm{w}{W}$: let
      $\inftype{\Gamma}{k}{V'}$, then we distinguish the cases when
      $V'$ is a singleton, and when $V'$ is not a singleton. In the
      latter case, the derivation is obtained directly from the
      correctness of type-inference.  In the first case we use the
      rule \ruleref{sing-el}, with the derivation obtained by i.h.\
      and then we conclude with conversion.
      \end{enumerate}
    \item Inference:
      \begin{enumerate}[$\bullet$]
      \item for $\lift{i}{\q}$, if $i=0$, then we use \ruleref{hyp},
        and conversion; if $i>0$, then we have a derivation with
        conclusion $\dterm{\Gamma}{\subsTy{A_i}{\p}}{\q}$, and clearly
        $\dsubs{\Gamma}{\Gamma.A_i\ldots.A_0}{\p^i}$, hence by
        \ruleref{subs-term}, we have
        $\dterm{\Gamma.A_i\ldots.A_0}{\lift{i+1}{A_i}}{\lift{i}{\q}}$,
        we conclude by correctness of $\norm{\_}$ and by conversion.
      \item by i.h. we have derivations with conclusions
        $\dterm{\Gamma}{V'}{k}$, with $\ertype{V'}=\F{V}{W}$, hence we
        have a derivation $\dterm{\Gamma}{\F{V}{W}}{k}$ (using
        \ruleref{sing-el} if necessary) and $\dterm{\Gamma}{V}{v}$,
        hence by the rule \ruleref{fun-el}, we have
        $\dterm{\Gamma}{\subsTy{W}{\subid{}{v}}}{\appTm{k}{v}}$.
        We conclude by conversion and correctness of $\norm{\_}$.\qed
      \end{enumerate}
    \end{enumerate}
}

\DisplayProof{compl-tc}{Theorem}{thm:compl-tc}{We prove simultaneously
  all the points. The first point is by induction on the structure of
  the type. In the last two points we use well-founded induction on
  the order $\prec$.
  \begin{enumerate}[(1)]
   \item Types:
     \begin{enumerate}[$\bullet$]
     \item $ \dtype{\Gamma}{\F{V'}{W}}$; by inversion we know
       $\dtype{\Gamma}{V'}$, and $\dtype{\ctxe{\Gamma}{V'}}{W}$;
       hence by i.h.\ we have respectively $\chktype{\Gamma}{V'}$,
       and $\chktype{\ctxe{\Gamma}{V'}}{W}$.
     \item $V = \singTm{v}{V'}$: by inversion we have
       $\dtype{\Gamma}{V'}$, and $\dterm{\Gamma}{\nbe{V'}}{v}$, hence by
       i.h.\ we have both $\chktype{\Gamma}{V'}$, and
       $\chkterm{\Gamma}{V'}{v}$.
     \item $\dtype{\Gamma}{k}$, we have to show
       $\chktype{\Gamma}{k}$. By lemma~\ref{lem:invneut}, we know
       $\dterm{\Gamma}{\TmU}{k}$; hence by i.h.\ we have
       $\inftype{\Gamma }{k}{A}$, and $\deqtype{\Gamma}{A}{\TmU}$,
       hence $\chkterm{\Gamma}{\TmU}{k}$.
     \end{enumerate}
   \item Terms: We omit the trivial cases, e.g. $(\TmU,A)$; we
     have re-arranged the order of the cases for the sake of clarity.
     \begin{enumerate}[$\bullet$]
     \item $v = \F{V'}{W}$: 
       \begin{enumerate}[(a)]
       \item either $\deqtype{\Gamma}{A}{\TmU}$,
         $\dterm{\Gamma}{\TmU}{V'}$, and
         $\dterm{\ctxe{\Gamma}{V'}}{\TmU}{W}$; hence, by i.h.\ we know
         both $\chkterm{\Gamma}{\TmU}{V'}$, and
         $\chkterm{\ctxe{\Gamma}{V'}}{\TmU}{W}$; hence we can conclude
         $\chkterm{\Gamma}{\TmU}{\F{V}{W}}$.
       \item Or $\deqtype{\Gamma}{A}{\singTm{a}{A'}}$,
         $\dterm{\Gamma}{A'}{v}$, and $\deqterm{\Gamma}{A'}{v}{a}$,
         hence by i.h. we know
         $\chkterm{\Gamma}{\nbe{B}}{v}$, by conversion we also have
         and transitivity of the equality
         $\deqterm{\Gamma}{\nbe{B}}{\nbe{a}}{v}$, hence
         $\chkterm{\Gamma}{\singTm{\nbe{a}}{\nbe{B}}}{v}$.
       \end{enumerate}
     \item $v = \singTm{v'}{V}$: 
       \begin{enumerate}[(a)]
       \item $\dterm{\Gamma}{\TmU}{V}$, and
         $\dterm{\Gamma}{V}{v'}$. From those derivations we have by
         i.h.\ $\chkterm{\Gamma}{\TmU}{V}$, and
         $\chkterm{\Gamma}{\nbe{V}}{v'}$, respectively; from which 
         we conclude $\chkterm{\Gamma}{\TmU}{\singTm{v'}{V}}$
       \item $\deqtype{\Gamma}{A}{\singTm{a}{A'}}$, with
         $\dterm{\Gamma}{A'}{v}$, and $\deqterm{\Gamma}{A'}{v}{a}$,
         hence by i.h. we know $\chkterm{\Gamma}{\nbe{B}}{v}$. We can
         also derive $\deqterm{\Gamma}{\nbe{B}}{\nbe{a}}{v}$, hence
         $\chkterm{\Gamma}{\singTm{\nbe{a}}{\nbe{B}}}{v}$.
       \end{enumerate}
     \item $v = \lambda v'$
       \begin{enumerate}[(a)]
       \item $\deqtype{\Gamma}{V}{\F{A'}{B}}$, and
         $\dterm{\ctxe{\Gamma}{A'}}{B}{v'}$; from this we can conclude
         $\dterm{\ctxe{\Gamma}{\nbe{A'}}}{B}{v'}$ by ind. hyp. we get
         $\chkterm{\ctxe{\Gamma}{\nbe{A'}}}{\nbe{B}}{v'}$; therefore
         $\chkterm{\Gamma}{\F{\nbe{A'}}{\nbe{B'}}}{\lambda v'}$.
       \item Or $\deqtype{\Gamma}{A}{\singTm{a}{A'}}$,
         $\dterm{\Gamma}{A'}{v}$, and $\deqterm{\Gamma}{A'}{v}{a}$,
         hence by i.h. we know
         $\chkterm{\Gamma}{\nbe{B}}{v}$, by conversion we also have
         and transitivity of the equality
         $\deqterm{\Gamma}{\nbe{B}}{\nbe{a}}{v}$, hence
         $\chkterm{\Gamma}{\singTm{\nbe{a}}{\nbe{B}}}{v}$.
       \end{enumerate}
     \item $v \in \neterms$: then we do case analysis on $\nbe{A}$.
       \begin{enumerate}[(a)]
       \item If $\nbe{A} = \singTm{w}{W}$, then by soundness of
         $\nbe{\_}$, and conversion we have
         $\dterm{\Gamma}{\singTm{w}{W}}{k}$; and by inversion of
         singletons we have $\dterm{\Gamma}{W}{k}$, and also
         $\deqterm{\Gamma}{W}{k}{w} (*)$. Clearly $(k,W)\prec (k,A)$,
         hence we can apply the inductive hypothesis and conclude
         $\chkterm{\Gamma}{W}{k}$; from that and $(*)$, we conclude 
         $\chkterm{\Gamma}{\singTm{w}{W}}{k}$, i.e.,
         $\chkterm{\Gamma}{\nbe{A}}{k}$.
       \item If $V \not\equiv \singTm{w}{W}$, then $\ertype{V}\equiv
         V$. We use the last clause for concluding
         $\chkterm{\Gamma}{\nbe{A}}{k}$; but we need to show that if
         $\inftype{\Gamma}{k}{V'}$, then
         $\deqtype{\Gamma}{\ertype{V}}{\ertype{V'}}$; we show this in
         the next point.
         \end{enumerate}
     \end{enumerate}
   \item Inference: let $\dterm{\Gamma}{A}{k}$,
     $\inftype{\Gamma}{k}{V'}$, and $V = \nbe{A}$. Show
     $\deqtype{\Gamma}{\ertype{V}}{\ertype{V'}}$.
     \begin{enumerate}[$\bullet$]
     \item let us consider first the case when $V = \singTm{w}{W}$; by
       inversion we have derivations $\dterm{\Gamma}{W}{k}$, and
       $\deqterm{\Gamma}{W}{k}{w}$.  Hence by i.h. we know that
       $\deqtype{\Gamma}{\ertype{V'}}{\ertype{W}}$, and
       $\ertype{W}=\ertype{\singTm{w}{W}}$.
     \item Now we consider the case when $V$ is not a singleton, and
       $k = \lift{i}{\q}$; this case is trivial because by inversion
       we know that
       $\deqtype{\Gamma}{V}{\nbe{\lift{i+1}{(\Gamma!i)}}}$.
     \item the last case to consider is $k = \appTm{k'}{v}$ and $V$
       not a singleton. By inversion we know
       $\dterm{\Gamma}{\subsTy{B}{\subid{}{v}}}{\appTm{k}{v}}$, and
       $\dterm{\Gamma}{\F{A}{B}}{k}$, hence
       $\dterm{\Gamma}{\F{\nbe{A}}{\nbe{B}}}{k}$, and
       $\dterm{\Gamma}{A}{v}$, hence $\dterm{\Gamma}{\nbe{A}}{v}$.  By
       i.h.\ we know that if $\inftype{\Gamma}{k}{V'}$, then
       $\ertype{V'}=\F{\nbe{A}}{\nbe{B}}$, and also
       $\chkterm{\Gamma}{\nbe{A}}{v}$. Hence we can conclude
       $\inftype{\Gamma}{\appTm{k}{v}}
       {\nbe{\subsTy{\nbe{B}}{\subid{}{v}}}}$. And
       $\deqtype{\Gamma}{\nbe{\subsTy{\nbe{B}}{\subid{}{v}}}}{
         \nbe{\subsTy{B}{\subid{}{v}}}}$ (by correctness of the
       $\nbe{\_}$ algorithm).\qed
     \end{enumerate}
   \end{enumerate}
}
